\documentclass[pdflatex,sn-mathphys-num,iicol]{sn-jnl}

\usepackage[english]{babel}
\usepackage{mathrsfs}
\usepackage{physics}
\usepackage{amsmath,amssymb}
\usepackage{amsthm}

\newcounter{common}[section]

\theoremstyle{thmstyleone}
\newtheorem{theorem}[common]{Theorem}

\newtheorem{corollary}[common]{Corollary}

\theoremstyle{thmstyletwo}
\newtheorem{remark}[common]{Remark}

\theoremstyle{thmstylethree}

\begin{document}

\title[A quantum theory of gravity]{A quantum theory of gravity}

\author*[1]{\fnm{Antonio} \sur{Pe{\~n}a Pe{\~n}a}}
\email{antoniopenapena52@gmail.com}

\affil*[1]{\orgdiv{Higher Polytechnic School},
  \orgname{Autonomous University of Madrid},
  \orgaddress{\postcode{28049}, \city{Cantoblanco}, \country{Spain}}}

\abstract{We develop a covariant Hamiltonian formalism in which evolution is
measured along the proper-time flow of an observer congruence, replacing the
external coordinate-time derivative with the Lie derivative. Our formalism
keeps the Hamiltonian split between evolution and configuration variables and
expresses it through spacetime tensors projected relative to the chosen
congruence. After quantization, the resulting observer-congruent brackets lead
to relativistic Heisenberg and Schr{\"o}dinger equations whose evolution
parameter is proper time. We then apply this framework to gravity by taking the
projected spatial metric and its conjugate momentum as canonical variables,
obtaining a complete model of canonical quantum gravity. Among the predictions
of this theory we find a resolution to the singularity problems in cosmology,
black holes and black hole thermodynamics.}


\maketitle

\section{Introduction}

As much as textbooks try to convince us otherwise, we still do not have a
completely satisfactory relativistic theory of quantum physics. For example,
Dirac's equation is capable of accurately describing
electrons moving at relativistic speeds and successfully models the spin
of fermions \cite{Dirac1928Electron} but it still measures its evolution
along a coordinate time labelled by $t$. In no part of the standard
formalism of quantum mechanics in the Schrödinger or Heisenberg pictures
do we see the appearance of the proper time
$\tau$ of the observer who performs the measurement.

Covariant formulations of quantum field theory such as the
Tomonaga--Schwinger equation already soften this by replacing equal-time
slices with arbitrary spacelike hypersurfaces
\cite{Tomonaga1946Invariant,Schwinger1948QED}, and
path-integral methods also make spacetime covariance more transparent
\cite{Feynman1948PathIntegral}; but they do not fully solve the time problem
in quantum mechanics. Even parametrized and proper-time formulations
of relativistic dynamics, which promote an invariant evolution parameter to a
primitive role \cite{Stueckelberg1942,HorwitzPiron1973}, still do not make the
proper-time flow of a congruence of observers the primitive notion of quantum
evolution.

This ``problem of time'', as it is often called, is one of the
greatest obstacles to devising a quantum theory of gravity
\cite{Kuchar2011Time,Isham1993Canonical,Anderson2012Problem}. The object
that governs quantum physics is the Hamiltonian and, due to its very
nature, the Hamiltonian formalism asks us to separate evolution from
configuration. In the standard canonical language this means foliating
spacetime into $1+3$ dimensions, where time is separated from space, fields are
placed on a spatial hypersurface, and the Hamiltonian is interpreted as
the generator of evolution from one hypersurface to another. This is
natural in non-relativistic mechanics but it is conceptually uncomfortable
in relativity where coordinates do not have invariant meaning
and where no preferred global time exists.

The problem is aggravated in gravity. In Special Relativity (SR) one may still
choose an inertial frame and use its coordinate time as a practical external
parameter. In General Relativity (GR) this is no longer a choice because the metric itself is
dynamical and the causal structure of spacetime is part of the system. Thus,
the canonical formalism seems to require two ingredients foreign to the
geometric spirit of relativity; a breaking of manifest covariance through the
choice of a time direction, and an external time parameter with respect to
which quantum states evolve. The usual Hamiltonian formalism
therefore works extremely well as a calculational device, but it does not
allow us to model evolution relative to an observer.

The canonical approach to gravity is motivated by this tension. The ADM
formalism rewrites GR as a Hamiltonian system by decomposing the spacetime
metric into spatial geometry and embedding data \cite{Arnowitt1962Dynamics}.
Its quantization leads to the Wheeler--DeWitt equation, in which the quantum
state is a functional of the spatial metric and matter fields
\cite{DeWitt1967Quantum}. However, because the Hamiltonian appears as a
constraint rather than as an ordinary generator of external-time evolution
one obtains a timeless kernel equation $\hat{H}\Psi=0$ rather
than an ``eigenvalue equation'' given by
$i\hbar\partial_t\Psi=\hat{H}\Psi$.

The canonical program did not stop at these beginnings. Its roots lie in
Dirac's generalized Hamiltonian dynamics for constrained systems
\cite{Dirac1950GenHamCJM,Dirac1958GenHam} and his Hamiltonian form of
gravitation \cite{Dirac1958HamiltonianGR}, together with the covariant
constraint analysis of Anderson and Bergmann \cite{AndersonBergmann1951}.
ADM then recast this structure as canonical variables and a definition of
gravitational energy
\cite{Arnowitt1962Dynamics,ADM1959DynamicalStructure,ADM1960Canonical}. Making
the variational problem well posed and giving the energy an invariant meaning
later required the boundary terms of York and of Gibbons--Hawking
\cite{York1972ConformalThreeGeometry,GibbonsHawking1977Action} and the
quasilocal surface charges of Brown and York \cite{BrownYork1993Quasilocal}. A
decisive subsequent step was Ashtekar's reformulation in connection variables
\cite{Ashtekar1986NewVariables,Ashtekar1987NewHamiltonian}, which rendered the
constraints polynomial and seeded the loop program discussed below. The
construction of the present paper draws on this same lineage---constraint
analysis, projected geometry, and boundary energy---while replacing the
coordinate-time foliation with an observer congruence.

Many other approaches have been developed. Perhaps the most straightforward
one is to avoid the Hamiltonian formalism at all and work with
covariant path integrals or sum-over-geometries. The Hartle--Hawking
no-boundary proposal is a famous example of this approach in quantum
cosmology \cite{HartleHawking1983WaveFunction}, and the question of how to
define its contour has remained active, as in the recent Lorentzian and
Picard--Lefschetz analyses \cite{FeldbruggeLehnersTurok2017}. These methods make
covariance more explicit, but come at the price of leaving the measure,
contour, class of geometries, and nonperturbative definition of the path
integral obscure.

Perturbative quantum gravity takes a different route; one expands the metric
around a fixed background and quantizes the fluctuation field, staying within
ordinary quantum field theory. It also allows us to compute scattering amplitudes, but
pure Einstein gravity is not perturbatively renormalizable in four dimensions.
The one-loop and two-loop analyses show that new counterterms are unavoidable
once gravity is treated as an ordinary quantum field theory
\cite{tHooftVeltman1974OneLoop,GoroffSagnotti1985TwoLoop}. Still, as an
effective field theory GR is perfectly meaningful at low energies and
predicts universal quantum corrections to long-distance gravitational
interactions \cite{Donoghue1994EFT}. This is an interesting result, but it is
not a fundamental ultraviolet completion.

Other ultraviolet proposals modify the high-energy behavior of gravity more
radically. Higher-derivative gravity improves power counting and can be made
perturbatively renormalizable, but the standard local quadratic theory
introduces massive ghost degrees of freedom \cite{Stelle1977Renormalization}.
Asymptotic safety instead proposes that quantum gravity may be
non-perturbatively renormalizable because the renormalization-group flow
approaches a non-Gaussian ultraviolet fixed point
\cite{Weinberg1979UV,Reuter1998Evolution}, an idea since developed into a
functional renormalization-group program with mounting evidence for such a
fixed point \cite{ReuterSaueressig2012}. String theory replaces point
particles by extended objects and thereby softens ultraviolet divergences
while also incorporating a massless spin-two excitation naturally interpreted
as the graviton
\cite{GreenSchwarzWitten1987Superstring,Polchinski1998StringTheory}; it has
since matured through D-branes \cite{Polchinski1995Dbranes} and the
holographic gauge/gravity duality \cite{Maldacena1998AdSCFT}. These
approaches are well-defined and highly developed, but they also introduce additional
structures that are not obviously forced by the original Hamiltonian problem.

There are also background-independent and discrete approaches. Loop quantum
gravity quantizes geometry non-perturbatively using connection variables and
spin-network states
\cite{RovelliSmolin1995SpinNetworks,AshtekarLewandowski2004Status,Rovelli1998LoopReview},
and has been applied to cosmology, where it replaces the big-bang singularity
by a bounce \cite{AshtekarSingh2011LQC}.
Spin-foam models provide a covariant transition-amplitude version of the same
general program \cite{Perez2013SpinFoam}. Causal dynamical triangulations
define a non-perturbative sum over geometries by regulating spacetime with
causally well-behaved triangulations
\cite{AmbjornJurkiewiczLoll2004Emergence,AmbjornJurkiewiczLoll2010Quest,Loll2020CDTReview}.
Causal set theory replaces the continuum by a locally finite partially
ordered set whose order relation encodes causality
\cite{Bombelli1987CausalSet,Surya2019CausalSets}. Group field theory, Ho\v{r}ava--Lifshitz
gravity, and noncommutative-geometric models provide still other ways of
modifying the microscopic description of spacetime
\cite{Oriti2006GroupField,Horava2009Lifshitz,Connes1995Noncommutative}.
However, these approaches are hindered by the difficulty of arriving
at testable predictions, as well as justifying their derivation. No universally
accepted formulation among these has yet emerged.

The present work follows a more conservative route. Instead of starting by
changing the microscopic ontology of spacetime we start by asking whether
the Hamiltonian formalism itself can be made fully relativistic in the
sense required by GR. In this formalism the ordinary time derivative
$\partial_t$ is replaced by an observer-congruence evolution derivative
$D^\parallel$ so that evolution is measured along the flow of a unit
timelike vector field $n^\mu$. The parameter of evolution is of course
the proper-time parameter of the observer congruence.
Correspondingly, spatial change is described only after projection onto the
local rest space orthogonal to $n^\mu$. In this way the split between
evolution and configuration is kept---because a Hamiltonian formalism needs
such a split---but the split is made covariant and observer-dependent rather
than absolute.

This places the present work within the long tradition of spacetime-splitting
techniques relative to a unit timelike vector field, developed independently by
the Russian, American, Italian, French, and German schools and later unified in
the gravitoelectromagnetic (GEM) splitting formalism of Jantzen, Carini, and
Bini \cite{JantzenCariniBini1992GEM}. We review this provenance at the beginning
of the next section, from which all our kinematical objects are taken.

There are differences between our formalism and that of ADM, though.
We do not begin by selecting a coordinate time and
decomposing the metric into lapse, shift, and spatial components; instead,
we choose an arbitrary, physical congruence of observers and define evolution
along its flow. Both setups share the same global-hyperbolicity assumptions and
the formalism still uses spacelike hypersurfaces when the congruence is
hypersurface orthogonal, but the fundamental objects here are \textit{spacetime}
tensors ($g_{\mu\nu}$ versus $g_{ij}$), projectors, and Lie derivatives $\pounds_n$
along $n^\mu$. So our construction
keeps the Hamiltonian idea of evolution while taking the evolution parameter to
be the proper time measured along the congruence.

To do so we work in a formalism with elements of the GEM literature and add some original
contributions. The first one is the Hamiltonian
layer built on the splitting kinematics; i.e., (1) a Legendre transformation
taken with respect to the Lie derivative along the congruence, (2) the
observer-congruent Poisson brackets and (3) Hamilton equations that follow from
it. Here we make use of the Lie derivative out of necessity from the selection
argument that we shall see in the next section. Note that the splitting literature proposes many admissible temporal
derivatives and we have to single out only the one that allows a coherent theory
of gravity.

Our second and most important contribution is the canonical formulation
of gravity this layer produces, under which the projected metric and its
momentum are subject only to second-class projection constraints
\textit{so that no Hamiltonian or momentum constraint arises},
the theory evolves unitarily in proper time with a conserved
positive inner product, and the kinetic ordering is fixed uniquely by the
theorem of Appendix \ref{app:ordering}.

Lastly, our third contribution consists of the consequences extracted
from the applications; namely, (1) the resolution of the
cosmological and Schwarzschild singularities through conserved uncertainty
invariants, (2) the quantized lower bounds obtained from the conformal
structure of the boundary theory, (3) the entropy operator conjugate to
the horizon boost angle with its temperature and mass bounds, and the
(4) ultraviolet cutoff these bounds induce in the graviton sector. Direct
comparisons with the corresponding literature are given where each of these
results is derived.

In the next section we construct the covariant Hamiltonian formalism associated
with $\pounds_n$ and the projected configuration derivative $D^\perp_\mu$, derive
the corresponding Hamilton equations, Poisson brackets, and observer-dependent
variational derivatives, and quantize the result. This leads to a relativistic
wavefunction equation in which the generator evolves states with respect to the
proper time of the observer congruence; in the one-particle limit it gives the
square-root relativistic Hamiltonian and recovers the Klein--Gordon, Dirac and
non-relativistic Schrödinger equations as limiting cases. In the following
section we apply the same logic to gravity by
treating the spatially projected metric $\gamma_{\mu\nu}$ and its conjugate
momentum as the canonical variables. This yields an observer-congruent
canonical formulation of Einstein gravity in which all objects remain spacetime
tensors projected relative to the congruence.

The resulting theory should be understood as an attempt to repair the point at
which canonical quantum gravity inherits a non-relativistic structure. If quantum
evolution is to be Hamiltonian, and if Hamiltonian evolution requires a notion
of time, then in relativity that time should be the proper time measured along
a congruence of physical observers. The formalism developed below implements
this idea and studies its consequences for relativistic wave functions and for
the quantization of the gravitational field. We find that it produces many of
the predictions expected of a quantum-gravity theory by developing the subsequent
equations.

Throughout this work, we use the $(+ ---)$ signature for the metric and work in
four spacetime dimensions. Additionally, here we use Greek letters to represent
four spacetime tensors---e.g., $g_{\mu\nu}$---and Latin
letters for spatial, three-dimensional tensors---e.g., $g_{ij}$---.
We shall mostly use Greek letters, which serve as a covariance check; any
equation written here with Greek indices is covariant.
Furthermore, for partial derivatives we might interchangeably use $\partial_n$
with $n$ an integer like $\partial_0 \equiv \partial/\partial x^0$ or
$\partial_\kappa$ with $\kappa$ any variable, like $\kappa = E$ for energy
$\partial_E \equiv \partial/\partial E$. Also, we do not work in natural units
in order to get more explicit and exact calculations.
 \section{Covariant Hamiltonian formalism and the relativistic wave equation}
\label{sec:covariant-hamiltonian-formalism}

Hamiltonian mechanics requires the notion of an ``evolution'' along something. In classical, non-relativistic field theory we use $\partial_{t}$ to measure the evolution of the system. But as we stated before in relativity we have the problem that there is no absolute time we can default to, and especially in GR the time coordinate does not hold any prominent role compared to the other coordinates. We therefore need a generalisation of the procedure from classical mechanics.

The geometric ingredient of this generalisation is the splitting of spacetime relative to a unit timelike vector field, an already known material that we shall not rederive here. It was developed independently by several schools from the 1940s onward. Landau and Lifshitz introduced the threading formalism \cite{LandauLifshitz1975ClassicalFields} and Zelmanov his chronometric invariants \cite{Zelmanov1956Chronometric}; Arnowitt, Deser and Misner (ADM) built the canonical decomposition \cite{Arnowitt1962Dynamics,ADM1959DynamicalStructure}; Cattaneo and Ferrarese developed the projection formalism \cite{Cattaneo1958Projections,FerrareseBini2008Continuum}; Lichnerowicz and Choquet-Bruhat established the initial-value formulations \cite{Lichnerowicz1944,FouresBruhat1952}; and Ehlers worked out the relativistic kinematics of continuous media \cite{Ehlers1993Translation} that was later extended into the $1+3$ covariant approach to cosmology \cite{EllisvanElst1999Cosmological}. These traditions were unified and standardized in the gravitoelectromagnetic (GEM) splitting formalism of Jantzen, Carini and Bini \cite{JantzenCariniBini1992GEM} and a modern textbook account can be found in \cite{Gourgoulhon2012ThreePlusOne}. Every kinematical object appearing below is standard in that literature. This covers the orthogonal projector and the spatial metric together with the projected spatial connection and the acceleration, expansion, vorticity and extrinsic curvature of the congruence.

The first version of this work arrived at these objects \textit{ab initio} without knowledge of these developments, but of course we now adopt the terminology and attributions of that literature throughout and then add what we need for the quantization of gravity. Namely, as we said, the Hamiltonian layer; which consists of (1) a Legendre transformation together with (2) Poisson brackets, and (3) Hamilton equations in which evolution is measured by (4) a Lie derivative along the congruence. As a result the proper time of the observers is taken to be the evolution parameter and all objects keep their four spacetime indices throughout.

From first principles we thus see that the requirements of this generalisation are that it should (1) preserve manifest covariance, (2) have a notion of ``evolution'', (3) be relative to any set of moving observers according to the rules of SR, and (4) reduce to the classical Hamiltonian in the low-velocity limit.
In departure from the GEM literature here we also add one last requirement, namely that (5) the metric should have a non-trivial Legendre transformation. This requirement comes from trying to treat the metric as a field like any other and thus to define a notion of evolution that allows us to treat gravity in a
coherent manner within our generalised Hamiltonian formalism.

Requirement (5) has the effect of selecting the evolution operator, which deserves some care because the splitting literature offers several candidates. Along a congruence of observers with four-velocity $u^\mu$ one can transport tensors with the covariant directional derivative $c^{-1}u^\mu\nabla_\mu$, with the Fermi--Walker derivative or with the Lie derivative. At the purely kinematical level these are equally admissible temporal derivatives and there is no compelling reason to privilege one over the others; the GEM formalism accordingly employs spatially projected versions of these operators, which take spatial fields to spatial fields, the projected Lie derivative being one example \cite{JantzenCariniBini1992GEM}.

Here we shall use the Hamiltonian criterion to break this tie. Note that the directional and Fermi--Walker derivatives \textit{are metric compatible}, so under either of them the velocity of the metric vanishes identically and $g_{\mu\nu}$ admits no non-trivial Legendre transformation. The Lie derivative instead gives

\begin{equation}
\begin{aligned}
    \pounds_n g_{\mu\nu}
    &=
    \nabla_\mu n_\nu + \nabla_\nu n_\mu,
    \\
    n^\mu &\equiv u^\mu / c,
    \qquad
    u^\mu u_\mu = c^2,
\end{aligned}
\end{equation}

\noindent which vanishes if and only if $n^\mu$ is Killing; i.e., the metric stops evolving iff spacetime possesses a timelike symmetry, as any well-defined notion of evolution of $g_{\mu\nu}$ should demand. The Lie derivative moreover keeps velocities transforming as tensors instead of pseudotensors. From these properties we thus find it almost compulsory to define the evolution derivative as

\begin{equation}
    D^\parallel \equiv \pounds_n, \qquad n^\mu n_\mu = 1,
\end{equation}

\noindent where $n^\mu \equiv u^\mu / c$ is a unit vector for the $(+---)$ signature and also dimensionless, which allows us to use it anywhere in the equations and makes $D^\parallel$ satisfy the $\partial_{ct}$ limit for time-independent spacetimes. Unlike the projected temporal derivatives of the GEM formalism $D^\parallel$ acts unprojected on arbitrary spacetime tensors, and this unprojected action has the benefit of leaving the metric with the non-vanishing velocity demanded by requirement (5). The proper time of the observers is encoded in the vector itself through

\begin{equation}
    \frac{d x^\mu}{d\tau} = u^\mu\left(x(\tau)\right).
\end{equation}

\noindent Acting on a general tensor the evolution derivative reads

\begin{equation}\label{def:evolDeriv}
\begin{aligned}
D^{\parallel} T^{\mu_1\cdots\mu_p}{}_{\nu_1\cdots\nu_q}
&=
n^\rho \nabla_\rho
T^{\mu_1\cdots\mu_p}{}_{\nu_1\cdots\nu_q}
\\[1mm]
&\quad
-
\sum_{i=1}^{p}
T^{\mu_1\cdots\rho\cdots\mu_p}{}_{\nu_1\cdots\nu_q}
\nabla_\rho n^{\mu_i}
\\[1mm]
&\quad
+
\sum_{j=1}^{q}
T^{\mu_1\cdots\mu_p}{}_{\nu_1\cdots\rho\cdots\nu_q}
\nabla_{\nu_j} n^\rho,
\end{aligned}
\end{equation}

\noindent so once we make the calculations in our formalism from the perspective of the $n^\mu$
congruence of observers we can obtain the corrected predictions for another $m^\mu$
congruence through the transformation rule

\begin{equation}
    m^\mu = \frac{1}{\sqrt{1 - v_\perp^{\nu}v^\perp_{\nu}}}\left(n^\mu + v_\perp^{\mu}\right),
\end{equation}

\noindent where $n_\mu v_\perp^{\mu} = 0$ and $m^\mu m_\mu = 1$. We thus see that our
formalism does not depend on the congruence and can be expressed in terms of any
set of observers, although of course we generally pick $n^\mu$ to be the simplest and neatest one
in which to express our calculations. Also note that we shall work with
spacelike hypersurface leaves $\Sigma_\tau$ ordered by $\tau$ as in the ADM formalism; and yet, even so we find no need to use
spatial indices, to break manifest covariance or to introduce a $1+3$ dimensional breakdown.
We work instead with spatial projections and Lie derivatives along timelike evolution vectors
and never introduce lapse and shift functions in the metric itself.

For the configuration side we use the standard orthogonal projector and positive spatial metric

\begin{equation}
    {\mathcal{P}^\mu}_\nu \equiv {\delta^\mu}_\nu - n^\mu n_\nu, \qquad \gamma_{\mu\nu} \equiv - \mathcal{P}_{\mu\nu},
\end{equation}

\noindent where because of working with the $(+---)$ convention we require the sign reversal above, thus making $\gamma_{\mu\nu}$ the congruence-covariant form of the threading three-metric $-g_{ik} + g_{0i}g_{0k}/g_{00}$ of Landau--Lifshitz and Zelmanov \cite{LandauLifshitz1975ClassicalFields,Zelmanov1956Chronometric}. We also use the fully projected spatial derivative

\begin{equation}\label{def:configDeriv}
\begin{aligned}
D_{\lambda}^{\perp}T^{\mu_1\cdots\mu_p}{}_{\nu_1\cdots\nu_q}
\equiv\;&
\mathcal{P}_{\lambda}{}^{\rho}
\mathcal{P}^{\mu_1}{}_{\alpha_1}\cdots
\mathcal{P}^{\mu_p}{}_{\alpha_p}
\\
&\times
\mathcal{P}_{\nu_1}{}^{\beta_1}\cdots
\mathcal{P}_{\nu_q}{}^{\beta_q}
\\
&\times
\nabla_{\rho}
T^{\alpha_1\cdots\alpha_p}{}_{\beta_1\cdots\beta_q},
\end{aligned}
\end{equation}

\noindent which is the familiar spatial covariant derivative of the splitting literature
\cite{JantzenCariniBini1992GEM,Cattaneo1958Projections,Gourgoulhon2012ThreePlusOne}.
Only when every index is projected does this operator become the Levi--Civita connection of
$\gamma_{\mu\nu}$ with $D^\perp_\lambda \gamma_{\mu\nu} = 0$; projecting the derivative index
alone leaves $\mathcal{P}^\rho{}_\lambda\nabla_\rho\gamma_{\mu\nu} \neq 0$ in general. We are
going to call $D^\parallel$ the ``evolution'' derivative and $D^\perp_\mu$ the
``configuration'' derivative because each corresponds to the evolution and configuration
variables in the phase space of the Hamiltonian.

The full projection also matters physically for our purposes.
$D^\perp_\lambda$ is supposed to be a derivative within the observer's
configuration space and should take spatial tensors to spatial tensors. If it were allowed to
develop normal pieces the Hamiltonian would no longer live purely on the reduced
observer-spatial phase space defined by the constraints $n^\mu\gamma_{\mu\nu} = 0$ and
$n^\mu\pi_{\mu\nu} = 0$ that we shall encounter for gravity.

Additionally, for this formalism to work well we shall also require the congruence to be hypersurface
orthogonal. By Frobenius's theorem this is equivalent to the vanishing of its vorticity,
that is to pick a congruence such that
$n_{[\mu}\nabla_\nu n_{\rho]} = 0 \Leftrightarrow \omega_{\mu\nu}=0$
\cite{EllisvanElst1999Cosmological}. We need it because only then do the local rest spaces
orthogonal to $n^\mu$ integrate to the spacelike leaves $\Sigma_\tau$ used
throughout, and only then are the ``equal-$\tau$'' surfaces and the configuration
derivative $D^\perp_\mu$ globally well defined. Without this condition the physics does not
change (the congruence remains a mere choice of observers) but there are then no surfaces
on which to define quantum states, hence no Heisenberg or Schrödinger equations in their
standard form relative to these observers; one would have to rebuild the quantum theory on
a separate hypersurface-orthogonal congruence and transform back. This is the same
global-hyperbolicity setting assumed by the ADM decomposition but different from the threading
and $1+3$ formalisms, because these allow rotating congruences with
$\omega_{\mu\nu}\neq 0$ \cite{JantzenCariniBini1992GEM,EllisvanElst1999Cosmological}.
Here for the aforementioned reasons we shall work throughout with $\omega_{\mu\nu}=0$.

If we are going to work with these two derivatives we shall need the decomposition of $\nabla_\alpha$ in terms of $D^\parallel$ and $D^\perp_\alpha$ because most Lagrangians in curved spacetime are written using $\nabla_\alpha$ on the field configuration. Making the appropriate replacements one finds

\begin{equation}
\begin{aligned}
\nabla_{\lambda} S^{\mu_1 \cdots \mu_p}{}_{\nu_1 \cdots \nu_q}
={}&
n_{\lambda}\Bigg[
D^{\parallel}
S^{\mu_1 \cdots \mu_p}{}_{\nu_1 \cdots \nu_q}
\\
&\quad
+ \sum_{i=1}^{p}
\left(D^{\perp}_{\rho} n^{\mu_i}\right)
\\
&\qquad\quad\times
S^{\mu_1 \cdots \rho \cdots \mu_p}
{}_{\nu_1 \cdots \nu_q}
\\
&\quad
- \sum_{j=1}^{q}
\left(
D^{\perp}_{\nu_j} n^{\rho}
+ n_{\nu_j}D^{\parallel} n^{\rho}
\right)
\\
&\qquad\quad\times
S^{\mu_1 \cdots \mu_p}
{}_{\nu_1 \cdots \rho \cdots \nu_q}
\Bigg] \\
&+
D^{\perp}_{\lambda}
S^{\mu_1 \cdots \mu_p}{}_{\nu_1 \cdots \nu_q}
\\
&- \sum_{i=1}^{p}
n^{\mu_i}
\left(D^{\perp}_{\lambda} n_{\rho}\right)
\\
&\qquad\quad\times
S^{\mu_1 \cdots \rho \cdots \mu_p}
{}_{\nu_1 \cdots \nu_q}
\\
&- \sum_{j=1}^{q}
n_{\nu_j}
\left(D^{\perp}_{\lambda} n^{\rho}\right)
\\
&\qquad\quad\times
S^{\mu_1 \cdots \mu_p}
{}_{\nu_1 \cdots \rho \cdots \nu_q}.
\end{aligned}
\end{equation}

\noindent And given a Lagrangian of the form

\begin{equation}
    L\left(T^A, D^\parallel T^A, D^\perp_\lambda T^A\right), \qquad T^A \equiv T^{\mu_1\dots\mu_p}{}_{\nu_1\dots\nu_q},
\end{equation}

\noindent then holding $n^\mu$-related variables fixed we obtain

\begin{equation}
\begin{aligned}
    \frac{\partial L}{\partial\left(\nabla_\mu T^B\right)}
    ={}&
    n^\mu\Pi^A{}_B
    \frac{\partial L}{\partial\left(D^\parallel T^A\right)}
    \\
    &+
    \mathcal{P}_\lambda{}^\mu\Pi^A{}_B
    \frac{\partial L}
    {\partial\left(D^\perp_\lambda T^A\right)},
\end{aligned}
\end{equation}

\noindent with

\begin{equation}
    \Pi^A{}_B \equiv \mathcal{P}^{\mu_1}{}_{\alpha_1}\cdots\mathcal{P}^{\mu_p}{}_{\alpha_p}\mathcal{P}_{\nu_1}{}^{\beta_1}\cdots\mathcal{P}_{\nu_q}{}^{\beta_q}.
\end{equation}

\noindent Now if we define the observer-congruent Hamiltonian as

\begin{equation}
    H_\parallel[\Sigma_\tau]
    \equiv
    \int_{\Sigma_\tau} d\Sigma \, \mathcal{H}_\parallel ,
\end{equation}

\noindent with $\mathcal{H}_\parallel$ the Hamiltonian density coming from $\mathcal{L}$, then we can derive $H_\parallel$ for $\mathfrak{n}$ fields in the usual way. Using the momentum definition $\pi_A{}^\mathfrak{n} \equiv \partial L / \partial \left(D^\parallel \Phi^A{}_\mathfrak{n}\right)$, the differential of the Lagrangian can be written as

\begin{equation}
    \begin{split}
        dL
        &=
        \frac{\partial L}{\partial \Phi^A{}_\mathfrak{n}}
        d\Phi^A{}_\mathfrak{n}
        \\
        &\quad+
        d\left(
        \pi_A{}^\mathfrak{n}
        D^\parallel \Phi^A{}_\mathfrak{n}
        \right)
        \\
        &\quad-
        \left(D^\parallel \Phi^A{}_\mathfrak{n}\right)
        d\pi_A{}^\mathfrak{n}
        \\
        &\quad+
        \frac{\partial L}
        {\partial \left(D^\perp_\lambda
        \Phi^A{}_\mathfrak{n}\right)}
        d\left(D^\perp_\lambda
        \Phi^A{}_\mathfrak{n}\right)
        \\
        &\quad+
        \frac{\partial L}{\partial x^\mu}
        \frac{dx^\mu}{dc\tau},
    \end{split}
\end{equation}

\noindent and rearranging terms we obtain

\begin{equation}
    \begin{split}
        dH_\parallel &\equiv d\left(\pi_A{}^\mathfrak{n}D^\parallel \Phi^A{}_\mathfrak{n}\right) - dL \\
        &= \left(D^\parallel \Phi^A{}_\mathfrak{n}\right)d\pi_A{}^\mathfrak{n} - \frac{\partial L}{\partial \Phi^A{}_\mathfrak{n}}d\Phi^A{}_\mathfrak{n} \\
        &- \frac{\partial L}{\partial \left(D^\perp_\lambda \Phi^A{}_\mathfrak{n}\right)}d\left(D^\perp_\lambda \Phi^A{}_\mathfrak{n}\right) - \frac{\partial L}{\partial x^\mu}\frac{dx^\mu}{dc\tau}.
    \end{split}
\end{equation}

\noindent Expanding $dH_\parallel$ in its arguments $\left(\pi_A{}^\mathfrak{n}, \Phi^A{}_\mathfrak{n}, D^\perp_\lambda\Phi^A{}_\mathfrak{n}\right)$ and equating coefficients we obtain

\begin{equation}
    \frac{\partial H_\parallel}{\partial \left(D^\perp_\lambda \Phi^A{}_\mathfrak{n}\right)} = -\frac{\partial L}{\partial \left(D^\perp_\lambda \Phi^A{}_\mathfrak{n}\right)}, \qquad \frac{dH_\parallel}{dc\tau} = -\frac{dL}{dc\tau},
\end{equation}

\noindent and

\begin{equation}\label{eq:eomHamil}
\boxed{
\begin{aligned}
    D^\parallel \Phi^A{}_\mathfrak{n}
    &=
    \frac{\delta^\perp H_\parallel}
    {\delta \pi_A{}^\mathfrak{n}},
    \\
    D^\parallel \pi_A{}^\mathfrak{n}
    + \theta_n \pi_A{}^\mathfrak{n}
    &=
    -\frac{\delta^\perp H_\parallel}
    {\delta \Phi^A{}_\mathfrak{n}}
\end{aligned}
},
\end{equation}

\noindent where $\theta_n \equiv \nabla_\alpha n^\alpha$, and the observer-dependent variational derivative is defined as

\begin{equation}
    \frac{\delta^\perp H_\parallel}
    {\delta \Phi^A{}_\mathfrak{n}}
    \equiv
    \frac{\partial H_\parallel}
    {\partial \Phi^A{}_\mathfrak{n}}
    -
    \left(D^\perp_\lambda - a_\lambda\right)
    \frac{\partial H_\parallel}
    {\partial\left(D^\perp_\lambda \Phi^A{}_\mathfrak{n}\right)},
\end{equation}

\noindent with the acceleration $a_\lambda \equiv n^\alpha\nabla_\alpha n_\lambda$. Note that we can use $\delta^\perp$ for the $D^\parallel \Phi^A{}_\mathfrak{n}$ equation because $H_\parallel$ does not depend on $D^\perp_\lambda \pi_A{}^\mathfrak{n}$, so we are free to add that term. This shall prove valuable for the Poisson brackets. The momentum density expression at Eq. (\ref{eq:eomHamil}) was obtained by equating coefficients and expressing the Euler-Lagrange equation in evolution-configuration form

\begin{equation}
\begin{aligned}
    \frac{\partial L}
    {\partial \Phi^A{}_\mathfrak{n}}
    ={}&
    D^\parallel \pi_A{}^\mathfrak{n}
    +
    \theta_n \pi_A{}^\mathfrak{n}
    \\
    &+
    \left(D^\perp_\lambda
    -
    a_\lambda \right)
    \frac{\partial L}
    {\partial\left(D^\perp_\lambda
    \Phi^A{}_\mathfrak{n}\right)} .
\end{aligned}
\end{equation}
We now rewrite Eq. \eqref{eq:eomHamil} in Poisson-bracket form. The important point is that, since the Hamiltonian density depends on the observer-spatial derivative $D^\perp_\lambda\Phi^A{}_{\mathfrak n}$, the functional derivatives entering the bracket must be the observer-orthogonal variational derivatives, obtained by integrating by parts only along the hypersurface directions through the identity
\begin{equation}
\begin{aligned}
    \int_{\Sigma_\tau} d\Sigma \,
    V^\lambda D^\perp_\lambda X
    ={}&
    \int_{\partial\Sigma_\tau} dS \,
    s_\lambda V^\lambda X
    \\
    &-
    \int_{\Sigma_\tau} d\Sigma \,
    \left[
    \left(D^\perp_\lambda-a_\lambda\right)V^\lambda
    \right]X.
\end{aligned}
\end{equation}

\noindent Hence, when the boundary term is fixed or vanishes, the variation of any phase-space functional
$F[\Phi,\pi;\tau] = \int_{\Sigma_\tau} d\Sigma\,\mathcal F\left(\Phi^A{}_{\mathfrak n}, \pi_A{}^{\mathfrak n}, D^\perp_\lambda\Phi^A{}_{\mathfrak n}, D^\perp_\lambda\pi_A{}^{\mathfrak n}; \tau\right)$
is expressed through $\delta^\perp$ derivatives alone, and the equal-$\tau$ Poisson bracket is defined by

\begin{equation}\label{eq:observerPoissonBracket}
\left\{F,G\right\}_{\Sigma_\tau}
=
\int_{\Sigma_\tau} d\Sigma
\left[
\frac{\delta^\perp F}
{\delta \Phi^A{}_{\mathfrak n}}
\frac{\delta^\perp G}
{\delta \pi_A{}^{\mathfrak n}}
-
\frac{\delta^\perp F}
{\delta \pi_A{}^{\mathfrak n}}
\frac{\delta^\perp G}
{\delta \Phi^A{}_{\mathfrak n}}
\right].
\end{equation}

\noindent Equivalently, in local form,

\begin{equation}
    \left\{
    \Phi^A{}_{\mathfrak n}(x),
    \pi_B{}^{\mathfrak m}(y)
    \right\}_{\Sigma_\tau}
    =
    \delta^A{}_B
    \delta_{\mathfrak n}{}^{\mathfrak m}
    \delta_{\Sigma_\tau}(x,y),
\end{equation}

\noindent with all other fundamental brackets vanishing and $\delta_{\Sigma_\tau}(x,y)$ the hypersurface delta distribution. Applying Eq. \eqref{eq:observerPoissonBracket} to the Hamiltonian gives

\begin{equation}
\begin{aligned}
    D^\parallel \Phi^A{}_{\mathfrak n}
    &=
    \left\{\Phi^A{}_{\mathfrak n},
    H_\parallel\right\},
    \\
    D^\parallel \pi_A{}^{\mathfrak n}
    + \theta_n\pi_A{}^{\mathfrak n}
    &=
    \left\{\pi_A{}^{\mathfrak n},
    H_\parallel\right\},
\end{aligned}
\end{equation}

\noindent and therefore

\begin{equation}\label{eq:FHundensitized}
\begin{aligned}
    \left\{F,H_\parallel\right\}_{\Sigma_\tau}
    ={}&
    \\
    &
    \int_{\Sigma_\tau} d\Sigma
    \left[
    \frac{\delta^\perp F}
    {\delta \Phi^A{}_{\mathfrak n}}
    \right.
    \\
    &\qquad\left.
    {}\times
    D^\parallel \Phi^A{}_{\mathfrak n}
    +
    \frac{\delta^\perp F}
    {\delta \pi_A{}^{\mathfrak n}}
    \right.
    \\
    &\qquad\left.
    {}\times
    \left(
    D^\parallel \pi_A{}^{\mathfrak n}
    +
    \theta_n\pi_A{}^{\mathfrak n}
    \right)
    \right].
\end{aligned}
\end{equation}

\noindent The term proportional to $\theta_n$ appears because $\pi_A{}^{\mathfrak n}$ is written as an undensitized momentum while the symplectic pairing is integrated with the hypersurface volume element. Indeed, $D^\parallel d\Sigma = \theta_n d\Sigma$, so the bracket naturally evolves the momentum density $d\Sigma\,\pi_A{}^{\mathfrak n}$.

Now let $F^I = F^{\alpha_1\cdots\alpha_p}{}_{\beta_1\cdots\beta_q}$ be a tensor-valued phase-space functional. Since $D^\parallel=\pounds_n$, the free indices of $F^I$ are Lie-dragged. We write

\begin{equation}
    D^\parallel F^I
    =
    n^\rho\nabla_\rho F^I
    +
    \mathfrak L_n F^I,
\end{equation}

\noindent where $\mathfrak L_n F^I$ denotes the free-index part of the Lie derivative already displayed in Eq. \eqref{def:evolDeriv}. The same notation is used for the canonical variables themselves. The Lie-chain rule for $F^I[\Phi,\pi;\tau]$ is then

\begin{equation}\label{eq:LieChainRule}
\begin{aligned}
    D^\parallel F^I
    =
    &
    \int_{\Sigma_\tau} d\Sigma
    \left[
    \frac{\delta^\perp F^I}
    {\delta \Phi^A{}_{\mathfrak n}}
    \right.
    \\
    &\qquad\left.
    {}\times
    D^\parallel \Phi^A{}_{\mathfrak n}
    +
    \frac{\delta^\perp F^I}
    {\delta \pi_A{}^{\mathfrak n}}
    \right.
    \\
    &\qquad\left.
    {}\times
    D^\parallel \pi_A{}^{\mathfrak n}
    \right]
    +
    \partial^\parallel_{\mathrm L}F^I ,
\end{aligned}
\end{equation}

\noindent where the explicit Lie-time derivative is

\begin{equation}\label{eq:partialLieDefinition}
\begin{aligned}
    \partial^\parallel_{\mathrm L}F^I
    \equiv{}&
    \left(
    \frac{\partial F^I}{\partial c\tau}
    \right)_{\Phi,\pi}
    +
    \mathfrak L_n F^I
    \\
    &-
    \int_{\Sigma_\tau} d\Sigma
    \left[
    \frac{\delta^\perp F^I}
    {\delta \Phi^A{}_{\mathfrak n}}
    \right.
    \\
    &\qquad\left.
    {}\times
    \left(\mathfrak L_n\Phi\right)^A{}_{\mathfrak n}
    +
    \frac{\delta^\perp F^I}
    {\delta \pi_A{}^{\mathfrak n}}
    \right.
    \\
    &\qquad\left.
    {}\times
    \mathfrak L_n\pi_A{}^{\mathfrak n}
    \right].
\end{aligned}
\end{equation}

\noindent This is the Lie-derivative replacement of the old connection correction; one subtracts the index-dragging already generated through the canonical variables and keeps only the explicit free-index dragging of $F^I$ itself. Comparing Eq. \eqref{eq:FHundensitized} with Eq. \eqref{eq:LieChainRule} gives the observer-covariant Poisson evolution law

\begin{equation}\label{eq:observerPoissonEvolutionUndensitized}
\begin{aligned}
    D^\parallel F^I
    ={}&
    \left\{F^I,H_\parallel\right\}_{\Sigma_\tau}
    +
    \partial^\parallel_{\mathrm L}F^I
    \\
    &-
    \int_{\Sigma_\tau} d\Sigma\,
    \frac{\delta^\perp F^I}
    {\delta \pi_A{}^{\mathfrak n}}
    \theta_n\pi_A{}^{\mathfrak n}.
\end{aligned}
\end{equation}

\noindent This is the correct bracket formula when the canonical momentum is written as the undensitized field $\pi_A{}^{\mathfrak n}$. If we define instead the densitized canonical momentum $\varpi_A{}^{\mathfrak n} \equiv d\Sigma\,\pi_A{}^{\mathfrak n}$, for which $D^\parallel \varpi_A{}^{\mathfrak n}=d\Sigma\left(D^\parallel\pi_A{}^{\mathfrak n}+\theta_n\pi_A{}^{\mathfrak n}\right)$, the expansion term is absorbed into the canonical pair $\left(\Phi^A{}_{\mathfrak n},\varpi_A{}^{\mathfrak n}\right)$ and the evolution law takes the clean canonical form

\begin{equation}\label{eq:observerPoissonEvolutionDensitized}
    D^\parallel F^I
    =
    \left\{F^I,H_\parallel\right\}_{\Sigma_\tau}^{(\varpi)}
    +
    \partial^\parallel_{\mathrm L}F^I .
\end{equation}

\noindent For scalar canonical variables and scalar functionals the Lie-index correction vanishes and $\partial^\parallel_{\mathrm L}F$ reduces to $\left(\partial F/\partial c\tau\right)_{\Phi,\pi}$. Likewise, whenever the observer congruence preserves the hypersurface volume, $\theta_n=0$, the undensitized and densitized descriptions coincide.
 We now pass from the classical brackets to the wave
equation they produce. After quantizing the formalism for a scalar
field $F$ we shall see that it describes proper-time evolution of
$F$ along the congruence. To this end, let us use the quantization rule
$\left\{F, H_\parallel\right\} \longrightarrow [\hat{F}, \hat{H}_\parallel] / i\hbar$
to get the relativistic Heisenberg equation
\begin{equation}
    D^\parallel \hat{F}
    =
    \frac{1}{i\hbar}[\hat{F}, \hat{H}_\parallel]
    +
    \partial^\parallel \hat{F}.
\end{equation}

\noindent Now assume that $\hat{F}$ has no explicit dependence on the proper
time $\tau$ of the congruence. Given that it is a scalar field no
Christoffel symbols appear, thus the explicit evolution term vanishes
$\partial^\parallel \hat{F}=0$. Then the evolution derivative also
reduces to differentiation with respect to the length-time parameter
$c\tau$ as $D^\parallel \longrightarrow d /d(c\tau)$, so the
Heisenberg equation becomes

\begin{equation}
    i\hbar\frac{d\hat{F}}{dc\tau}
    =
    [\hat{F}, \hat{H}_\parallel].
\end{equation}

\noindent For a $\tau$-independent Hamiltonian, the solution is
\begin{equation}
    \hat{F}_H(\tau)
    =
    e^{ic\tau\hat{H}_\parallel/\hbar}
    \hat{F}_H(0)
    e^{-ic\tau\hat{H}_\parallel/\hbar},
\end{equation}
\noindent where the subscript $H$ denotes the Heisenberg picture. The
corresponding evolution operator is
\begin{equation}
    \hat{U}(\tau)
    \equiv
    e^{-ic\tau\hat{H}_\parallel/\hbar}.
\end{equation}
\noindent With this definition the Heisenberg and Schrödinger operators
are of course related by $\hat{F}_H(\tau) = \hat{U}^\dagger(\tau)\hat{F}_S\hat{U}(\tau)$
with $\hat{F}_S=\hat{F}_H(0)$; thus states carry the
proper-time dependence in the Schrödinger picture via
\begin{equation}
    \ket{\Psi_S(\tau)}
    \equiv
    \hat{U}(\tau)\ket{\Psi_H},
\end{equation}
\noindent where differentiating this definition gives the relativistic wave
equation
\begin{equation}\label{relativisticSchrodinger}
    \hat{H}_\parallel\ket{\Psi_S(\tau)}
    =
    i\hbar\frac{\partial}{\partial c\tau}\ket{\Psi_S(\tau)}.
\end{equation}
The factor $c$ is only the conversion between proper time and the length-time parameter used by $H_\parallel$; it cancels against the corresponding factor in the hypersurface integral defining the congruent Hamiltonian. Equation~\eqref{relativisticSchrodinger} realizes, for an observer congruence, the long-standing idea of an invariant proper-time evolution parameter for relativistic quantum dynamics \cite{Stueckelberg1942,HorwitzPiron1973}.

The physical content of Eq.~\eqref{relativisticSchrodinger} is easiest to see for a single particle. Decomposing the worldline displacement as $dx^\mu = n^\mu c\,d\tau + dx^\mu_\perp$ with projected velocity $v^\mu_\perp \equiv dx^\mu_\perp/(c\,d\tau)$ and $v^2_\perp \equiv -g_{\mu\nu}v^\mu_\perp v^\nu_\perp$, the line element becomes $ds = c\,d\tau\sqrt{1-v^2_\perp}$, so the congruence-adapted free-particle Lagrangian is $L_\parallel = -mc\sqrt{1-v^2_\perp}$ and its Legendre transformation yields the square-root Hamiltonian. Equation~\eqref{relativisticSchrodinger} then reads

\begin{equation}\label{schPenyaEq}
\begin{aligned}
    i\hbar\frac{d}{d\tau}\ket{\Psi(\tau)}
    &=
    \sqrt{m^2c^4+c^2\hat{p}^2_\perp}\,\ket{\Psi(\tau)},
    \\
    \hat{p}^\perp_\mu
    &=
    -i\hbar\,D^\perp_\mu,
\end{aligned}
\end{equation}

\noindent with $\hat{p}^2_\perp \equiv -g^{\mu\nu}\hat{p}^\perp_\mu\hat{p}^\perp_\nu$. The projected momentum operator is automatically spatial, $n^\mu\hat{p}^\perp_\mu = 0$, generates translations in the observer's local rest space, and obeys the canonical commutator $[\hat{x}^\mu_\perp,\hat{p}^\perp_\nu] = i\hbar\,{\mathcal{P}^\mu}_\nu$ together with the corresponding uncertainty relation. The square root is a pseudodifferential operator of the spinless-Salpeter type \cite{LuchaSchoberl2004Facets}, non-local over roughly a Compton wavelength, and it defines a positive-energy one-particle theory in the sense of Newton and Wigner \cite{NewtonWigner1949}, here referred to the proper time of the congruence. Since it is self-adjoint on that Hilbert space the norm $\int_{\Sigma_\tau}d\Sigma\,|\psi|^2$ is conserved along $\tau$ and the expansion of the square root generates the associated probability current order by order, whose leading term is the usual Schrödinger current written in the observer's rest space.

The known wave equations are recovered from Eq.~\eqref{schPenyaEq} in the appropriate limits. Applying $i\hbar\,d/d\tau$ once more squares the energy and gives

\begin{equation}
    \frac{1}{c^2}\frac{d^2\psi}{d\tau^2}
    -
    \Delta_\Sigma\psi
    +
    \frac{m^2c^2}{\hbar^2}\psi
    =
    0,
    \qquad
    \Delta_\Sigma \equiv D^\perp_\mu D^{\perp\mu},
\end{equation}

\noindent which is the Klein--Gordon equation \cite{Klein1926,Gordon1926} referred to the proper time of the congruence. In an adapted inertial congruence, where $d/d\tau\to\partial/\partial t$ and $\Delta_\Sigma\to\nabla^2$, it reduces to the standard flat-spacetime form. Linearizing the square root instead with the split Clifford algebra $\gamma^\parallel \equiv n_\mu\gamma^\mu$ and $\gamma^\mu_\perp \equiv {\mathcal{P}^\mu}_\nu\gamma^\nu$ gives the first-order generator $\hat{E} = c\,\gamma^\parallel\gamma^\mu_\perp\hat{p}^\perp_\mu + mc^2\gamma^\parallel$, which satisfies $\hat{E}^2 = m^2c^4 + c^2\hat{p}^2_\perp$, and hence the Dirac equation written along the congruence,

\begin{equation}
    \gamma^\parallel\frac{d\psi}{d\tau}
    +
    c\,\gamma^\mu_\perp D^\perp_\mu\psi
    +
    i\frac{mc^2}{\hbar}\,\psi
    =
    0,
\end{equation}

\noindent which becomes the usual Dirac equation \cite{Dirac1928Electron} in the same adapted limit. Finally, expanding the square root as $mc^2 + \hat{p}^2_\perp/2m + \cdots$ in the spirit of the Foldy--Wouthuysen reduction \cite{FoldyWouthuysen1950} and removing the rest-energy phase $\psi = e^{-imc^2\tau/\hbar}\varphi$ leaves the non-relativistic Schrödinger equation for $\varphi$ in the low-velocity limit.

Having shown the limits under which we recover the usual quantum mechanics in the one-particle/fields regime, we now apply the formalism to gravity.
 \section{Canonical gravity}
\label{sec:lie-dragged-gravity-projection-constraints}

We now have the formalism and all the tools needed to devise a canonical
formulation of Einstein's gravity. The kinematics of the metric along a
hypersurface-orthogonal congruence is classical material
\cite{JantzenCariniBini1992GEM,Ehlers1993Translation,EllisvanElst1999Cosmological,Gourgoulhon2012ThreePlusOne},
so we only recall the decomposition we need. Projecting
$\pounds_n g_{\mu\nu} = \nabla_\mu n_\nu + \nabla_\nu n_\mu$ along and orthogonal
to the congruence the normal-normal component vanishes identically because
$n^\mu$ has unit norm and the mixed normal-spatial component reproduces the
acceleration $a_\mu = \pounds_n n_\mu$ of the observers, which is congruence
data rather than metric data. The propagating velocity of the metric therefore
resides entirely in the spatial-spatial projection, and using
$\gamma_{\mu\nu} = n_\mu n_\nu - g_{\mu\nu}$ we can write the decomposition as

\begin{equation}
  \pounds_n g_{\mu\nu} = a_\mu n_\nu + n_\mu a_\nu - 2K_{\mu\nu},
\end{equation}

\noindent where

\begin{equation}
  K_{\mu\nu} = \frac{1}{2}\mathcal{P}_\mu{}^\alpha\mathcal{P}_\nu{}^\beta\pounds_n \gamma_{\alpha\beta} = -\mathcal{P}_\mu{}^\alpha\mathcal{P}_\nu{}^\beta\nabla_{(\alpha}n_{\beta)}
\end{equation}

\noindent is the extrinsic curvature of the leaves, here playing the role of the
actual (spatial) velocity of the metric. With our conventions the trace
$\gamma^{\mu\nu}K_{\mu\nu} = \theta_n$ is positive for an expanding congruence,
so $K_{\mu\nu}$ coincides with the expansion tensor of the GEM and $1+3$
literature \cite{JantzenCariniBini1992GEM,EllisvanElst1999Cosmological} and
carries the opposite sign to the extrinsic curvature convention of
\cite{Gourgoulhon2012ThreePlusOne}. We shall accordingly work with
$\left(\gamma_{\mu\nu}, K_{\mu\nu}\right)$ as the propagating degrees of freedom rather than
$\left(g_{\mu\nu}, 0\right)$, as $g$ itself does not give us a satisfactory notion of velocity.
It is worth noting that the metric dynamics are observer-dependent, as the velocity
can also be expressed in terms of $n^\mu$ alone.

Inspired by the methodology originally used in the ADM papers
\cite{Arnowitt1962Dynamics,Dirac1958HamiltonianGR,ADM1959DynamicalStructure,ADM1960Canonical},
we shall now put the Einstein-Hilbert action

\begin{equation}
  S_\mathrm{EH} = \frac{1}{2\kappa} \int{d^4x\sqrt{-g}R},
\end{equation}

\noindent into a more explicit form in terms of $\gamma_{\mu\nu}$ and $K_{\mu\nu}$. Let
$R^\perp$ be the Ricci scalar of $\gamma_{\mu\nu}$ on the leaves $\Sigma_\lambda$. The
contracted Gauss--Codazzi relations, together with the Ricci identity applied to $n^\mu$
give the standard decomposition of the scalar curvature
\cite{ADM1960Canonical,York1972ConformalThreeGeometry,Gourgoulhon2012ThreePlusOne}

\begin{equation}
  R
  =
  R^\perp
  +
  K_{\mu\nu}K^{\mu\nu}
  -
  K^2
  +
  2\nabla_\mu
  \left(
    K n^\mu-a^\mu
  \right),
  \label{eq:scalar-curvature-gcodazzi}
\end{equation}

\noindent with $K \equiv \gamma^{\mu\nu}K_{\mu\nu}$, whose last term is a total divergence.
This divergence amounts to the observer-congruent Gibbons--Hawking--York term
\cite{York1972ConformalThreeGeometry,GibbonsHawking1977Action} that
shall be added to the action as a boundary term to also obtain the right total energy
of the system. In this regard, we find

\begin{equation}
\begin{aligned}
  S_\mathrm{CG}
  ={}&
  S_{\partial\Sigma}
  +
  \alpha\int d\lambda
  \int_{\Sigma_\lambda}d\Sigma
  \\
  &\times
  \left[
  R^\perp
  +
  K_{\mu\nu}K^{\mu\nu}
  -
  K^2
  \right],
\end{aligned}
\end{equation}

\noindent with $\lambda \equiv c\tau$ and $\alpha \equiv c^3/16\pi G$. Thus, for the metric momentum

\begin{equation}
  \pi^{\mu\nu} \equiv \frac{\partial \mathcal{L}_\mathrm{CG}}{\partial \left(\pounds_n \gamma_{\mu\nu}\right)} = \alpha \left(K^{\mu\nu} - K\gamma^{\mu\nu}\right),
\end{equation}

\noindent we get

\begin{equation}
\begin{aligned}
  \frac{1}{\alpha}\mathcal{H}_\mathrm{CG}
  \equiv{}&
  \left(K^{\mu\nu}
  -
  K\gamma^{\mu\nu}\right)
  \pounds_n\gamma_{\mu\nu}
  -
  R^\perp
  \\
  &-
  K_{\mu\nu}K^{\mu\nu}
  +
  K^2.
\end{aligned}
\end{equation}

\noindent Up to this point the construction parallels the ADM reduction; the difference is
that no coordinate system adapted to the foliation is introduced, so both members of the
canonical pair remain spacetime tensors and the reduction to spatial variables is carried
by projection constraints instead of by lapse and shift components.
From the fact that
$\gamma_{\mu\nu}$ is spatial we immediately see that the Hessian

\begin{equation}
\begin{aligned}
  W^{\mu\nu\sigma\rho}
  &\equiv
  \frac{\alpha}{2}
  \left(
  \mathbb{I}^{\mu\nu\sigma\rho}_\perp
  -
  \gamma^{\mu\nu}\gamma^{\sigma\rho}
  \right),
  \\
  \mathbb{I}^{\mu\nu\sigma\rho}_\perp
  &\equiv
  \frac{1}{2}
  \left(
  \gamma^{\mu\sigma}\gamma^{\nu\rho}
  +
  \gamma^{\mu\rho}\gamma^{\nu\sigma}
  \right)
\end{aligned}
\end{equation}

\noindent must be singular. We can prove this by considering any velocity variation of the form

\begin{equation}\label{eq:poundsGammaVarFrame}
  \delta\left(\pounds_n \gamma_{\mu\nu}\right) = 2n_{(\mu}\xi_{\nu)} + \chi n_\mu n_\nu, \qquad n^\mu\xi_\mu = 0,
\end{equation}

\noindent which has zero spatial projection. It can then be checked that

\begin{equation}
  W^{\mu\nu\sigma\rho}\delta\left(\pounds_n \gamma_{\sigma\rho}\right) = 0,
\end{equation}

\noindent so the spacetime-index Hessian has a four-dimensional kernel. In an adapted frame this
kernel corresponds to the velocity components $\pounds_n \gamma_{00}$ and $\pounds_n \gamma_{0i}$.
Thus, by Dirac's criterion \cite{Dirac1950GenHamCJM,Dirac1958GenHam} the degeneracy of the Legendre map produces two primary constraints that
can be expressed as

\begin{equation}
  \varphi^\nu \equiv n_\mu \pi^{\mu\nu} \approx 0, \qquad \chi^\nu \equiv n^\mu\gamma_{\mu\nu} \approx 0,
\end{equation}

\noindent where $\approx$ here means ``weak equality'' in the usual sense. These constraints just require
us to keep $\gamma_{\mu\nu}$ and $\pi^{\mu\nu}$ spatial, which come naturally from using a covariant formalism
with a positive-spatial metric. The unreduced Poisson bracket is the symmetric tensor bracket
\begin{equation}
  \{\gamma_{\mu\nu}(x),\pi^{\rho\sigma}(y)\}_{\Sigma_\lambda}
  =
  \delta^\rho{}_{(\mu}\delta^\sigma{}_{\nu)}
  \delta_{\Sigma_\lambda}(x,y),
  \label{eq:unreduced-gamma-pi-bracket}
\end{equation}
with the remaining fundamental brackets vanishing. Treating $n^\mu$ as the
observer structure defining the Hamiltonian split, the bracket between
$\chi_\mu$ and $\varphi^\nu$ is
\begin{equation}
  \{\chi_\mu(x),\varphi^\nu(y)\}_{\Sigma_\lambda}
  =
  M_\mu{}^\nu\,\delta_{\Sigma_\lambda}(x,y),
  \label{eq:constraint-bracket-M}
\end{equation}
where
\begin{equation}
  M_\mu{}^\nu
  =
  \frac12
  \left(
    \delta_\mu{}^\nu+n_\mu n^\nu
  \right).
  \label{eq:M-matrix}
\end{equation}
The matrix $M_\mu{}^\nu$ is invertible as $(M^{-1})_\mu{}^\nu = 2\delta_\mu{}^\nu-n_\mu n^\nu$,
which also satisfies $M_\mu{}^\rho(M^{-1})_\rho{}^\nu = \delta_\mu{}^\nu$. The remaining brackets are
\begin{equation}
  \{\chi_\mu(x),\chi_\nu(y)\}_{\Sigma_\lambda}=0,
  \qquad
  \{\varphi^\mu(x),\varphi^\nu(y)\}_{\Sigma_\lambda}=0,
  \label{eq:constraint-brackets-zero}
\end{equation}
so the constraint matrix for the set $\Xi_A=(\chi_\mu,\varphi^\mu)$ is
\begin{equation}
  C_{AB}(x,y)
  =
  \begin{pmatrix}
    0 & M_\mu{}^\nu \\
    -M_\nu{}^\mu & 0
  \end{pmatrix}
  \delta_{\Sigma_\lambda}(x,y),
  \label{eq:second-class-matrix}
\end{equation}
and since $M_\mu{}^\nu$ is invertible, $C_{AB}$ is too. Therefore
$\chi_\mu$ and $\varphi^\mu$ are second-class projection constraints in the
sense of the Dirac--Bergmann analysis \cite{Dirac1958GenHam,AndersonBergmann1951}.
The total Hamiltonian before reduction is
\begin{equation}
  H_{\rm T}
  =
  H_{\rm c}
  +
  \int_{\Sigma_\lambda}d\Sigma\,
  \left(
    \lambda_\mu\varphi^\mu
    +
    \eta^\mu\chi_\mu
  \right),
  \label{eq:total-H-projection-constraints}
\end{equation}
where $\lambda_\mu$ and $\eta^\mu$ are Lagrange multipliers. Constraint
preservation gives
\begin{equation}
  \begin{split}
    0 &\approx \pounds_n\chi_\mu =
    \{\chi_\mu,H_{\rm c}\}_{\Sigma_\lambda}
    + M_\mu{}^\nu\lambda_\nu, \\
    0 &\approx \pounds_n\varphi^\mu =
    \{\varphi^\mu,H_{\rm c}\}_{\Sigma_\lambda}
    - \eta^\nu M_\nu{}^\mu .
  \end{split}
\end{equation}
Because $M_\mu{}^\nu$ is invertible these equations determine the multipliers
$\lambda_\mu$ and $\eta^\mu$ and thus they do not generate secondary constraints.
The Dirac bracket associated with the second-class set $(\chi_\mu,\varphi^\mu)$ is
\begin{align}
  \{F,G\}_{\rm D}
  &=
  \{F,G\}_{\Sigma_\lambda}
  +
  \int_{\Sigma_\lambda}d\Sigma\,
  \{F,\chi_\mu\}_{\Sigma_\lambda}
  \nonumber\\
  &\quad\times
  (M^{-1})^\mu{}_\nu
  \{\varphi^\nu,G\}_{\Sigma_\lambda}
  \nonumber\\
  &\quad
  -
  \int_{\Sigma_\lambda}d\Sigma\,
  \{F,\varphi^\mu\}_{\Sigma_\lambda}
  \nonumber\\
  &\quad\times
  (M^{-1})^\nu{}_\mu
  \{\chi_\nu,G\}_{\Sigma_\lambda},
  \label{eq:dirac-bracket-projection}
\end{align}
where with this bracket the projection constraints may be imposed strongly
as $\chi_\mu=0$ and $\varphi^\mu=0$;
and the fundamental Dirac bracket becomes the projected symmetric identity
\begin{equation}
  \{\gamma_{\mu\nu}(x),\pi^{\rho\sigma}(y)\}_{\rm D}
  =
  \mathbb P_{\mu\nu}{}^{\rho\sigma}
  \delta_{\Sigma_\lambda}(x,y),
  \label{eq:projected-fundamental-dirac-bracket}
\end{equation}
where
\begin{equation}
  \mathbb P_{\mu\nu}{}^{\rho\sigma}
  \equiv
  \frac12
  \left(
    \mathcal P_\mu{}^\rho\mathcal P_\nu{}^\sigma
    +
    \mathcal P_\mu{}^\sigma\mathcal P_\nu{}^\rho
  \right).
  \label{eq:projected-identity-mixed}
\end{equation}
The remaining fundamental Dirac brackets vanish
\begin{equation}
\begin{aligned}
  \{\gamma_{\mu\nu}(x),
  \gamma_{\rho\sigma}(y)\}_{\rm D}
  &=0,
  \\
  \{\pi^{\mu\nu}(x),
  \pi^{\rho\sigma}(y)\}_{\rm D}
  &=0,
\end{aligned}
  \label{eq:remaining-dirac-brackets}
\end{equation}
so the Dirac bracket automatically projects all canonical variations onto the
observer-spatial subspace. If we enforce these constraints on Eq. (\ref{eq:poundsGammaVarFrame})
we get the determinant of the reduced spatial Hessian, which is
\begin{equation}
  \det W_{\rm red}
  =
  \left(\frac{\alpha}{2}\right)^5(-\alpha)
  =
  -\frac{\alpha^6}{32},
  \label{eq:reduced-hessian-determinant}
\end{equation}
and thus non-degenerate. The Hamilton equations then become

\begin{equation}
  \begin{split}
    \pounds_n\gamma_{\mu\nu}
    &=
    \{\gamma_{\mu\nu},H_{\rm GR}[n]\}_{\rm D}, \\
    \pounds_n\pi^{\mu\nu}
    &=
    \{\pi^{\mu\nu},H_{\rm GR}[n]\}_{\rm D}
    -
    \theta_n\pi^{\mu\nu}.
  \end{split}
\end{equation}

\noindent Imposing the constraints strongly, this can be expressed as

\begin{equation}
  \pounds_n\gamma_{\mu\nu}
  =
  \frac{2}{\alpha}
  \left(
    \pi_{\mu\nu}
    -
    \frac12\pi\gamma_{\mu\nu}
  \right),
\end{equation}

\begin{equation}
  \begin{split}
    \pounds_n\pi^{\mu\nu}
    &=
    -\alpha
    \left(
      R_\perp^{\mu\nu}
      -
      \frac12\gamma^{\mu\nu}R^\perp
    \right)
    \\
    &\quad
    -
    \frac{1}{\alpha}
    \left(
      2\pi^\mu{}_\rho
      \pi^{\nu\rho}
      -
      \pi\pi^{\mu\nu}
    \right)
    \\
    &\quad
    -
    \frac{1}{2\alpha}\gamma^{\mu\nu}
    \left(
      \pi_{\rho\sigma}\pi^{\rho\sigma}
      -
      \frac12\pi^2
    \right)
    \\
    &\quad
    -
    \theta_n\pi^{\mu\nu},
  \end{split}
\end{equation}

\noindent with $\pi \equiv \gamma_{\mu\nu}\pi^{\mu\nu}$. Alternatively,
we can also express them in a more compact form in terms of $K_{\mu\nu}$ as

\begin{equation}
  \pounds_n\gamma_{\mu\nu} = 2K_{\mu\nu},
\end{equation}

\begin{equation}
  \begin{split}
    \pounds_nK_{\mu\nu}
    &=
    2K_{\mu\rho}K^\rho{}_\nu
    -R^\perp_{\mu\nu}
    -
    KK_{\mu\nu}
    \\
    &+
    \frac14\gamma_{\mu\nu}
    \left(
      R^\perp
      +
      K_{\rho\sigma}K^{\rho\sigma}
      -
      K^2
    \right).
  \end{split}
\end{equation}
 \section{Quantization of gravity}
\label{sec:quantization-of-gravity}

The Hamiltonian from the previous section can be written as
\begin{equation}
  \begin{split}
    \mathcal{H}^\mathrm{CG}_\parallel
    &=
    \alpha
    \left(
      K_{\mu\nu}K^{\mu\nu}
      -
      K^2
      -
      R^\perp
    \right)
    \\
    &=
    \frac{1}{\alpha}
    \left(
      \pi_{\mu\nu}\pi^{\mu\nu}
      -
      \frac{1}{2}\pi^2
    \right)
    -
    \alpha R^\perp .
  \end{split}
  \label{eq:cg-hamiltonian-density-before-quantization}
\end{equation}
Equivalently, if we choose three independent vectors $e^\mu{}_i$ tangent to the leaf $\Sigma_\lambda$ then they satisfy $n_\mu e^\mu{}_i = 0$ and we can define
\begin{equation}\label{def:detSigma}
\begin{aligned}
  \gamma_\Sigma(x)
  &\equiv
  \det_\Sigma\gamma_{\mu\nu}(x)
  =
  \det\gamma_{ij}(x),
  \\
  \gamma_{ij}
  &=
  \gamma_{\mu\nu}e^\mu{}_ie^\mu{}_j
  =
  -g_{\mu\nu}e^\mu{}_ie^\mu{}_j,
\end{aligned}
\end{equation}
and defining $d\Sigma=d^3x\sqrt{\gamma_\Sigma}$ as the spatial volume element, then the full Hamiltonian is
\begin{equation}
  \begin{aligned}
    H^\mathrm{CG}_\parallel[\gamma,\pi]
    ={}&
    \int_{\Sigma_\tau}d\Sigma
    \left[
      \frac{1}{\alpha}
      \mathbb G_{\mu\nu\rho\sigma}[\gamma]\,
    \right.\\
    &\left.\qquad\qquad
      {}\times\pi^{\mu\nu}\pi^{\rho\sigma}
      -
      \alpha R^\perp[\gamma]
    \right],
  \end{aligned}
  \label{eq:cg-hamiltonian-coordinate-form}
\end{equation}
where
\begin{equation}
  \mathbb G_{\mu\nu\rho\sigma}[\gamma]
  \equiv
  \frac{1}{2}
  \left(
    \gamma_{\mu\rho}\gamma_{\nu\sigma}
    +
    \gamma_{\mu\sigma}\gamma_{\nu\rho}
    -
    \gamma_{\mu\nu}\gamma_{\rho\sigma}
  \right)
  \label{eq:dewitt-supermetric-lowered}
\end{equation}
is the observer-projected analogue of the Wheeler--DeWitt supermetric metric used in canonical geometrodynamics \cite{DeWitt1967Quantum,Misner1969QuantumCosmology,Kiefer2012QuantumGravity}.
All indices in Eq. \eqref{eq:dewitt-supermetric-lowered} are understood
as projected with respect to the observer congruence $n^\mu$.

We now quantize the reduced observer-spatial phase space following the
Dirac/canonical quantization of constrained gravitational systems
\cite{DeWitt1967Quantum,Dirac1950GenHamCJM,Dirac1958GenHam,Dirac1958HamiltonianGR}.
The equal-$\tau$ Dirac bracket is promoted to the projected commutator
\begin{equation}
  [\hat\gamma_{\mu\nu}(x),\hat\pi^{\rho\sigma}(y)]
  =
  i\hbar\,
  \mathbb P_{\mu\nu}{}^{\rho\sigma}
  \delta_{\Sigma_\tau}(x,y),
  \label{eq:commutatorRelations}
\end{equation}
where
\begin{equation}
  \mathbb P_{\mu\nu}{}^{\rho\sigma}
  =
  \frac{1}{2}
  \left(
    \mathcal P_\mu{}^\rho\mathcal P_\nu{}^\sigma
    +
    \mathcal P_\mu{}^\sigma\mathcal P_\nu{}^\rho
  \right).
  \label{eq:projected-symmetric-identity-qg}
\end{equation}
The projection constraints are imposed strongly,
\begin{equation}
  n_\mu\hat\pi^{\mu\nu}\ket{\Psi}=0,
  \qquad
  n^\mu\hat\gamma_{\mu\nu}\ket{\Psi}=0.
  \label{eq:qgConstraints}
\end{equation}
Thus the wavefunctional depends only on projected symmetric metric
configurations,
\begin{equation}
  \Psi=\Psi[\gamma_{\mu\nu};\tau],
  \qquad
  n^\mu\gamma_{\mu\nu}=0.
  \label{eq:spatial-wavefunctional}
\end{equation}
In the metric representation familiar from canonical quantum
gravity \cite{DeWitt1967Quantum,Jackiw1990Analysis,Symanzik1981Schrodinger},
\begin{equation}
  \hat\gamma_{\mu\nu}(x)\Psi[\gamma;\tau]
  =
  \gamma_{\mu\nu}(x)\Psi[\gamma;\tau],
  \label{eq:metric-representation}
\end{equation}
and
\begin{equation}
  \hat\pi^{\mu\nu}(x)\Psi[\gamma;\tau]
  =
  -i\hbar
  \frac{\delta_\perp}{\delta\gamma_{\mu\nu}(x)}
  \Psi[\gamma;\tau],
  \label{eq:projected-momentum-representation}
\end{equation}
where the projected functional derivative is
\begin{equation}
  \frac{\delta_\perp}{\delta\gamma_{\mu\nu}(x)}
  \equiv
  \mathbb P^{\mu\nu}{}_{\rho\sigma}
  \frac{\delta}{\delta\gamma_{\rho\sigma}(x)}.
  \label{eq:projected-functional-derivative}
\end{equation}

\noindent At this point one must choose an ordering for the kinetic term.
Algebraic symmetrization of the products $\pi_{\mu\nu}\pi^{\mu\nu}$ and
$\pi^2$ depends on the coordinate used on configuration space, so we use
the minimal field-redefinition-covariant ordering proved in Theorem
\ref{thm:ordering}, i.e., the kinetic term is ordered as the
Laplace--Beltrami operator associated with the kinetic metric on the reduced
configuration space of projected spatial metrics, as is common in
minisuperspace treatments
\cite{DeWitt1967Quantum,Misner1969QuantumCosmology,Kiefer2012QuantumGravity,HawkingPage1986Ordering,Halliwell1990QuantumCosmology}.

To derive it, define the gravitational inverse field-space
metric by absorbing the volume density into the kinetic tensor
\begin{equation}
  \mathscr K_{\mu\nu\rho\sigma}(x;\gamma)
  \equiv
  \frac{2}{\alpha}
  \sqrt{\gamma_\Sigma(x)}\,
  \mathbb G_{\mu\nu\rho\sigma}(x;\gamma),
  \label{eq:gravity-inverse-field-space-metric}
\end{equation}
then the kinetic Hamiltonian can be written in the natural form
\begin{equation}
  H_{\mathrm{kin}}^\mathrm{CG}
  =
  \frac{1}{2}
  \int_{\Sigma_\lambda}d^3x\,
  \mathscr K_{\mu\nu\rho\sigma}(x;\gamma)\,
  \pi^{\mu\nu}(x)\pi^{\rho\sigma}(x),
  \label{eq:gravity-kinetic-natural-form}
\end{equation}
and the corresponding field-space metric is the inverse of
$\mathscr K_{\mu\nu\rho\sigma}$. In three observer-spatial dimensions
it is
\begin{equation}
\begin{aligned}
  \mathscr K^{\mu\nu\rho\sigma}
  & (x;\gamma)
  =
  \frac{\alpha}{2}
  \frac{1}{\sqrt{\gamma_\Sigma(x)}}
  \\
  &\times
  \left[
    \frac{1}{2}
    \left(
      \gamma^{\mu\rho}\gamma^{\nu\sigma}
      +
      \gamma^{\mu\sigma}\gamma^{\nu\rho}
    \right)
    -
    \gamma^{\mu\nu}\gamma^{\rho\sigma}
  \right],
\end{aligned}
  \label{eq:gravity-field-space-metric}
\end{equation}
with all inverses understood on the observer-spatial subspace. Now using the $\det\nolimits_\Sigma \gamma_{\mu\nu}(x)$ defined in Eq.~\eqref{def:detSigma} the determinant of the purely DeWitt part at each spatial point is
\begin{equation}
  \det_\Sigma
  \left[
    \frac{1}{2}
    \left(
      \gamma^{\mu\rho}\gamma^{\nu\sigma}
      +
      \gamma^{\mu\sigma}\gamma^{\nu\rho}
    \right)
    -
    \gamma^{\mu\nu}\gamma^{\rho\sigma}
  \right]
  =
  -16\,\gamma_\Sigma^{-4},
  \label{eq:dewitt-inverse-determinant}
\end{equation}
where the determinant is taken on the six-dimensional space of projected
symmetric components. To compute it we use the reduced index notation
$A\equiv (\mu\nu)$ and $B=(\sigma\rho)$; this way, we can say
$G^{\mu\nu\sigma\rho} \leftrightarrow G^{AB}$ and take the determinant
of $G^{AB}$ given $A = 1 \leftrightarrow (11)$,
$A = 2 \leftrightarrow (22)$, $A = 3 \leftrightarrow (33)$,
$A = 4 \leftrightarrow (23)$, $A = 5 \leftrightarrow (12)$,
$A = 6 \leftrightarrow (13)$ and the same for $B$.
Therefore
\begin{equation}
  \begin{split}
  \det_\Sigma \mathscr K^{\mu\nu\rho\sigma}(x;\gamma)
  &=
  \left(
    \frac{\alpha}{2\sqrt{\gamma_\Sigma(x)}}
  \right)^6
  \left(
    -16\,\gamma_\Sigma^{-4}(x)
  \right),
  \end{split}
  \label{eq:gravity-field-space-metric-determinant}
\end{equation}
and hence the local Laplace--Beltrami density is
\begin{equation}
\begin{aligned}
  \sqrt{\mathscr M_{\mathrm{grav},x}[\gamma]}
  &=
  \sqrt{\left|
  \det_\Sigma
  \mathscr K^{\mu\nu\rho\sigma}(x;\gamma)
  \right|}
  \\
  &=
  \frac{\alpha^3}{2}
  \left[\gamma_\Sigma(x)\right]^{-7/2},
\end{aligned}
  \label{eq:gravity-local-lb-density}
\end{equation}
so the full functional density becomes
\begin{equation}
  \sqrt{\mathscr M_{\mathrm{grav}}[\gamma]}
  =
  \prod_{x\in\Sigma_\lambda}^{\prime}
  \frac{\alpha^3}{2}
  \left[\gamma_\Sigma(x)\right]^{-7/2},
  \label{eq:gravity-functional-lb-density}
\end{equation}
where the prime denotes the product over the projected symmetric
configuration variables only, an ultralocal product whose action on
cylinder wavefunctionals is finite block by block in the sense of
Remark \ref{rem:field-space-extension}.
This is the same as to say
\begin{equation}
  \ln\sqrt{\mathscr M_{\mathrm{grav}}[\gamma]}
  =
  \mathrm{Tr}^{\perp}
  \left[
    \ln\left(\frac{\alpha^3}{2}\right)
    -
    \frac{7}{2}\ln\gamma_\Sigma
  \right],
  \label{eq:gravity-functional-lb-density-log}
\end{equation}
where the trace is over the reduced projected field-space modes.
Our ordered kinetic operator is therefore
\begin{equation}
  \hat H_{\mathrm{kin}}^\mathrm{CG}
  =
  -\frac{\hbar^2}{2}\Delta_{\mathrm{grav}},
  \label{eq:pena-ordered-kinetic-hamiltonian}
\end{equation}
where
\begin{equation}
\begin{aligned}
  \Delta_{\mathrm{grav}}
  ={}&
  \frac{2}{\alpha}
  \frac{1}{\sqrt{\mathscr M_{\mathrm{grav}}[\gamma]}}
  \\
  &\times
  \int_{\Sigma_\lambda}d^3x\,
  \frac{\delta_\perp}{\delta\gamma_{\mu\nu}(x)}
  \\
  &\times
  \left[
    \sqrt{\mathscr M_{\mathrm{grav}}[\gamma]}\,
    \sqrt{\gamma_\Sigma(x)}
  \right.\\
  &\left.\qquad\qquad
    {}\times\mathbb G_{\mu\nu\rho\sigma}(x;\gamma)
    \frac{\delta_\perp}{\delta\gamma_{\rho\sigma}(x)}
  \right].
\end{aligned}
  \label{eq:gravity-laplace-beltrami-expanded}
\end{equation}
This ordering satisfies the minimality condition
$\hat H_{\mathrm{kin}}^\mathrm{CG}\,1=0$
and therefore introduces no additional scalar quantum potential on
configuration space. Furthermore, contrary to Tsamis and
Woodard assumptions \cite{TsamisWoodard1987Regulated} here we do not
have Hamiltonian or momentum constraints that would make diverging
$\delta^3(0)$ or $[\delta^3(0)]^2$ terms appear.

Now note that $\hat{R}^\perp$ satisfies
$\hat R^\perp[\gamma]\Psi[\gamma;\tau] = R^\perp[\gamma]\Psi[\gamma;\tau]$,
so our equation for quantum gravity without matter is then given by
\begin{equation}
  \begin{split}
    \frac{\partial}{\partial\tau}\Psi[\gamma;\tau]
    &=
    i
    \int_{\Sigma_\tau}d^3x
    \\
    &\quad\times
    \Bigg[
      \frac{16\pi G\hbar}{c^2}
      \frac{1}{\sqrt{\mathscr M_{\mathrm{grav}}[\gamma]}}
      \\
    &\qquad
      {}\times
      \frac{\delta_\perp}{\delta\gamma_{\mu\nu}}
      \\
    &\qquad
      {}\times
      \Bigg(
        \sqrt{\mathscr M_{\mathrm{grav}}[\gamma]}\,
        \sqrt{\gamma_\Sigma}
      \\
    &\qquad\qquad
        {}\times
        \mathbb G_{\mu\nu\rho\sigma}(x;\gamma)
        \frac{\delta_\perp}{\delta\gamma_{\rho\sigma}}
      \Bigg)
      \\
    &\qquad
      +
      \frac{c^4}{16\pi G\hbar}
      \sqrt{\gamma_\Sigma}\,R^\perp[\gamma]
    \Bigg]
    \Psi[\gamma;\tau],
  \end{split}
  \label{eq:QG}
\end{equation}
which is free of ordering ambiguities in the sense of Theorem
\ref{thm:ordering}.

Since factor ordering is a long-standing difficulty of ADM-based canonical
quantum gravity \cite{DeWitt1967Quantum,Kiefer2012QuantumGravity}, two
clarifications are in order. The first concerns the residual ambiguity of
quantization on curved configuration spaces, where covariance and
Hermiticity alone admit the one-parameter family
\begin{equation}
  \hat T_\xi
  =
  -\frac{\hbar^2}{2}
  \left(
    \Delta_{\mathrm{grav}}
    -
    \xi R[\mathscr K]
  \right),
  \label{eq:xi-curvature-family}
\end{equation}
with $R[\mathscr K]$ the scalar curvature of the kinetic metric on the
reduced configuration space \cite{DeWitt1957DynamicalTheory}. Each member
is covariant and Hermitian so these requirements cannot by themselves
select a $\xi$. However, in Corollary \ref{cor:xi-ambiguity} we see that
the $\xi$ term is equivalent to a non-minimal scalar potential
(an order-$\hbar^2$ interaction with no classical counterpart)
and it is shown that the minimality requirement
$\hat T\,1=0$ excludes it because quantization must order the kinetic term of
Eq.~\eqref{eq:cg-hamiltonian-coordinate-form} without introducing new
interactions.

The second clarification concerns why the same theorem does not settle
the ordering of the Wheeler--DeWitt operator in the ADM quantization,
where Laplace--Beltrami orderings on superspace have been employed
since DeWitt's original paper
\cite{DeWitt1967Quantum,HawkingPage1986Ordering}. The operator is
available there but the hypotheses that single it out are not. The ADM
Hamiltonian density is a first-class constraint imposed as
$\hat{\mathcal H}\,\Psi=0$ so no external evolution parameter survives
and the physical inner product on solutions remains unknown until the
problem of time is resolved
\cite{Kuchar2011Time,Isham1993Canonical,Anderson2012Problem}. The
Hermiticity hypothesis of Theorem \ref{thm:ordering} therefore has no
invariant formulation on superspace. The ordering of
$\hat{\mathcal H}$ is moreover entangled with the closure of the
quantum constraint algebra, and this entanglement makes a
regulator indispensable in that formulation. As Tsamis and Woodard
pointed out \cite{TsamisWoodard1987Regulated} the factor-ordering
problem of ADM quantum gravity is by definition the search for an
ordering of the constraint operators under which the physical
conditions $\hat{\mathcal H}\,\Psi=0$ and
$\hat{\mathcal H}_i\,\Psi=0$ admit nonzero solutions, which requires
the commutator of two constraints to reproduce a combination of
constraints with the structure operators standing to their left.
Checking this closure forces one to commute nonlinear products of
$\hat\gamma$ and $\hat\pi$ at coincident points, rearranging distributions
with identities that are valid for smooth test functions and undefined for
operator-valued ones. Applied in different orders those identities yield
contradictory answers, which is how Komar's proposed ordering was shown to
fail and why the constraints carry undetermined $\delta^3(0)$ and
$[\delta^3(0)]^2$ terms; the closure requirement cannot even be checked at
the formal, unregulated level
\cite{Komar1979Ordering,TsamisWoodard1987Regulated,FriedmanJack1988Constraints}.
Remark \ref{rem:tsamis-woodard} traces the argument in detail.

Thankfully, neither obstruction arises in our theory thanks to the lack of
any Hamiltonian or momentum constraints, and with them goes the need for a
regulator. The congruence supplies the proper time $\tau$ and the unitarity
of Eq. \eqref{relativisticSchrodinger} forces its generator to be symmetric
with respect to the invariant configuration-space measure. The only
constraints are the second-class projection pair of Eq.
\eqref{eq:qgConstraints} whose preservation
fixes the multipliers without secondary constraints. No product of
canonical operators enters them and their strong imposition merely
restricts the argument of the wavefunctional in Eq.
\eqref{eq:spatial-wavefunctional}. The premise of the Tsamis--Woodard
necessity argument (a first-class algebra whose quantum closure must be
verified) is therefore absent by construction, and the one remaining
problem (the ordering of the kinetic term) is answered by Theorem
\ref{thm:ordering} at the level of finite-dimensional geometry without any
ultraviolet regulator.

With the ordering settled, we can now ask how many of the structural problems
of ADM-based canonical quantum gravity the present construction actually
resolves. Both formalisms live on a globally hyperbolic spacetime carrying a
unit timelike vector field so the question deserves some attention. In the
ADM scheme the metric is decomposed in a chart adapted to a foliation and
the lapse and shift become undetermined multipliers of first-class constraints,
\begin{equation}
\begin{aligned}
  H_{\rm ADM}
  &=
  \int_{\Sigma_t}d^3x
  \left(
    N\mathcal H
    +
    N^i\mathcal H_i
  \right),
  \\
  \hat{\mathcal H}\,\Psi&=0,
  \\
  \hat{\mathcal H}_i\,\Psi&=0,
\end{aligned}
  \label{eq:adm-constraint-structure}
\end{equation}
so that, up to boundary charges of Regge--Teitelboim type
\cite{ReggeTeitelboim1974SurfaceIntegrals}, the canonical Hamiltonian
vanishes on the constraint surface and the quantum state satisfies a
kernel equation carrying no evolution parameter. This is the frozen
dynamics at the origin of the problem of time
\cite{Kuchar2011Time,Isham1993Canonical,Anderson2012Problem,DeWitt1967Quantum}.

However, note how in our theory the Hamiltonian and momentum constraints
never arise because there is
no embedding data to vary. The congruence $n^\mu$ is observer structure
external to the configuration space so deformations of the leaves are not
transformations the theory must declare gauge. The normal components of the
metric velocity fall in the kernel of the Hessian (Section
\ref{sec:lie-dragged-gravity-projection-constraints}) and the degeneracy they
produce is exhausted by the projection pair of Eq. \eqref{eq:qgConstraints}.
Displacement along the congruence is therefore genuine evolution generated
by the nonvanishing $H^{\rm CG}_\parallel$, while invariance under
observer-spatial diffeomorphisms is carried by the tensorial character of the
variables and the invariant densities $\sqrt{\gamma_\Sigma}$ and
$\sqrt{\mathscr M_{\mathrm{grav}}}$.

The frozen dynamics is also what impedes us to build a sensible Hilbert space for
Wheeler--DeWitt's theory. The kernel equation is second order in the
superspace variables so its Klein--Gordon-type current on the indefinite
DeWitt supermetric fails to be positive and no conserved inner product
exists before a time variable is isolated among the degrees of freedom
\cite{Kuchar2011Time,Isham1993Canonical,DeWitt1967Quantum}. Equation
\eqref{eq:QG} is instead first order in $\tau$ and the
Laplace--Beltrami form of the generator makes it symmetric with
respect to the invariant configuration-space measure, so the inner
product
\begin{equation}
  \langle\Psi_1,\Psi_2\rangle_\tau
  =
  \int D\gamma\,
  \sqrt{\mathscr M_{\mathrm{grav}}[\gamma]}\,
  \Psi_1^*[\gamma;\tau]\,
  \Psi_2[\gamma;\tau]
  \label{eq:qg-inner-product}
\end{equation}
is conserved along the proper-time flow, where $D\gamma$ denotes the
ultralocal product of the single-point volume elements over the
projected symmetric modes evaluated on cylinder wavefunctionals in
the blockwise sense of Remark \ref{rem:field-space-extension}. Its
positivity is a property of the measure alone and holds regardless of
the signature of the projected DeWitt metric which only fixes the
hyperbolic type of $\Delta_{\mathrm{grav}}$, as noted in Remark
\ref{rem:indefinite-kinetic-metric}. For the free wavepackets
constructed below the kinetic tensor is frozen at the background, so the
measure factor is constant and Eq.
\eqref{eq:qg-inner-product} reduces to the norm conservation of Eq.
\eqref{eq:free-norm-conservation}. The state space is thus the $L^2$
completion of the cylinder wavefunctionals in this inner product
carrying a unitary proper-time evolution, and the inner-product
problem of the ADM quantization does not reappear. The only freedom
left in Eq. \eqref{eq:qg-inner-product} is the state-independent
normalization of the ultralocal product measure, which cancels in
every normalized matrix element as discussed in Remark
\ref{rem:field-space-extension}.

Let us now briefly discuss covariance and foliation dependence.
The difference between ADM and our formalism lies in where the split resides. In ADM
it resides in an adapted chart whose slicing freedom is gauge, so the quantum
theory must undo the arbitrariness of the foliation through the closure of the
constraint algebra entailing the difficulties recalled above. But in our formalism it resides in
the projector $\mathcal P^\mu{}_\nu$. Both members of the canonical pair
remain spacetime tensors and no lapse or shift is ever introduced, so there
is no adapted chart to undo. The only thing shared by the two constructions is global
hyperbolicity; i.e., the timelike congruence itself. Yet notoriously the extra ADM layer of an
arbitrary slicing plus the gauge machinery guaranteeing independence from it
is removed in the present formalism. The leaves $\Sigma_\tau$ are
the rest spaces of the chosen observers and the freedom that remains, the
choice of congruence, is physical. Different congruences are different
families of observers, related by the transformation theory of the splitting
formalism \cite{JantzenCariniBini1992GEM}. Relating their quantum descriptions
is a transformation question of the same kind as a change of inertial frame
in special-relativistic quantum theory, which is answered by operators implementing
the change of observer rather than by constraints on the states.

In summary, the frozen dynamics, the inner product, and the factor
ordering are solved by construction in our theory and the foliation
gauge problem never appears because nothing chart-dependent is
introduced.

We now study three solutions of Eq. \eqref{eq:QG}; namely, general
graviton wavepackets, quantum cosmology in a reduced minisuperspace, and
Schwarzschild black holes. In the last two the classical singularities are
resolved.

\subsection{General wavepacket solutions}
\label{subsec:general-wavepacket-solutions}

To solve Eq. (\ref{eq:QG}) let us linearize around a Ricci-flat background

\begin{equation}
  \gamma_{\mu\nu} = \bar{\gamma}_{\mu\nu} + h_{\mu\nu}, \qquad R^\perp[\bar{\gamma}] = 0,
\end{equation}

\noindent In the free approximation we can freeze the kinetic tensor at the background

\begin{equation}
\begin{aligned}
  \sqrt{\mathscr M_{\mathrm{grav}}[\gamma]}
  \mathbb{G}_{\mu\nu\sigma\rho}[\gamma]
  &\longrightarrow
  \bar{\mathbb{G}}_{\mu\nu\sigma\rho}
  \\
  &\equiv
  \mathbb{G}_{\mu\nu\sigma\rho}[\bar{\gamma}]
  \sqrt{\mathscr M_{\mathrm{grav}}[\bar{\gamma}]},
\end{aligned}
\end{equation}
\noindent and in this section $d\Sigma = d^3x(\bar{\gamma}_\Sigma)^{1/2}$. To study the free
wavepacket we for now ignore the curvature potential $R^\perp[\gamma]\longrightarrow 0$, which shall be incorporated later as an interaction term. Then Eq. (\ref{eq:QG}) becomes
\begin{equation}\label{eq:free-QG-equation}
  \partial_\tau\Psi_0[h,\tau]
  =
  i\frac{16\pi G\hbar}{c^2}
  \int_{\Sigma_\tau}d\Sigma\,
  \bar{\mathbb{G}}_{\mu\nu\sigma\rho}
  \frac{\delta^2_\perp\Psi_0}
  {\delta h_{\mu\nu}\delta h_{\rho\sigma}}.
\end{equation}

\noindent
The momentum eigenfunctionals are
\begin{equation}\label{eq:momentum-eigenfunctional}
  \Phi_\Pi[h]
  =
  \exp\!\left[
    \frac{i}{\hbar}
    \int_{\Sigma_\tau}d\Sigma\,
    \Pi^{\mu\nu}h_{\mu\nu}
  \right],
\end{equation}
so
\begin{equation}\label{eq:functional-laplacian-eigenvalue}
\begin{aligned}
  &\int_{\Sigma_\tau}d\Sigma\,
  \bar{\mathbb{G}}_{\mu\nu\sigma\rho}
  \\
  &\qquad\times
  \frac{\delta^2_\perp\Phi_\Pi}
  {\delta h_{\mu\nu}\delta h_{\rho\sigma}}
  \\
  &=
  -
  \frac{1}{\hbar^2}
  \left(
    \int_{\Sigma_\tau}d\Sigma\,
    \bar{\mathbb{G}}_{\mu\nu\sigma\rho}
    \Pi^{\mu\nu}\Pi^{\rho\sigma}
  \right)
  \Phi_\Pi,
\end{aligned}
\end{equation}
and thus the most general free wavepacket is the functional Fourier superposition \cite{Jackiw1990Analysis,Symanzik1981Schrodinger,GuthPi1985,LongShore1998}
\begin{equation}\label{eq:general-free-wavepacket}
\begin{split}
  \Psi_0[h,\tau]
  &=
  \frac{1}{(2\pi\hbar)^{N/2}}
  \int D\Pi\,
  \widetilde{\Psi}_0[\Pi]
  \\
  &\quad\times
  \exp\!\Bigg[
    \frac{i}{\hbar}
    \int_{\Sigma_\tau}d\Sigma\,
    \Pi^{\mu\nu}h_{\mu\nu}
    \\
  &\qquad
    -
    \frac{i}{\hbar}
    \frac{16\pi G}{c^2}
    (\tau-\tau_0)
    \\
  &\qquad\quad\times
    \int_{\Sigma_\tau}d\Sigma\,
    \bar{\mathbb{G}}_{\mu\nu\sigma\rho}
    \Pi^{\mu\nu}\Pi^{\rho\sigma}
  \Bigg].
\end{split}
\end{equation}
Here $D\Pi$ is the momentum-space measure on the projected symmetric modes (see Remark \ref{rem:field-space-extension}) and $N$ is the formal count of those modes.
The profile $\widetilde{\Psi}_0[\Pi]$ is arbitrary, subject only to square integrability
\begin{equation}\label{eq:momentum-profile-normalisation}
  \int D\Pi\,
  |\widetilde{\Psi}_0[\Pi]|^2
  =
  1.
\end{equation}
A useful closed subfamily is obtained by choosing
a Gaussian in momentum space
\cite{JackiwKerman1979Variational,CornwallJackiwTomboulis1974,Stevenson1985Gaussian,Yee1997Variational}
\begin{equation}\label{eq:gaussian-momentum-profile}
\begin{split}
  \widetilde{\Psi}_0[\Pi]
  &=
  \mathcal{N}_\Pi
  \exp\!\Bigg[
    -
    \frac{1}{4\hbar}
    \int_{\Sigma_\tau}d\Sigma_x
    \\
  &\qquad\quad\times
    \int_{\Sigma_\tau}d\Sigma_y\,
    \\
  &\qquad\quad\times
    \left[
      \Pi^{\mu\nu}(x)
      -
      \Pi_{(0)}^{\mu\nu}(x)
    \right]
    \\
  &\qquad\quad\times
    B^{-1}_{\mu\nu\rho\sigma}(x,y)
    \\
  &\qquad\quad\times
    \left[
      \Pi^{\rho\sigma}(y)
      -
      \Pi_{(0)}^{\rho\sigma}(y)
    \right]
    \\
  &\qquad
    -
    \frac{i}{\hbar}
    \int_{\Sigma_\tau}d\Sigma\,
    \Pi^{\mu\nu}h^{(0)}_{\mu\nu}
  \Bigg],
\end{split}
\end{equation}
where $h^{(0)}_{\mu\nu}$ and $\Pi^{\mu\nu}_{(0)}$ are the initial centre of the packet, and
$B^{-1}_{\mu\nu\rho\sigma}(x,y)$ is the inverse covariance kernel.
For a normalisable Gaussian $B > 0$ must be symmetric and strictly positive on the reduced phase space, and trace class for the mode sums below to converge.
On the other side, the normalisation we need is
\begin{equation}\label{eq:NPi-normalisation}
  \mathcal{N}_\Pi
  =
  \left[\det(2\pi\hbar B)\right]^{-1/4},
\end{equation}
and if we introduce the notation
\begin{equation}\label{eq:Ttau-definition}
  T_\tau
  \equiv
  \frac{16\pi G}{c^2}(\tau-\tau_0),
\end{equation}
to be able to write the Fourier integral more compactly, we can define
\begin{equation}\label{eq:Atau-definition}
  A(\tau)
  \equiv
  \left(
    B^{-1}
    +
    \frac{64\pi iG}{c^2}
    (\tau-\tau_0)
    \bar{\mathbb{G}}
  \right)^{-1},
\end{equation}
To then have a Gaussian integral of the form
\begin{equation}\label{eq:gaussian-configuration-wavepacket}
\begin{split}
  \Psi_0[h,\tau]
  &=
  \mathcal{N}(\tau)
  \exp\!\left[
    -
    \frac{1}{\hbar}
    \int_{\Sigma_\tau}d\Sigma_x
    \int_{\Sigma_\tau}d\Sigma_y
    \right.\\
  &\left.
    Z_{\mu\nu}(\tau,x)
    A^{\mu\nu\rho\sigma}(\tau;x,y)
    Z_{\rho\sigma}(\tau,y)
    \right.\\
  &\left.
    +
    \frac{i}{\hbar}
    \int_{\Sigma_\tau}d\Sigma\,
    \Pi_{(0)}^{\mu\nu}
    \left(h_{\mu\nu}-h^{(0)}_{\mu\nu}\right)
    \right.\\
  &\left.
    -
    \frac{i}{\hbar}T_\tau
    \int_{\Sigma_\tau}d\Sigma\,
    \bar{\mathbb{G}}_{\mu\nu\sigma\rho}
    \Pi_{(0)}^{\mu\nu}\Pi_{(0)}^{\rho\sigma}
  \right],
\end{split}
\end{equation}
where
\begin{equation}\label{eq:Ztau-definition}
\begin{aligned}
  Z_{\mu\nu}(\tau,x)
  ={}&
  h_{\mu\nu}(x)
  -
  h^{(0)}_{\mu\nu}(x)
  -
  2T_\tau
  \bar{\mathbb{G}}_{\mu\nu\sigma\rho}(x)
  \\
  &\times
  \Pi_{(0)}^{\rho\sigma}(x).
\end{aligned}
\end{equation}
Therefore the centre of the packet follows
\begin{equation}\label{eq:packet-centre-evolution}
\begin{aligned}
  (h_{\rm cl})_{\mu\nu}(\tau,x)
  ={}&
  h^{(0)}_{\mu\nu}(x)
  +
  2T_\tau
  \bar{\mathbb{G}}_{\mu\nu\sigma\rho}(x)
  \Pi_{(0)}^{\rho\sigma}(x)
  \\
  ={}&
  h^{(0)}_{\mu\nu}(x)
  +
  \frac{32\pi G}{c^2}
  (\tau-\tau_0)
  \bar{\mathbb{G}}_{\mu\nu\sigma\rho}(x)
  \\
  &\times
  \Pi_{(0)}^{\rho\sigma}(x).
\end{aligned}
\end{equation}
With this eigenstate the configuration-space normalisation factor is
\begin{equation}\label{eq:configuration-normalisation}
  \mathcal{N}(\tau)
  =
  \left[\det\!\left(\frac{2B}{\pi\hbar}\right)\right]^{1/4}
  \left[
    \det
    \left(
      \mathbf{1}
      +
      4iT_\tau B\bar{\mathbb{G}}
    \right)
  \right]^{-1/2},
\end{equation}
where the determinant branch is chosen continuously from $\tau=\tau_0$. Because the free evolution only multiplies each momentum mode by a phase, the norm is conserved
\begin{equation}\label{eq:free-norm-conservation}
  \int Dh\,|\Psi_0[h,\tau]|^2
  =
  \int D\Pi\,|\widetilde{\Psi}_0[\Pi]|^2
  =
  1.
\end{equation}
The measures $Dh$ and $D\Pi$ and all determinants above are understood with the same gauge-fixed physical measure that defines the functional Laplacian, block by block on the projected modes as in Remark \ref{rem:field-space-extension}. In every normalized quantity the determinants combine into expressions of the Fredholm type $\det(\mathbf 1 + 4iT_\tau B\bar{\mathbb{G}})$, which converge when $B$ is trace class while the remaining formal factors are state independent and cancel.

We can now recover the curvature term if we split the full generator as
\begin{equation}
  \hat{H}_\parallel = \hat{H}_0 + \hat{H}_R,
\end{equation}
where
\begin{equation}
  \hat{H}_0
  \equiv
  \frac{16\pi G\hbar}{c^2}
  \int_{\Sigma_\tau}d\Sigma\,
  \bar{\mathbb{G}}_{\mu\nu\sigma\rho}
  \frac{\delta^2_\perp}{\delta h_{\mu\nu}\delta h_{\rho\sigma}},
\end{equation}
and
\begin{equation}
  \hat{H}_R
  =
  \int_{\Sigma_\tau}d\Sigma\,
  \frac{c^4}{16\pi G\hbar}
  R^\perp[\bar{\gamma}+h].
\end{equation}
The exact interacting wavepacket may be written in the interaction picture using the Dyson expansion \cite{Jackiw1990Analysis,Dyson1949SMatrix} as
\begin{equation}\label{eq:interaction-picture-wavepacket}
\begin{aligned}
  \Psi[h,\tau]
  ={}&
  e^{i(\tau-\tau_0)\hat{H}_0}
  \mathcal{T}
  \\
  &\times
  \exp\!\left[
    i\int_{\tau_0}^{\tau}d\tau'\,
    \hat{H}_{R,I}(\tau')
  \right]
  \Psi[h,\tau_0],
\end{aligned}
\end{equation}
with $\mathcal{T}$ the proper-time-ordering and
\begin{equation}
  \hat{H}_{R,I}(\tau)
  =
  e^{-i(\tau-\tau_0)\hat{H}_0}
  \hat{H}_R
  e^{i(\tau-\tau_0)\hat{H}_0},
\end{equation}
which can be expanded into a Dyson series. Thus Eq. (\ref{eq:general-free-wavepacket}) gives the exact free wavepacket basis, while Eq. (\ref{eq:interaction-picture-wavepacket}) gives the exact curvature-corrected wavepacket as an interaction-picture evolution.

\subsection{Exact solution for cosmological quantum evolution}
\label{sec:flat-flrw-sector}

We now solve Eq. (\ref{eq:QG}) in the spatially flat FLRW minisuperspace
sector, following the symmetry-reduction logic of quantum cosmology
\cite{HartleHawking1983WaveFunction,Misner1969QuantumCosmology,Kiefer2012QuantumGravity,Halliwell1990QuantumCosmology}.
Quantum fluctuations would produce local deviations from homogeneity, but
at cosmological scales these are small and we neglect them. Choosing the
comoving observer congruence adapted to the homogeneous slices,

\begin{equation}
  \gamma_{\mu\nu}(\lambda)
  =
  a^2(\lambda)\,\bar{\gamma}_{\mu\nu},
  \qquad
  \bar{\gamma}_{\mu\nu}
  \equiv
  n_\mu n_\nu-\eta_{\mu\nu}.
  \label{eq:flrw-projected-metric}
\end{equation}
For this adapted congruence one also has
$\gamma_\Sigma = a^6\bar{\gamma}_\Sigma$ and
$\sqrt{\gamma_\Sigma} = a^3\sqrt{\bar{\gamma}_\Sigma}$.
Now since the observer-spatial leaves are flat we can find
$R^\perp[\gamma]=0$, which greatly simplifies our calculations.
Now, if a cosmological constant is included in the
Einstein--Hilbert action the reduced Hamiltonian receives the
multiplicative term
\begin{equation}
  H_\Lambda
  =
  2\alpha\Lambda V_0 a^3,
  \qquad
  V_0\equiv\int_{\Sigma_\lambda}d^3x\sqrt{\bar{\gamma}_\Sigma},
  \label{eq:flrw-lambda-hamiltonian}
\end{equation}
and the gravitational kinetic term appears as
\begin{equation}
  K_{\mu\nu}K^{\mu\nu}-K^2
  =
  -6\left(\frac{a'}{a}\right)^2,
\end{equation}
where a prime denotes $d/d\lambda=d/d(c\tau)$. Thus the Lagrangian
becomes
\begin{equation}
  L_{\rm FLRW}
  =
  -6\alpha V_0aa'^2
  -
  2\alpha\Lambda V_0a^3,
  \label{eq:flrw-reduced-lagrangian}
\end{equation}
up to the sign convention for the cosmological term. With the convention
used in Eq. (\ref{eq:flrw-lambda-hamiltonian}), this gives the reduced
Hamiltonian
\begin{equation}
  H_{\rm FLRW}
  =
  -\frac{p_a^2}{24\alpha V_0a}
  +
  2\alpha\Lambda V_0a^3,
  \label{eq:flrw-reduced-hamiltonian-a}
\end{equation}
and if we define the canonical variable
\begin{equation}
  u
  =
  a^{3/2}
  \sqrt{\frac{c^3V_0}{3\pi G}}
  =
  \frac{4}{3}\sqrt{3\alpha V_0}\,a^{3/2},
  \label{eq:flrw-u-variable}
\end{equation}
then the corresponding momentum can be checked to satisfy
\begin{equation}
  p_u
  =
  \frac{p_a}{2\sqrt{3\alpha V_0a}},
  \qquad
  p_a
  =
  2\sqrt{3\alpha V_0a}\,p_u,
\end{equation}
so Eq. (\ref{eq:flrw-reduced-hamiltonian-a}) becomes
\begin{equation}
  H_{\rm FLRW}
  =
  -\frac{1}{2}p_u^2
  +
  \frac{3\Lambda}{8}u^2.
  \label{eq:flrw-reduced-hamiltonian-u}
\end{equation}

\noindent Notoriously, this is the Hamiltonian of an inverted harmonic oscillator.
This becomes clearer if we define $\Omega^2 \equiv 3\Lambda/4$ and
perform the canonical transformation $Q = -p_u / \Omega$ and
$P = \Omega u$. Then

\begin{equation}
  H = \frac{1}{2}P^2 - \frac{1}{2}\Omega^2 Q^2,
\end{equation}

\noindent which has a flipped sign in the potential term. This is known
\cite{kumar2026saddlepoint,Barton1986InvertedOscillator} to produce a continuous energy spectrum
which is unbounded from below and from above; i.e., its spectrum
covers all real values. We interpret this result as the
ability of the cosmological Hamiltonian to cancel the matter Hamiltonian,
(or that total energy is null $E = 0$) for reasons that shall become clear
later. For now, we can proceed to quantization thanks to
Theorem \ref{thm:ordering}; here the reduced
configuration-space metric in the variable $u$ is constant
with inverse metric $G^{uu}=-1$, so the Laplace--Beltrami operator
reduces to the flat second derivative. Covariance guarantees that
quantizing in the original variable $a$ and passing to $u$ afterwards
yields the same operator, and by Corollary \ref{cor:xi-ambiguity}
minimality leaves no residual scalar potential, so the quantum
kinetic Hamiltonian becomes

\begin{equation}
  \hat H_{\rm kin}
  =
  -\frac{\hbar^2}{2}\Delta_{\rm FLRW}
  =
  \frac{\hbar^2}{2}\frac{\partial^2}{\partial u^2},
  \label{eq:flrw-lb-operator-u}
\end{equation}
and the full equation is given by
\begin{equation}
  \frac{\partial}{\partial\tau}\Psi(u,\tau)
  =
  -\frac{i\hbar c}{2}
  \frac{\partial^2\Psi}{\partial u^2}
  -
  \frac{3i\Lambda c}{8\hbar}u^2\Psi.
  \label{eq:flrw-quantum-wave-eq}
\end{equation}
For $\Lambda>0$ a general solution of this equation
on the full line is given by the propagator
\begin{equation}
\begin{aligned}
  K_\Lambda(u,u_0;\Delta\tau)
  ={}&
  \sqrt{
    \frac{i\Omega}{2\pi\hbar\sinh(\Omega c\Delta\tau)}
  }
  \\
  &\times
  \exp\left[
    -\frac{i\Omega}{2\hbar}
    \frac{1}{\sinh(\Omega c\Delta\tau)}
    \right.\\
  &\qquad\left.
    {}\times
    \left(
      u^2+u_0^2
    \right)
    \cosh(\Omega c\Delta\tau)
    \right.\\
  &\qquad\left.
    {}
    -
    2uu_0
  \right],
\end{aligned}
  \label{eq:flrw-propagator-full-line}
\end{equation}
where $\Delta\tau=\tau-\tau_0$ and the
square-root branch is chosen continuously from $\Delta\tau=0$.
But since the cosmological scale factor satisfies $a\geq0$, the physical
configuration space is the half-line $u\geq0$. A self-adjoint boundary condition at $u=0$ must therefore be imposed (for more details one can also see \cite{ReedSimon1975Vol2,BonneauFarautValent2001SelfAdjoint,KunstatterLouko2012HalfLine}). For the DeWitt condition
$\Psi(0,\tau)=0$ \cite{DeWitt1967Quantum,Kiefer2012QuantumGravity,Halliwell1990QuantumCosmology}
the half-line propagator is obtained by the image method,

\begin{equation}
\begin{aligned}
  K_\Lambda^{\rm D}(u,u_0;\Delta\tau)
  ={}&
  K_\Lambda(u,u_0;\Delta\tau)
  \\
  &
  -
  K_\Lambda(u,-u_0;\Delta\tau),
  \\
  u,u_0>0,
\end{aligned}
  \label{eq:flrw-dirichlet-propagator}
\end{equation}
and hence the solution is
\begin{equation}
  \Psi(u,\lambda)
  =
  \int_0^\infty du_0\,
  K_\Lambda^{\rm D}(u,u_0;\lambda-\lambda_0)\Psi(u_0,\lambda_0).
  \label{eq:flrw-exact-solution-half-line}
\end{equation}
Other self-adjoint extensions correspond to other boundary conditions at
$u=0$ and define different quantum cosmologies.

Anyway, it might be more interesting to study the Heisenberg equations
generated by Eq. (\ref{eq:flrw-quantum-wave-eq}). These are
\begin{equation}
  \hat u'
  =
  -\hat p_u,
  \qquad
  \hat p_u'
  =
  -\frac{3\Lambda}{4}\hat u,
  \qquad
  \hat u''
  =
  \frac{3\Lambda}{4}\hat u,
  \label{eq:flrw-heisenberg-u}
\end{equation}
with exact solutions given by
\begin{equation}
  \begin{split}
    \hat u(\tau)
    &=
    \hat u_0\cosh(\Omega c\Delta\tau)
    -
    \frac{\hat p_{u,0}}{\Omega}\sinh(\Omega c\Delta\tau), \\
    \hat p_u(\tau)
    &=
    \hat p_{u,0}\cosh(\Omega c\Delta\tau)
    -
    \Omega\hat u_0\sinh(\Omega c\Delta\tau).
  \end{split}
\end{equation}
Now, to derive the quantum-corrected Friedmann equations, define
the moments
\begin{equation}
  S\equiv\langle\hat u^2\rangle,
  \qquad
  R\equiv\langle\hat p_u^{\,2}\rangle,
  \qquad
  T\equiv
  \frac{1}{2}
  \langle\hat u\hat p_u+\hat p_u\hat u\rangle,
  \label{eq:flrw-moments}
\end{equation}
such that the conserved energy is
\begin{equation}
  E
  \equiv
  \langle\hat H_{\rm FLRW}\rangle
  =
  -\frac{1}{2}R
  +
  \frac{3\Lambda}{8}S,
  \label{eq:flrw-energy}
\end{equation}
and the moment equations are
\begin{equation}
  S'=-2T,
  \qquad
  T'=-R-\frac{3\Lambda}{4}S,
  \qquad
  R'=-\frac{3\Lambda}{2}T.
  \label{eq:flrw-moment-equations}
\end{equation}
Therefore
\begin{equation}
  I_q
  \equiv
  SR-T^2
  \label{eq:flrw-quantum-invariant}
\end{equation}
is conserved; i.e., $I_q'=0$ identically. Moreover, the Robertson--Schrödinger inequality \cite{Robertson1929Uncertainty,Schrodinger1930Uncertainty} gives
\begin{equation}
  I_q=SR-T^2\geq\frac{\hbar^2}{4},
  \label{eq:flrw-robertson-bound}
\end{equation}
which shall become useful later. We can now use
Eq. (\ref{eq:flrw-energy}) and Eq. (\ref{eq:flrw-quantum-invariant})
to obtain
\begin{equation}
  (S')^2
  =
  3\Lambda S^2
  -
  8ES
  -
  4I_q,
  \qquad
  S''
  =
  3\Lambda S
  -
  4E,
  \label{eq:flrw-S-equations}
\end{equation}
and then if we define the effective scale factor by
\begin{equation}
  a_{\rm eff}^3
  \equiv
  \frac{3S}{16\alpha V_0},
  \label{eq:flrw-effective-scale-factor}
\end{equation}
we can divide Eq. (\ref{eq:flrw-S-equations}) by the appropriate powers
of $S$ to get
\begin{equation}
  H_{\rm eff}^2
  \equiv
  \left(
    \frac{a_{\rm eff}'}{a_{\rm eff}}
  \right)^2
  =
  \frac{\Lambda}{3}
  -
  \frac{E}{6\alpha V_0a_{\rm eff}^3}
  -
  \frac{I_q}{64\alpha^2V_0^2a_{\rm eff}^6},
  \label{eq:flrw-corrected-friedmann}
\end{equation}
and
\begin{equation}
  \frac{a_{\rm eff}''}{a_{\rm eff}}
  =
  \frac{\Lambda}{3}
  +
  \frac{E}{12\alpha V_0a_{\rm eff}^3}
  +
  \frac{I_q}{32\alpha^2V_0^2a_{\rm eff}^6},
  \label{eq:flrw-corrected-acceleration}
\end{equation}
so the quantum invariant $I_q$ produces a repulsive
$a_{\rm eff}^{-6}$ correction. For $\Lambda>0$, the bounce occurs when
$H_{\rm eff}=a'_{\rm eff}=0$, giving
\begin{equation}
  V_{\rm eff,min}
  \equiv
  V_0a_{\rm eff,min}^3
  =
  \frac{
    E+
    \sqrt{E^2+\frac{3\Lambda}{4}I_q}
  }{
    4\alpha\Lambda
  },
  \label{eq:flrw-min-volume-general}
\end{equation}
where by using Eq. (\ref{eq:flrw-robertson-bound}), a zero-energy packet
satisfies
\begin{equation}
  V_0a_{\rm eff,min}^3
  \geq
  \frac{\sqrt{3}\,\pi\hbar G}{c^3\sqrt{\Lambda}}.
  \label{eq:flrw-min-volume-bound}
\end{equation}
Therefore the effective FLRW volume cannot reach zero unless the
quantum invariant is removed by the classical limit, which
resolves the FLRW singularity. Note that the value of the curvature
parameter $k\in\left\{-1, 0, 1\right\}$ does not matter here, as we
can check to always get $V_\text{eff,min} > 0$ except for
$\hbar \longrightarrow 0$.

Consider now the sign of $E$ in Eq.
(\ref{eq:flrw-corrected-friedmann}). The energy enters as an effective
density $\rho_E \propto -E$, so a positive $\rho_E$ requires $E \leq 0$,
while $E > 0$ makes the equation behave as if the total energy were negative.
Avoiding both exotic situations demands $E = 0$. The case becomes stronger
once matter is added, since the quantum-corrected equation then differs from
the classical one by the arguably small $I_q$ and by $E$, which nothing
prevents from being large; keeping the dynamics close to the observed
Friedmann equations \cite{Planck2018Parameters,EfstathiouGratton2020Flatness}
therefore forces $E$ to be very small, and for our purposes we set $E = 0$,
in agreement with the usual approach to canonical cosmology
\cite{DeWitt1967Quantum,Halliwell1990QuantumCosmology}. This justifies the
choice made in Eq. (\ref{eq:flrw-min-volume-bound}).

Let us now place this resolution against the two main treatments of the same
minisuperspace in the literature. In the Wheeler--DeWitt quantization the
reduced model obeys the timeless kernel equation, so singularity avoidance
must be encoded in the choice of wave function or of boundary condition on
the lines of the DeWitt condition $\Psi(0)=0$ or of the no-boundary and
tunneling proposals
\cite{DeWitt1967Quantum,HartleHawking1983WaveFunction,Kiefer2012QuantumGravity,Halliwell1990QuantumCosmology},
and the interpretation of the avoidance depends on an inner product that the
kernel equation does not itself supply. Indeed, when the model is
deparametrized with a massless scalar clock the Wheeler--DeWitt wavepackets
are known to remain peaked on the classical trajectory all the way into the
big bang \cite{AshtekarPawlowskiSingh2006Improved}. In loop quantum cosmology
the singularity is resolved by a bounce with the effective Friedmann
equation acquiring the correction factor $(1-\rho/\rho_c)$, but the mechanism
relies on the polymer representation of the holonomy algebra and imports the
critical density $\rho_c$ from the area gap of the full theory
\cite{AshtekarSingh2011LQC,AshtekarPawlowskiSingh2006Improved,Bojowald2001Absence}.

The bounce of Eq. (\ref{eq:flrw-corrected-friedmann}) has a different origin.
The quantization is an ordinary Schrödinger one on the half-line, unitary in
the proper time of the comoving congruence, and the repulsive
$a_{\rm eff}^{-6}$ term is carried by the conserved invariant $I_q$ which
Eq. (\ref{eq:flrw-robertson-bound}) bounds away from zero on every state. The
mechanism needs no polymer representation and no special boundary condition
at $a=0$. The moment system of Eq. (\ref{eq:flrw-moment-equations}) closes
exactly for this Hamiltonian, so Eqs.
(\ref{eq:flrw-corrected-friedmann}) and
(\ref{eq:flrw-corrected-acceleration}) are exact statements about expectation
values rather than a semiclassical truncation. What makes the resolution
available here is the ingredient the timeless quantization lacks; namely,
a unitary proper-time flow along which the uncertainty-bounded
invariant is conserved.

\subsection{Boundary quantum evolution for the Schwarzschild--Lemaître sector}
\label{sec:schwarzschild-lemaitre-boundary-sector}

We now study the Schwarzschild geometry in the same observer-congruent
formalism, again assuming exact spherical symmetry and reducing the degrees
of freedom to the radial function $r$. The Schwarzschild metric can be put in
Painlevé--Gullstrand (PG) coordinates or Lemaître form
\cite{Lemaitre1933Expansion,MisnerThorneWheeler1973} as

\begin{equation}
  ds^2
  =
  d\lambda^2
  -
  \frac{r_s}{r(\lambda,\rho)}\,d\rho^2
  -
  r^2(\lambda,\rho)d\Omega^2,
  \label{eq:schw-lemaitre-metric}
\end{equation}

\noindent where $r_s$ is the Schwarzschild radius and

\begin{equation}\label{def:eff-r}
  r^{3}(\lambda,\rho)
  =
  \frac{9}{4}r_s\,(\rho-\lambda)^2,
\end{equation}

\noindent is the effective radius. This adapted
time coordinate follows freely falling observers and the spatial slices are regular at the horizon; and we have chosen these PG coordinates to free the observer-congruence from $r$ dependence, as in this coordinate system freely falling observers can be taken to be $n^\mu = (1, 0, 0, 0) = n_\mu$.

Here we take the black hole mass $\hat{M}$ as the quantum operator. Matter
inside the black hole can sit in a superposition of energy states, so with
$Mc^2 = E$ the Schwarzschild radius $\hat{r}_s$ becomes an operator; this is
also common practice in reduced spherically symmetric geometrodynamics,
where the mass is the surviving physical observable after solving the
constraints \cite{Kuchar1994Schwarzschild,LoukoWhiting1995HamiltonianThermo}.
The effective radius $\hat{r}^3(\lambda,\rho) \propto \hat{r}_s$ then becomes
an operator too, and with it the volume and area, as it should be if the size
of the black hole is to reflect the superposition of masses.

Furthermore, since quantum particles can be off-shell and thus incorporate energies that do not satisfy the on-shell conditions we should also allow for this $r(\lambda, \rho)$ to be off-shell and thus treat it as our degree of freedom. However, if this is the case then $R^\perp$ no longer identically vanishes as in the Schwarzschild--Lemaître case, which hardens the computation of the Hamiltonian. To work around this we can take advantage of the fact that Eq. (\ref{eq:schw-lemaitre-metric}) can also be expressed as

\begin{equation}\label{eq:SchLemOffShell}
  ds^2 = d\lambda^2 - r^2_\rho(\lambda, \rho) d\rho^2 - r^2(\lambda, \rho)d\Omega^2,
\end{equation}

\noindent with $r_\rho \equiv \partial_\rho r$. Now with this form of the metric we see that at every constant $\lambda$ slice the metric is spatially flat. This becomes clear if we make the coordinate change $dR = r_\rho d\rho$ to get an explicitly flat metric. Then $R^\perp = 0$ identically at all slices, which helps us simplify the calculation of the Lagrangian.

We still have a problem, though. On-shell the bulk term vanishes because $R=0$. That would look like it has no dynamics, because from

\begin{equation}
  K_{\mu\nu} = \frac{1}{2}\pounds_n \gamma_{\mu\nu} = \frac{1}{2}\partial_\lambda\gamma_{\mu\nu},
\end{equation}

\noindent one can easily compute the bulk term using the definition of $r$ given by Eq. (\ref{def:eff-r}) to get

\begin{equation}
  K_{\mu\nu}K^{\mu\nu} - K^2 = 0,
\end{equation}

\noindent which consequently gives $H^\text{CG}_{\parallel} = 0$. This is also standard in the ADM formalism \cite{Arnowitt1962Dynamics,ReggeTeitelboim1974SurfaceIntegrals,Kuchar1994Schwarzschild}. But of course the total energy in spacetime should be proportional to $M$, and the way the literature has to include that energy is through a boundary term \cite{Arnowitt1962Dynamics,GibbonsHawking1977Action,BrownYork1993Quasilocal,ReggeTeitelboim1974SurfaceIntegrals,Szabados2009Quasilocal}.

We shall find this boundary term as one that vanishes on-shell but produces the dynamics we desire off-shell. For that let us define the boundary $S_\lambda \equiv \partial\Sigma_\lambda$ to be a round two-sphere at fixed radius $\rho = \rho_B$ in the spherically symmetric vacuum region at some time $\lambda$. The location $\rho_B$ of the round boundary sphere may be chosen
arbitrarily inside the vacuum Schwarzschild region, which means that the boundary can be an arbitrary round sphere compatible with the spherical reduction. In the asymptotically
flat case one may take this sphere to spatial infinity, where the same
boundary charge becomes the ADM energy \cite{Arnowitt1962Dynamics,BrownYork1993Quasilocal,ReggeTeitelboim1974SurfaceIntegrals}. If we then use this definition plus the off-shell Eq. (\ref{eq:SchLemOffShell}) we find
\begin{equation}\label{eq:LSchbulk}
  \begin{split}
    L^\text{CG}_{\parallel} &= \alpha\int_{\Sigma_\lambda}{d\Sigma\left(K_{\mu\nu}K^{\mu\nu} - K^2\right)} \\
    &= 4\pi\alpha\int_{\Sigma_\lambda}{d\rho r_\rho r^2\left(-4\frac{r_{\lambda\rho}}{r_\rho}\frac{r_\lambda}{r} - 2\frac{r^2_\lambda}{r^2}\right)} \\
    &= 4\pi\alpha\int_{\Sigma_\lambda}{d\rho \left(-4rr_\lambda r_{\lambda\rho} - 2r_\rho r^2_\lambda \right)},
  \end{split}
\end{equation}
with $r_\lambda \equiv \partial_\lambda r$, $r_\rho \equiv \partial_\rho r$ and $r_{\lambda\rho} \equiv \partial_\lambda\partial_\rho r$. Now note that
\begin{equation}
  -4rr_\lambda r_{\lambda\rho} = -2r\partial_\rho(r^2_\lambda) = -2\partial_\rho(rr^2_\lambda) + 2r_\rho r^2_\lambda,
\end{equation}

\noindent where the last term cancels the last one in Eq. (\ref{eq:LSchbulk}) thus leaving

\begin{equation}\label{eq:bulkL}
    L^\text{CG}_{\parallel} = -8\pi\alpha\int_{\Sigma_\lambda}{d\rho \partial_\rho(r r^2_\lambda)},
\end{equation}

\noindent which is a total derivative that vanishes on-shell. This is our boundary term

\begin{equation}
  -8\pi\alpha\left[r r^2_\lambda\right]_{S},
\end{equation}

\noindent and the Hamiltonian is consequently

\begin{equation}
  H^{\text{CG}}_{\parallel} = \frac{c^3}{2G}r r^2_\lambda \underbrace{=}_{rr^2_\lambda = r_s} Mc,
\end{equation}

\noindent so the physical energy is $E = cH^{\text{CG}}_{\parallel} = Mc^2$ as expected. Note how the total derivative has no local radial dynamics but it does define a boundary charge which ends up giving us the $Mc$ we require to find the total energy.
Now if we express Eq. (\ref{eq:bulkL}) in observer-congruent form it gives

\begin{equation}
  H^\text{CG}_{\parallel}[S_\lambda] = 2\alpha \int_{S_\lambda}{r^{-1}(\pounds_n r)^2dS},
\end{equation}

\noindent which after quantization with $r' \equiv dr / dc\tau$ we get

\begin{equation}
  H_r = \frac{p_r^2}{32\pi\alpha r}, \qquad p_r = \frac{\partial L^\text{CG}_\parallel}{\partial r'} = 16\pi\alpha rr',
\end{equation}

\noindent also, after the radial total-derivative reduction the boundary minisuperspace is given by the areal radius of the chosen boundary two-sphere rather than all $r(\lambda, \rho)$; i.e., $R(\lambda)\equiv r(\lambda,\rho_B)>0$. Similar treatments have been done in the literature for spherically symmetric black holes \cite{Kuchar1994Schwarzschild,LoukoWhiting1995HamiltonianThermo}.
Thus the boundary configuration space is $\mathcal C_{\partial}=\{R>0\}$ and the observable $R$ determines the intrinsic geometry of the boundary sphere
\begin{equation}
  d\sigma^2=R^2d\Omega^2,
  \qquad
  A[S_\lambda]=4\pi R^2.
\end{equation}
By itself, this $R$ does not determine a full four-dimensional geometry. The full Schwarzschild--Lemaître geometry is recovered only on shell or semiclassically when the boundary state is peaked on a mass $M$ so that $RR'^2=r_s$. Thus the boundary minisuperspace represents the areal radius degree of freedom at the boundary for a spacetime with line element given by Eq. (\ref{eq:SchLemOffShell}) and containing Schwarzschild for a peaked mass around $M$. Also, if we take $\rho_B$ at the expectation value of $\hat{r}_s$ then we could interpret this $R$ as the effective areal radius of the black hole and, in the quantum picture, as the degree of freedom that allows us to study deviations around the classical areal radius.

Now if we relabel $R \rightarrow r$ again then we can put the Hamiltonian in better shape if we perform a canonical transformation

\begin{equation}
  v \equiv \frac{2}{3}\sqrt{16\pi\alpha} r^{3/2} = \kappa r^{3/2},
  \qquad
  \kappa \equiv \frac{2}{3}\sqrt{\frac{c^3}{G}},
\end{equation}

\noindent which gives the much cleaner form for the Hamiltonian

\begin{equation}
  H_r = \frac{1}{2}p^2_v,
\end{equation}
and then Eq. (\ref{eq:QG}) is the free half-line equation
\begin{equation}
  i\hbar\frac{\partial\psi}{\partial(c\tau)}
  =
  -\frac{\hbar^2}{2}
  \frac{\partial^2\psi}{\partial v^2},
  \qquad
  v>0.
  \label{eq:schw-boundary-schrodinger-v}
\end{equation}

\noindent The half-line condition $v>0$ is essential because $v=0$ is the
$r=0$ case. For the Hamiltonian $-\hbar^2\partial_v^2/2$ to be
self-adjoint we have to choose a boundary condition at $v=0$ \cite{ReedSimon1975Vol2,BonneauFarautValent2001SelfAdjoint,KunstatterLouko2012HalfLine},
which for regularization \cite{DeWitt1967Quantum,BrahmaChenYeom2022DeWittBoundary} must be taken to be
\begin{equation}
  \psi(0, \tau)=0,
  \label{eq:schw-dewitt-boundary-condition}
\end{equation}

\noindent whose solutions have the form

\begin{equation}
  \psi(v, \tau) =
  \int^\infty_0{dpA(p)\sin\left(\frac{pv}{\hbar}\right)
  \exp\left[\frac{-i}{2\hbar}p^2c\tau\right]},
\end{equation}

\noindent with

\begin{equation}
  A(p) \equiv \frac{2}{\pi\hbar}\int^\infty_0{
    dv\psi(v, 0)\sin\left(\frac{pv}{\hbar}\right)
  }.
\end{equation}

\noindent However, the more interesting physics appear if we study the
Heisenberg equations in the $v$ variable. These are $\hat v' = \hat p_v$, $\hat p_v' = 0$ and $\hat v'' = 0$; and since $\hat r = (\hat{v}/\kappa)^{2/3}$, the corresponding Heisenberg equation for $\hat r$ is most cleanly
written in symmetrised form
\begin{equation}
  \begin{aligned}
    \hat r'
    &=
    \frac{1}{3\kappa}
    \left\{
      \hat r^{-1/2},
      \hat p_v
    \right\},
    \\
    \hat r''
    &=
    -\frac{1}{18\kappa^2}
    \left\{
      \left\{
        \hat r^{-2},
        \hat p_v
      \right\},
      \hat p_v
    \right\},
  \end{aligned}
  \label{eq:r-heisenberg}
\end{equation}
and we can check the classical limit of Eq. (\ref{eq:r-heisenberg})
to be
\begin{equation}
  r''
  =
  -\frac{1}{2}\frac{r'^2}{r}
  =
  -\frac{H_r}{16\pi\alpha r^2},
  \label{eq:r-classical-acceleration}
\end{equation}
which for $H_r=Mc$ this becomes $r'' = -GM/c^2r^2$, with derivatives taken with respect to $c\tau$. Now define the moments
\begin{equation}
  S
  \equiv
  \langle \hat v^2\rangle,
  \qquad
  R
  \equiv
  \langle \hat p_v^2\rangle,
  \qquad
  T
  \equiv
  \frac{1}{2}
  \left\langle
    \hat v\hat p_v+\hat p_v\hat v
  \right\rangle.
  \label{eq:schw-moments}
\end{equation}
Because the Hamiltonian is free, then $E_r \equiv \langle\hat H_r\rangle = R/2$ is conserved. The moment equations are
\begin{equation}\label{eq:BHSTREqs}
  S'=2T,
  \qquad
  T'=R=2E_r,
  \qquad
  R'=0,
\end{equation}
thus $I_r \equiv SR-T^2$ is exactly conserved. Thus we get
\begin{equation}
  (S')^2
  =
  8E_rS-4I_r,
  \qquad
  S''
  =
  4E_r,
  \label{eq:schw-S-equations}
\end{equation}
and if we now define the effective areal radius by
\begin{equation}
  r_{\rm eff}^3
  \equiv
  \frac{S}{\kappa^2}
  =
  \frac{9G}{4c^3}S,
  \label{eq:reff-definition}
\end{equation}
then dividing Eq. (\ref{eq:schw-S-equations}) by the appropriate powers of
$S$ gives the quantum-corrected Schwarzschild--Lemaître radial equations
\begin{equation}
  \left(
    \frac{r_{\rm eff}'}{r_{\rm eff}}
  \right)^2
  =
  \frac{2GE_r}{c^3r_{\rm eff}^3}
  -
  \frac{9G^2I_r}{4c^6r_{\rm eff}^6},
  \label{eq:schw-corrected-r-friedmann}
\end{equation}
and
\begin{equation}
  \frac{r_{\rm eff}''}{r_{\rm eff}}
  =
  -
  \frac{GE_r}{c^3r_{\rm eff}^3}
  +
  \frac{9G^2I_r}{2c^6r_{\rm eff}^6}.
  \label{eq:schw-corrected-r-acceleration}
\end{equation}
For a Schwarzschild boundary state peaked on mass $M$ we have
$E_r = Mc$, so Eq. (\ref{eq:schw-corrected-r-friedmann}) becomes
\begin{equation}
  \left(
    \frac{r_{\rm eff}'}{r_{\rm eff}}
  \right)^2
  =
  \frac{r_s}{r_{\rm eff}^3}
  -
  \frac{9G^2I_r}{4c^6r_{\rm eff}^6}.
  \label{eq:schw-corrected-r-friedmann-M}
\end{equation}
The first term is the classical Lemaître infall relation. The
second is the quantum term. It has the opposite sign and therefore
prevents the effective radius from reaching zero.
Indeed, the minimum radius occurs when $r_{\rm eff}'=0$,
which from Eq. (\ref{eq:schw-S-equations}) can be found to be
\begin{equation}
  r_{\rm eff,min}^3
  =
  \frac{I_r}{2E_r\kappa^2}
  =
  \frac{9GI_r}{8c^3E_r}.
  \label{eq:schw-reff,min}
\end{equation}
For a mass eigenpacket with $E_r=Mc$ this becomes, making
use of the Robertson--Schrödinger inequality for $I_r$
\begin{equation}
  r_{\rm eff,min}^3
  \geq
  \frac{9G\hbar^2}{32Mc^4},
  \label{eq:schw-reff,min-bound}
\end{equation}
which says that the effective Schwarzschild--Lemaître radius cannot
reach zero unless the quantum invariant $I_r$ is sent to zero. This
is the same as to demand the classical limit where $\hbar \longrightarrow 0$,
thus recovering the classical singularity.

Now, interestingly, the DeWitt boundary condition we demanded earlier strengthens this bound to be a bit larger than that of Eq. (\ref{eq:schw-reff,min-bound}). To see why we must work within the framework of conformal algebra. Using de Alfaro, Fubini and Furlan (DFF) \cite{deAlfaroFubiniFurlan1976Conformal} technique as it is used in near-horizon black hole physics \cite{ClausEtAl1998SuperconformalBH,BrittoPacumioEtAl1999SCQMLectures,CamblongOrdonez2005NearHorizonCQM} we can build an eigenbases set $\{\ket{n}\}$ that defines the discrete spectrum of the DFF compact generator $L_0$. To see how, let us define the dilation and special conformal generators
\begin{equation}
  H = \frac{1}{2}\hat{p}_v^2,
  \qquad
  D = \frac{1}{4}\left(\hat{v}\hat{p}_v + \hat{p}_v\hat{v}\right),
  \qquad
  K = \frac{1}{2}\hat{v}^2,
  \label{eq:schw-hdk-generators}
\end{equation}

\noindent which can be checked to close the algebra

\begin{equation}
  \begin{aligned}
    [D,H]&=i\hbar H,
    \\
    [D,K]&=-i\hbar K,
    \\
    [H,K]&=-2i\hbar D,
  \end{aligned}
  \label{eq:schw-so21-algebra}
\end{equation}

\noindent which is $\mathfrak{so}(2,1) \cong \mathfrak{su}(1,1)$ \cite{deAlfaroFubiniFurlan1976Conformal}. Note that the three moments
of Eq. (\ref{eq:schw-moments}) are the expectation values of
these three generators, $S = 2\langle K\rangle$, $R = 2\langle H\rangle$
and $T = 2\langle D\rangle$, so the invariant $I_r$ is an object that
belongs to this algebra, whose quadratic Casimir given by

\begin{equation}
  \mathcal{C}
  =
  \frac{1}{2}\left(HK + KH\right) - D^2
  =
  -\frac{3}{16}\hbar^2
  \equiv
  \hbar^2 k(k - 1),
  \label{eq:schw-so21-casimir}
\end{equation}

\noindent where the numerical value $-3\hbar^2/16$ follows from direct
computation in the realization of Eq. (\ref{eq:schw-hdk-generators}).
The Bargmann index $k$ comes from the fact that the unitary irreducible representations
of $\mathfrak{su}(1,1)$ which are bounded from below (the discrete series
$D^+(k)$) are labelled by a single number $k > 0$, which is the lowest
eigenvalue in units of $\hbar$ of the compact generator of the algebra.
Since Eq. (\ref{eq:schw-so21-casimir}) is quadratic in $k$ a fixed
Casimir is compatible with exactly two indices, in our case
$k \in \{1/4,\, 3/4\}$.

The algebra alone cannot decide between them,
because both roots correspond to legitimate unitary representations on
the half-line. Instead what selects the representation is the domain of the
operators; i.e., the self-adjoint boundary condition at $v = 0$
\cite{BonneauFarautValent2001SelfAdjoint,KunstatterLouko2012HalfLine}.
It is thus satisfying to see that $k$ is fixed by the same DeWitt condition
\cite{DeWitt1967Quantum} that we already imposed in
Eq. (\ref{eq:schw-dewitt-boundary-condition}) to define the quantum
theory in the first place.

To make use of this selection consider for an arbitrary frequency $\omega > 0$ the combination

\begin{equation}
  \begin{aligned}
    L_0
    &=
    \frac{1}{2}
    \left(\frac{H}{\omega}+\omega K\right),
    \\
    L_1
    &=
    \frac{1}{2}
    \left(\frac{H}{\omega}-\omega K\right),
    \\
    L_2&=D,
  \end{aligned}
  \label{eq:schw-compact-generator}
\end{equation}

\noindent where $L_0$ is the compact generator of the algebra; thus, up to a factor $2\omega$ it is the auxiliary oscillator Hamiltonian we lacked in the beginning (since $H_r$ describes a free particle), but now it enters as a self-adjoint
element of the algebra for \emph{every} $\omega$, while the free
Hamiltonian is simply the non-compact element $H = \omega(L_0 + L_1)$
of the same algebra. Since our observables are $\omega$-independent, choosing $\omega$ amounts to choosing a basis (a Bogoliubov frame) on the same Hilbert space but the dynamics stay the same.

Now, since the even oscillators would produce non-vanishing states for $v=0$ these are excluded by the DeWitt condition $\psi(0) = 0$, then only the odd states survive. And because the odd tower carries the oscillator energies $\hbar\omega(2n + 3/2)$ with $n \in \mathbb{Z}_{\geq 0}$ and $L_0$ is the oscillator Hamiltonian divided by $2\omega$, we get

\begin{equation}
  L_0\ket{n}
  =
  \hbar\left(n + \frac{3}{4}\right)\ket{n},
  \qquad
  n = 0, 1, 2, \dots,
  \label{eq:schw-L0-spectrum}
\end{equation}

\noindent so the DeWitt condition selects the representation
$D^+(k = 3/4)$. Had we chosen the Neumann condition $\psi'(0) = 0$
instead, the even states would survive and we would land on $k = 1/4$. Choosing one or the other amounts to choosing which of the two roots of
Eq. (\ref{eq:schw-so21-casimir}) and thus which of the two natural quantum
theories on the half-line we are going to use. Note also that, as promised, the eigenvalues
in Eq. (\ref{eq:schw-L0-spectrum}) carry no $\omega$ factor.

We can now evaluate the moments on the special discrete basis
$\ket{n}$. Using Eq. (\ref{eq:schw-compact-generator}) to write
$H = \omega(L_0 + L_1)$ and $K = (L_0 - L_1)/\omega$, and the fact that
$\langle L_1\rangle = \langle L_2\rangle = 0$ on $L_0$ eigenstates
(these operators are linear combinations of the ladder operators of the
representation), we find

\begin{equation}
  S = \frac{\hbar}{2\omega}(4n + 3),
  \qquad
  R = \frac{\hbar\omega}{2}(4n + 3),
  \qquad
  T = 0,
  \label{eq:schw-moments-eigenbasis}
\end{equation}

\noindent so the quantum invariant is quantized as

\begin{equation}
  I_r
  =
  SR - T^2
  =
  \frac{\hbar^2}{4}\left(4n + 3\right)^2
  =
  4\hbar^2\left(n + \frac{3}{4}\right)^2,
  \label{eq:schw-Ir-quantized}
\end{equation}

\noindent with all $\omega$-dependence cancelling as anticipated.
Evaluating Eq. (\ref{eq:schw-reff,min}) on these states with $E_r = Mc$
then gives

\begin{equation}
  r^3_\text{eff,min}
  =
  \frac{9G\hbar^2}{32Mc^4}\left(4n + 3\right)^2,
  \label{eq:schw-reff-quantized}
\end{equation}
thus we prove that when the DeWitt boundary applies, $r_\text{eff,min}$ strengthens to nine times the bound we previously had.

One could be tempted to interpret this result as a quantization of the effective radius, but that would be wrong. Note that $r^3_\text{eff,min}$ was derived under the condition $r'_\text{eff}=0$, so this quantized $r^3_\text{eff,min}$ \textit{only quantizes lower bounds}. This means that the spectrum of $\hat{r}$, of $\hat{v}$ and of the masses remain continuous but they are nonetheless bounded below by a conformal algebra, state-dependent number.

What does this bound mean? It gives a \textit{bounce} for any classical particle trajectory. For example, let us conceive the movement of a classical freely falling comoving particle. Let the bounce occur at some time $\lambda_b$ and let us define $r_b = r_\text{eff,min}$. Then from Eq.~\eqref{eq:BHSTREqs} since $R$ is constant we can integrate $T' = R$ and $S'=2T$ to obtain
\begin{equation}
  S(\lambda) = \frac{I_r}{R} + R(\lambda - \lambda_b)^2,
\end{equation}
where the condition $T(\lambda_b) = 0$ removes the constant of integration for $T$ at some bouncing instant $\lambda_b$. Now because of the definitions $r^3_\text{eff} = S/\kappa^2$ and $r^3_b = I_r/R\kappa^2$, we find
\begin{equation}
  r^3_\text{eff}(\lambda) = r^3_b + \frac{R}{\kappa^2}(\lambda - \lambda_b)^2,
\end{equation}
and thanks to having $R=2E_r=2Mc$ semiclassically, we see that the quantum corrected trajectory is
\begin{equation}\label{eq:schw-bounce-trajectory}
  r^3_\text{eff}(\lambda) = r^3_b + \frac{9}{4}r_s(\lambda - \lambda_b)^2,
\end{equation}
which recovers the classical one for $\hbar \longrightarrow 0$. So the quantization of $r^3_\text{eff,min}$ in $\{\ket{n}\}$ only applies to $r^3_b$ and not to $r^3_\text{eff}$, so the spectrum of $\hat{r}$ continues to be continuous even though the spectrum of its lower bound $r^3_b$ is not.

There is another interesting consequence. Let us rearrange the definition of the effective volume to put it into the more useful form $r^3_\text{eff}[\psi;\lambda] = \langle\hat{v}^2\rangle_\lambda/\kappa^2 = \langle\hat{r}^3\rangle_\lambda$. Now if we insert $\hat{r}$ into the formula for spherical volumes $\hat{V} = 4\pi\hat{r}^3/3$ we can also define $r^3_\text{eff}$ in terms of the volume operator as $r^3_\text{eff}[\psi;\lambda] = 3\langle\hat{V}\rangle_\lambda/4\pi$. Then $r^3_b$ is also the minimum volume that the black hole can have. To see this, let us demand compatibility with the semiclassical picture where the radius state is sharply peaked at some $\bar{r} \equiv \langle\hat{r}\rangle$; consequently $(\Delta r)^2 \equiv \langle(\hat{r} -\bar{r})^2\rangle$ is small and thus if we expand $r_\text{eff}$

\begin{equation}
  r_\text{eff} = \langle\hat{r}^3\rangle^{1/3} \simeq \bar{r} + \frac{(\Delta r)^2}{\bar{r}} + \cdots,
\end{equation}

\noindent then for the peaked state $(\Delta r)^2/\bar{r} \ll 1$ and thus $r_\text{eff} \simeq \bar{r}$. Where in the volume picture for $\bar{r} = r_s$ we would get $r_\text{eff}\simeq r_s$; and for the minimum value $r_\text{eff} \geq r_b$, then

\begin{equation}\label{eq:rsrb}
  r_s \geq r_b,
\end{equation}

\noindent or, in terms of $M$,

\begin{equation}
  M \geq \frac{3}{4}M_\text{P} = \frac{3}{4}\sqrt{\frac{\hbar c}{G}},
\end{equation}

\noindent thus a lower bound for the black hole mass. This might explain why we do not see black holes form out of every elementary particle (dots with no dimension). However, this does not imply that the mass spectrum is discrete; only the lower bound is discrete on the eigenbasis $\{\ket{n}\}$. There is in principle no requirement that further values of $M$ go in discrete jumps, and indeed from the continuous spectrum of $\hat{r}$ we might argue that the spectrum of $\hat{r}_s$ should be continuous too. In fact, Eq.~\eqref{eq:rsrb} can also be seen as the requirement that there is no ``naked bounce''; i.e., that particles can always get arbitrarily close to the event horizon from outside with no ``invisible wall'' that would make them bounce before touching the black hole and thus never being able to get inside.

These results can now be measured against the canonical literature on the
same sector. In Kucha\v{r}'s reduced quantization of the full Schwarzschild
theory the constraints are solved classically, the surviving phase space
consists of the mass together with its conjugate Killing-time separation,
the reduced Hamiltonian vanishes, and the quantum theory is a superselected
family of mass eigenstates carrying no interior dynamics
\cite{Kuchar1994Schwarzschild}. The Hamiltonian thermodynamics of Louko and
Whiting obtains genuine dynamics but only by enclosing the exterior in a box
whose wall data drive the evolution \cite{LoukoWhiting1995HamiltonianThermo}.
Contrary to these approaches our boundary theory presented here acquired
dynamics without any box, because $H_r$
generates proper-time evolution of the areal radius and the mass enters
through the boundary charge as the operator $\hat M$. In loop quantum
gravity the interior singularity is likewise replaced by a bounce obtained
through effective holonomy corrections in the Kantowski--Sachs sector
\cite{AshtekarOlmedoSingh2018QuantumExtension}; the analogous statement here
is the bounded trajectory of Eq. (\ref{eq:schw-bounce-trajectory}) which
follows from the conserved invariant $I_r$ and the conformal structure of
the half-line theory with no polymer input. Finally, the discrete tower of
Eq. (\ref{eq:schw-Ir-quantized}) should be contrasted with the quantized
horizon area proposed by Bekenstein and Mukhanov
\cite{Bekenstein1973Entropy,BekensteinMukhanov1995Spectroscopy}. Yet,
in the present theory the spectra of $\hat r$ and of the mass remain
continuous and no discrete line spectrum is implied for the emitted
radiation.

\subsection{Black hole thermodynamics}

One could ask what happens if evolution is studied along an angle rather
than along proper time; time plays no special role here, so we may pass from
a time-radial picture to an angle-area one. Black hole thermodynamics
motivates the move. The bound of Eq. (\ref{eq:schw-reff,min-bound}) depends
on the mass, so a fully evaporated black hole would have an infinite minimum
effective radius, which is difficult to interpret, and the literature
\cite{LoukoWhiting1995HamiltonianThermo,Carlip_1995,BanadosTeitelboimZanelli1994Entropy}
still lacks a complete thermodynamic description of these bodies at the
smallest scales. Carlip and Teitelboim conclude that ``we still lack a
microscopic explanation for the exponential weight in the integration
measure for the surface degrees of freedom, or equivalently for the
$\hbar^{-1}$ dependence in (22)''.

In what follows we show that the thermodynamics description of black
holes naturally arises from a couple of canonical transformations \cite{Carlip_1995,BanadosTeitelboimZanelli1994Entropy} that
allow us to study evolution along the boost angle rather than proper
time. The first of these transformations is given by

\begin{equation}
  Q \equiv \frac{v}{p_v},
  \qquad
  E \equiv H_r = \frac{1}{2}p^2_v,
\end{equation}

\noindent note that this canonical pair satisfies
$\left\{Q, E\right\} = 1$. Now let us define the second transformation

\begin{equation}
  \Theta \equiv \frac{Qc^3}{4GE},
  \qquad
  P_\Theta \equiv \frac{2G}{c^3}E^2,
\end{equation}

\noindent which also satisfies $\left\{\Theta, P_\Theta\right\} = 1$.

Now, two remarks are in order before quantizing. First, the map
$(v, p_v) \rightarrow (\Theta, P_\Theta)$ is a canonical chart on the
open dense region $p_v \neq 0$ of the same reduced phase space; it does
not define a new phase space nor, upon quantization, a new Hilbert
space. Canonical transformations need not be
unitarily implementable and in general relate different coordinatizations
of one and the same quantum theory rather than different quantum
theories \cite{Anderson1994CanonicalTransformations}. All operators
below therefore act on the single Hilbert space
$L^2(\mathbb{R}_+, dv)$ with the DeWitt condition of
Eq. (\ref{eq:schw-dewitt-boundary-condition}). Second, although the
compound variables $Q$ and $\Theta$ suffer from ordering ambiguities,
we shall only ever need the momentum $P_\Theta$, which upon quantization
is the unambiguous positive operator
$\hat{P}_\Theta = 2G\hat{H}_r^2/c^3$.

If we use these to go into the quantum picture we get
$[\hat{\Theta}, \hat{P}_\Theta] = i\hbar$
so, in the $\Theta$ representation, $\hat{P}_\Theta = -i\hbar \partial_\Theta$. But if we recall that
$E = H_r$ then we also have $\hat{P}_\Theta = 2G \hat{H}^2_r / c^3$.
Applying both to an eigenstate $\psi$ we get

\begin{equation}\label{eq:ThetaEq}
  -i\frac{\partial\psi}{\partial\Theta} = \frac{\hbar^3G}{2c^3}\frac{\partial^4}{\partial v^4} \psi, \qquad v>0,
\end{equation}

\noindent in addition to Eq. (\ref{eq:schw-boundary-schrodinger-v}).
Both equations share eigenstates and apply simultaneously as they model
the same dynamics with different labels. We interpret the right-hand
side of this equation as being proportional to the area operator. Note
that the area of the black hole is $A = 4\pi r^2_s$ and, given
$H_r = c^3 r_s / 2G$ on shell, we can express $A$ in terms of $H_r$
upon quantization as

\begin{equation}
  \hat{A} = \frac{4\pi G^2\hbar^4}{c^6}\frac{\partial^4}{\partial v^4},
\end{equation}

\noindent so we can rearrange Eq. (\ref{eq:ThetaEq})
into an expression using $\hat{A}$ in a way such that if we also
multiply both sides by $2\pi k_\text{B}$ with $k_\text{B}$ the Boltzmann
constant we get

\begin{equation}
  \hat{S}_\text{BH}\psi
  =
  \frac{k_\text{B}c^3}{4G\hbar}\hat{A}\psi,
  \qquad
  \hat{S}_\text{BH} \equiv -2\pi ik_\text{B}\frac{\partial}{\partial \Theta},
\end{equation}

\noindent with $\hat{S}_\text{BH}$ the entropy operator for the black
hole.

The factor $2\pi$ in the definition of $\hat{S}_\text{BH}$
comes from the normalization of the boost
generator relative to the horizon entropy where regularity of the horizon fixes the total boost
angle of a smooth geometry to $2\pi$ and identifies the entropy as the
variable conjugate to the deficit around that value
\cite{GibbonsHawking1977Action,Carlip_1995,BanadosTeitelboimZanelli1994Entropy}.
Indeed, the pair $(\Theta, P_\Theta)$ we have constructed is just
the horizon canonical pair of Carlip and Teitelboim, where the horizon area
divided by $8\pi G$ is conjugate to the boost (opening) angle at the
horizon \cite{Carlip_1995,MassarParentani2000AreaBoost}, and one can
check from $P_\Theta = 2GE^2/c^3$ together with $r_s = 2GE/c^3$ that
$P_\Theta = Ac^3/8\pi G$ identically.

To visualize $\Theta$ and its relation to the entropy, note that it admits
two equivalent horizon definitions.

Let $\mathcal{B}_\tau \equiv \Sigma_\tau \cap \mathcal H$ be the
two-surface obtained by intersecting an observer-congruent
hypersurface $\Sigma_\tau$ with the black-hole horizon. Then at each
point of $\mathcal{B}_\tau$ the normal space is two-dimensional and
Lorentzian. Now if we denote $n^\mu$ as the observer congruence vector
and $s^\mu$ as the outward unit spacelike normal to $\mathcal{B}_\tau$
inside $\Sigma_\tau$ with $s^\mu s_\mu = -1$ and $n^\mu s_\mu = 0$,
then a change of orthonormal frame in the normal two-plane is a
Lorentz boost

\begin{equation}
  \begin{split}
    n'^\mu(\Theta) &= \cosh \Theta n^\mu + \sinh \Theta s^\mu, \\
    s'^\mu(\Theta) &= \sinh \Theta n^\mu + \cosh \Theta s^\mu,
  \end{split}
\end{equation}

\noindent where the parameter $\Theta$ is the rapidity or boost
angle. Remarkably, this variable also arises if we express the
Schwarzschild metric in Rindler form near the horizon. Let us
define $T = c\tau$ and have

\begin{equation}
  ds^2 = \left(1 - \frac{r_s}{r}\right)dT^2
  - \left(1 - \frac{r_s}{r}\right)^{-1}dr^2
  - r^2d\Omega^2,
\end{equation}

\noindent where since we are near the horizon we have
$r = r_s + \epsilon$, with $|\epsilon| \ll r_s$. Then

\begin{equation}
  1 - \frac{r_s}{r} = 1 - \frac{r_s}{r_s + \epsilon}
  \simeq \frac{\epsilon}{r_s}.
\end{equation}

\noindent Now introduce $\rho \equiv 2\sqrt{r_s \epsilon}$, such that

\begin{equation}
  d\rho^2 = \frac{r_s}{\epsilon}d\epsilon^2
  \simeq \left(1 - \frac{r_s}{r}\right)^{-1}dr^2,
\end{equation}

\noindent and

\begin{equation}
  \left(1 - \frac{r_s}{r}\right)dT^2
  \simeq \frac{\epsilon}{r_s}dT^2 = \frac{\rho^2}{4r^2_s}dT^2,
\end{equation}

\noindent so the normal part of the line element becomes

\begin{equation}
  ds^2_\perp \simeq \frac{\rho^2}{4r^2_s}dT^2 - d\rho^2,
\end{equation}

\noindent where as commonly derived in the literature \cite{GibbonsHawking1977Action,Hawking1975ParticleCreation,BardeenCarterHawking1973} by identifying $\Theta \equiv T/2r_s$ we obtain that $\Theta$ is
the near-horizon Rindler boost parameter

\begin{equation}
  ds^2_\perp \simeq \rho^2d\Theta^2 - d\rho^2.
\end{equation}

\noindent One caution is needed. Unlike a rotation angle, the boost rapidity
$\Theta$ ranges over the whole real line, so no periodicity can be invoked
in the Lorentzian theory and the spectrum of $\hat{P}_\Theta$ is continuous.
Since $\hat{P}_\Theta = 2G\hat{H}_r^2/c^3$ and the free half-line Hamiltonian
$\hat{H}_r$ has purely continuous spectrum $[0, \infty)$, we can parametrize
the spectrum of $\hat{P}_\Theta$ in the continuous half-line as
\begin{equation}\label{eq:PThetaContinuous}
  \hat{P}_\Theta\ket{\psi_\nu} = \hbar\nu\ket{\psi_\nu},
  \qquad
  \nu \in [0, \infty),
\end{equation}

\noindent where $\nu$ is a dimensionless and continuous label. Through $P_\Theta = Ac^3/8\pi G$ and $E = c^3r_s/2G$ each value of $\nu$ corresponds to

\begin{equation}\label{eq:NellP}
  \begin{aligned}
    A_\nu&=8\pi\nu\,\ell^2_\text{P},
    \\
    E_\nu&=\frac{E_\text{P}}{\sqrt{2}}\sqrt{\nu},
    \\
    r_{s,\nu}&=\sqrt{2\nu}\,\ell_\text{P},
  \end{aligned}
\end{equation}

\noindent so the horizon area, mass and Schwarzschild radius are
smooth functions of $\nu$ over a continuum of eigenstates.

If we recall from the no ``naked bounce'' condition in Eq.~\eqref{eq:rsrb} we would indeed need to also make it apply here. Semiclassically, if we allowed $r_b \geq r_{s,\nu}$ then no black hole formation could ever happen for that specific $\nu$ state, which seems unreasonable for any value of $\nu$. Namely, let us define
\begin{equation}
  F(r) \equiv \frac{r_s}{r}\left[1 - \left(\frac{r_b}{r}\right)^3\right],
\end{equation}
then on the collapsing branch we find it possible to express Eq.~\eqref{eq:schw-corrected-r-friedmann} as $r'_\text{eff} = -\sqrt{F(r_\text{eff})}$. From here the future-directed radial null expansions may be chosen (up to positive normalization) as
\begin{equation}
  k_\pm = \partial_\lambda \pm \frac{1}{\partial_\rho r_\text{eff}}\partial_\rho,
\end{equation}
where for a symmetrical sphere its radius is just $A=4\pi r^2_\text{eff}$ and thus we can express its null expansions as
\begin{equation}
  \theta_\pm = \frac{1}{A}k_\pm(A) = \frac{2}{r_\text{eff}}k_\pm(r_\text{eff}) = \frac{2}{r_\text{eff}}(r'_\text{eff} \pm 1).
\end{equation}
Then on the collapsing branch ($\theta_+ < 0$) we find $F(r) > 1$ and a marginally trapped sphere satisfies $F(r) = 1$. Now let us demand that a black hole is always a trapped region at some $r$ very close to zero for any $\nu > 0$. This must be the case because otherwise we would find the nonsensical case of a black hole not fully evaporated that although is an untrapped region (i.e., it stops being a black hole before fully evaporating). Mathematically this means that while $A > 0$ then there must be a trapped region $F(r) \geq 1$ for some $r$ as we approach $r \longrightarrow 0$. Notoriously this cannot be $r_b$ because $F(r_b) = 0$; which is expected if $r_b$ is going to represent a bounce. We also have $F(\infty) = 0$ so the $F(r)\geq 1$ value we are looking for must be somewhere in between. Now if we maximize $F(r)$ by setting $F'(r_\text{max}) = 0$ we find $r_\text{max} = \sqrt[3]{4}r_b$. Thus if we demand $F(r) \geq 1$ there are indeed two values that cross the $F(r) = 1$ limit. This also comes from the fact that $F(r) = 1$ is equal to
\begin{equation}\label{eq:schw-trapped-quartic}
  r^4 -r_{s,\nu}r^3 + r_{s,\nu}r^3_b = 0,
\end{equation}
which can be found to have two solutions that we can call $r_-$ and $r_+$. Then we see $F(r) < 1$ for $r > r_+$ and for $r_b \leq r \leq r_-$, but $F(r) > 1$ for $r_- < r < r_+$. For large black holes where $r_b/r_{s,\nu} \ll 1$ then $r_+$ can be approximated to be $r_+ \simeq r_{s,\nu}$ and for some $r_-$ slightly above the bounce radius---i.e., $r_b \simeq r_b(1 + \delta)$---then $F(r_-) \simeq 3r_{s,\nu}\delta/r_b = 1$, so $r_b \simeq r_b(1 + r_b/3r_{s,\nu})$, which for a large black hole with $r_b / r_{s,\nu} \ll 1$ is barely above $r_b$, which is sufficient for particles to bounce slightly before falling back in.

Now from these definitions it becomes clear that as the parameter $\nu$ is lowered then $r_- \longrightarrow r_+$. This is because the curve traced by $F(r)$ has a lower $F_\text{max}$ as we decrease $\nu$; note that at $r_\text{max} = r_b$ then $F_\text{max} = F(r_\text{max}) = 3r_s/4^{4/3}r_b$, so if $M$ decreases then $r_{s,\nu}$ decreases and $r_b$ increases, so there shall be a crossing point $F_\text{max} = 1$. This is the limit we must have for any $\nu > 0$; so if we require $F_\text{max} \geq 1$ at some $r$ for any $\nu > 0$ (thus having a trapped region for any $\nu > 0$), then $F_\text{max} \geq 1$ is the same as to ask
\begin{equation}
  r_b \leq \frac{3}{4^{4/3}}r_{s,\nu},
\end{equation}
and for a lower bound on $r_b \propto M^{-1}$ and $r_{s,\nu} \propto M$, if we replace $M = M_\text{P}\sqrt{\nu/2}$ in turn gives
\begin{equation}
  \nu \geq \nu_\text{b,trapped} = 2\sqrt{3},
\end{equation}
which gives us a physical lower bound for trapped or marginally trapped surfaces. This of course implies a reinforced scale given by
\begin{equation}
  \begin{aligned}
    r_\text{min}&=2\sqrt[4]{3}\ell_\text{P},
    \\
    A_\text{min}&=16\sqrt{3}\pi\ell^2_\text{P},
    \\
    M_\text{min}&=\sqrt[4]{3}M_\text{P},
    \\
    E_\text{min}&=\sqrt[4]{3}E_\text{P},
  \end{aligned}
\end{equation}
all of roughly the Planck scale, as one would expect. We can call the black holes with these minimum values the ``terminal family'' for the sake of economy in the rest of this work.

These relations have the additional benefit of sparing us
some unphysical divergences. Note that
the classical definition of temperature is
$T^{-1}(E) = k_\text{B}\beta(E)$, which for
$k_\text{B}\beta(E) = \partial_E S(E)$, our definition of
$\hat{S}_\text{BH}$ and its eigenvalues would give
\begin{equation}
  \hat{T}_\text{semi} = \frac{E^2_\text{P}}{8\pi k_\text{B}}\hat{E}^{-1},
\end{equation}
which is the semiclassical Hawking expression written in terms of the boundary charge. This expression, however, samples the entropy through the semiclassical surface term alone; i.e., through the area $A = 4\pi r^2_{s,\nu}$ of the on-shell identification $r_s = 2GE/c^3$. We have just seen that in our formalism the surface bounding the trapped region does not sit at $r_{s,\nu}$ and instead the bounce displaces the marginally trapped surface to the outer root $r_+$ of Eq.~\eqref{eq:schw-trapped-quartic}, and $r_+$ approaches $r_{s,\nu}$ only for large black holes. The effective surface term is therefore built on the area $A_+ = 4\pi r^2_+$ and it is this area that the thermodynamic derivative must sample as
\begin{equation}
  S(E) = \frac{k_\text{B}c^3}{4G\hbar}A_+ = \frac{\pi k_\text{B}}{\ell^2_\text{P}}r^2_+(E).
\end{equation}
To differentiate let us first put Eq.~\eqref{eq:schw-trapped-quartic} in dimensionless form. From Eq.~\eqref{eq:schw-reff-quantized} on the ground state we have $r^3_b = 81G\hbar^2/32Mc^4$, which together with $M = M_\text{P}\sqrt{\nu/2}$ gives $(r_b/r_{s,\nu})^3 = 81/64\nu^2$, so the ratio $x_\nu \equiv r_+/r_{s,\nu}$ is the largest root of
\begin{equation}\label{eq:schw-xnu-quartic}
  x^4_\nu - x^3_\nu + \frac{81}{64\nu^2} = 0,
\end{equation}
which ranges over $3/4 \leq x_\nu < 1$ for $\nu \geq 2\sqrt{3}$. Implicit differentiation of Eq.~\eqref{eq:schw-trapped-quartic} with $dr_{s,\nu}/dE = r_{s,\nu}/E$ and $dr^3_b/dE = -r^3_b/E$ cancels the $r_b$ contributions between the two terms of the quartic and leaves
\begin{equation}
  \frac{dr_+}{dE} = \frac{r_{s,\nu}r_+}{E\left(4r_+ - 3r_{s,\nu}\right)},
\end{equation}
so the corrected temperature reads
\begin{equation}\label{eq:schw-T-corrected}
  \hat{T} = \frac{E^2_\text{P}}{8\pi k_\text{B}}\hat{E}^{-1}\frac{4x_\nu - 3}{x^2_\nu},
\end{equation}
where $x_\nu$ is understood as the operator function $x(\hat{P}_\Theta/\hbar)$ defined by the functional calculus. Note that it is legitimate because $\hat{P}_\Theta$ and $\hat{E}$ are both functions of $\hat{H}_r$ and commute. For large black holes Eq.~\eqref{eq:schw-xnu-quartic} gives $x_\nu = 1 - 81/64\nu^2 + \mathcal{O}(\nu^{-4})$, so
\begin{equation}
  \hat{T} \simeq \frac{E^2_\text{P}}{8\pi k_\text{B}}\hat{E}^{-1}\left(1 - \frac{81}{32\nu^2}\right),
\end{equation}
and the semiclassical temperature is recovered with a correction that is utterly negligible for any astrophysical black hole; a solar mass corresponds to $\nu \sim 10^{76}$. At the opposite end the two roots of Eq.~\eqref{eq:schw-xnu-quartic} merge at $x_\nu = 3/4$ when $\nu = 2\sqrt{3}$, which is the same merging of $r_-$ and $r_+$ that defined the terminal family, so the factor $4x_\nu - 3$ vanishes exactly on it. Expanding the quartic around this double root we find
\begin{equation}
  T \simeq \frac{\sqrt{2}E_\text{P}}{9\pi k_\text{B}}\sqrt{\nu - 2\sqrt{3}},
\end{equation}
so the terminal family sits at $T = 0$ and is approached with the square-root law familiar from extremal horizons. The semiclassical expression $\hat{T}_\text{semi}$ would instead assign it the finite value $E_\text{P}/8\sqrt[4]{3}\pi k_\text{B}$; so the effective surface term here replaces this frozen maximal temperature with a cold endpoint. Since the corrected temperature vanishes both on the terminal family and in the limit of large mass then it must attain a maximum in between. Solving Eq.~\eqref{eq:schw-xnu-quartic} for $\nu = (9/8)x^{-3/2}_\nu(1 - x_\nu)^{-1/2}$ and maximizing Eq.~\eqref{eq:schw-T-corrected} along this curve gives $x_\nu = 15/16$; i.e.,
\begin{equation}
  T_\text{max} = \frac{2E_\text{P}}{15^{5/4}\pi k_\text{B}},
  \qquad
  \nu_{T_\text{max}} = \frac{32\sqrt{15}}{25} \simeq 4.96,
\end{equation}
still of the order of the Planck temperature but attained slightly above the terminal family rather than on it. The evaporation picture that follows is that a large black hole heats up as it radiates, reaches $T_\text{max}$, and then cools along the last stretch until the flux shuts off on the terminal family and it behaves as a zero-temperature remnant. Note also that the effective-surface entropy differs from the eigenvalues of $\hat{S}_\text{BH}$ only at order $\nu^{-1}$ because $S = 2\pi k_\text{B}\nu x^2_\nu = 2\pi k_\text{B}\nu - 81\pi k_\text{B}/16\nu + \cdots$, so the boost-conjugate counting used below is unaffected at leading order.

However, note that the area of applicability of these results is limited to the case of a reduced, spherically symmetric sector. Whether the dynamical process of evaporation actually reaches the terminal configuration, or
whether the emission of the last quanta can disperse it entirely,
cannot be decided within the reduced formalism. That would require coupling
the boundary theory to the radiation degrees of freedom and to the
wider spacetime dynamics of formation and evaporation. This remains a
work of future study; what the present kinematics does establish is
that, if a horizon exists at all, its size and mass are bounded by the
terminal configuration and its temperature by $T_\text{max}$.

This theory also provides us with the ``microscopic'' explanation that
Carlip and Teitelboim were looking for. The degeneracy of the horizon
macrostate \cite{GibbonsHawking1977Action,Bekenstein1973Entropy,Hawking1975ParticleCreation} is given by
\begin{equation}
  \Omega \sim \exp\left(\frac{S_\text{BH}}{k_\text{B}}\right)
  = \exp(2\pi \nu) \geq \exp(4\sqrt{3}\pi),
\end{equation}
where the label $\nu$ runs over the continuous spectrum of Eq.
(\ref{eq:PThetaContinuous}) and the terminal bound $\nu \geq 2\sqrt{3}$ cuts
the counting off from below. The exponential weight $e^{2\pi\nu}$ therefore
appears directly as the density of boundary states conjugate to the boost
angle, which is the microscopic origin that Carlip and Teitelboim identified
as missing from the Euclidean treatment \cite{Carlip_1995}. The counting is
continuous rather than a degeneracy of discrete area eigenvalues, in
agreement with the continuous spectra found throughout this sector.

The terminal family also makes contact with the remnant literature.
Arguments based on the generalized uncertainty principle deform the
canonical commutator and predict a maximal temperature together with a
Planck-mass endpoint of evaporation
\cite{AdlerChenSantiago2001GUP,ChenOngYeom2015Remnants}. The bounds
$T_\text{max}$ and $M_\text{min}$ obtained above are of the same form, yet
they arise here with the canonical commutator untouched. Its inputs are the
conserved invariant $I_r$, the DeWitt condition acting through the
representation $D^+(k=3/4)$, and the requirement that every black hole
possess a trapped region. A further difference is that the endpoint here
is reached at vanishing temperature through
Eq.~\eqref{eq:schw-T-corrected} in the manner of extremal horizons,
whereas the generalized uncertainty principle remnants typically
terminate while still at their maximal temperature. Whether evaporation dynamically halts on the
terminal family, and with it the fate of the information problem, remains
outside the reduced formalism for the reasons discussed above.

\subsection{Newtonian mechanics}
There is, however, another conceivable singularity arising from the scattering
amplitudes. Let us build the $S$ matrix using the linearised version of our
canonical Hamiltonian. If we impose an approximately on-shell regime, then
$R \approx \kappa T$ and thus we obtain the
$\sim h_{\mu\nu}T^{\mu\nu}$ interaction term to obtain the Newtonian potential.
Then for elastic scattering of two heavy scalar sources, define the momentum transfer
\begin{equation}
  q^\mu=p_f^\mu-p_i^\mu,
\end{equation}
which in the non-relativistic regime gives
\begin{equation}
  n_\mu q^\mu\simeq0,
  \qquad
  q_\perp^\mu\equiv\mathcal P^\mu{}_\nu q^\nu,
  \qquad
  q^2\simeq -Q^2,
\end{equation}
where
\begin{equation}
  Q^2
  \equiv
  \gamma_{\mu\nu}q_\perp^\mu q_\perp^\nu
  =
  -g_{\mu\nu}q_\perp^\mu q_\perp^\nu
  >0.
\end{equation}
Similarly, for an observer-spatial separation $X_\perp^\mu$,
\begin{equation}
  n_\mu X_\perp^\mu=0,
  \qquad
  \ell^2
  \equiv
  \gamma_{\mu\nu}X_\perp^\mu X_\perp^\nu.
\end{equation}
The non-relativistic amplitude is defined by
\begin{equation}
  \mathcal M_{\rm NR}(q_\perp)
  \equiv
  \frac{\mathcal M(q_\perp)}{4m_1m_2}.
\end{equation}
The effective potential is obtained by matching to the first Born
approximation on the asymptotic observer-spatial cotangent space
\begin{equation}
  V_{\rm eff}(\ell)
  =
  -\int_{\Sigma_n^\ast}
  \frac{d^3q_\perp}{(2\pi)^3}\,
  e^{iq^\perp_\mu X_\perp^\mu}\,
  \mathcal M_{\rm NR}(q_\perp).
\end{equation}
At tree level, one-graviton exchange gives the long-range amplitude
\begin{equation}
  \mathcal M_{\rm NR}^{(0)}(q_\perp)
  =
  \frac{4\pi Gm_1m_2}{Q^2}.
\end{equation}
The needed observer-spatial Fourier integral is
\begin{equation}
  \int_{\Sigma_n^\ast}
  \frac{d^3q_\perp}{(2\pi)^3}\,
  \frac{e^{iq^\perp_\mu X_\perp^\mu}}{Q^2}
  =
  \frac{1}{4\pi\ell}.
\end{equation}
Therefore
\begin{equation}
  V_{\rm eff}^{(0)}(\ell)
  =
  -\frac{Gm_1m_2}{\ell}.
\end{equation}
For a heavy source $m_1=M$ and a test particle $m_2=m$,
\begin{equation}
  V_{\rm eff}^{(0)}(\ell)=m\Phi_{\rm N}(\ell),
  \qquad
  \Phi_{\rm N}(\ell)
  =
  -\frac{GM}{\ell}.
\end{equation}
The leading long-range corrections come from the non-analytic terms of the
one-loop scattering amplitude \cite{Donoghue1994EFT,BjerrumBohrDonoghueHolstein2003}.
Analytic terms in $Q^2$ Fourier transform
into contact terms supported at $X_\perp^\mu=0$, such as the spatial delta
distribution $\delta_{\Sigma_n}(X_\perp)$ and its derivatives. They do not
control the long-distance potential. After subtracting the iterated Born term,
the non-analytic low-energy amplitude has the form
\begin{equation}
  \begin{aligned}
    \mathcal M_{\rm NR}(q_\perp)
    ={}&
    \frac{4\pi Gm_1m_2}{Q^2}
    \\
    &+
    \frac{6\pi^2G^2m_1m_2(m_1+m_2)}
    {c^2Q}
    \\
    &-
    \frac{41}{5}
    \frac{G^2m_1m_2\hbar}{c^3}
    \log\left(\frac{Q^2}{\mu^2}\right)
    \\
    &+
    \mathcal M_{\rm an.}(q_\perp),
  \end{aligned}
  \label{eq:one-loop-long-range-grav-amp}
\end{equation}
where $\mathcal M_{\rm an.}$ denotes analytic short-range terms. The two
additional Fourier transforms are
\begin{equation}
  \int_{\Sigma_n^\ast}
  \frac{d^3q_\perp}{(2\pi)^3}\,
  \frac{e^{iq^\perp_\mu X_\perp^\mu}}{Q}
  =
  \frac{1}{2\pi^2\ell^2},
\end{equation}
and
\begin{equation}
  \int_{\Sigma_n^\ast}
  \frac{d^3q_\perp}{(2\pi)^3}\,
  e^{iq^\perp_\mu X_\perp^\mu}
  \log\left(\frac{Q^2}{\mu^2}\right)
  =
  -\frac{1}{2\pi\ell^3},
\end{equation}
up to contact terms. Applying the Born matching term by term gives
\begin{align}
  V_{\rm eff}(\ell)
  ={}&
  -
  \frac{Gm_1m_2}{\ell}
  \nonumber\\
  &
  -
  \frac{3G^2m_1m_2(m_1+m_2)}{c^2\ell^2}
  \nonumber\\
  &
  -
  \frac{41}{10\pi}
  \frac{G^2m_1m_2\hbar}{c^3\ell^3}
  \nonumber\\
  &
  +
  V_{\rm contact}(\ell)
  +
  \cdots
  \nonumber\\
  ={}&
  -\frac{Gm_1m_2}{\ell}
  \left[
    1
    +
    3\frac{G(m_1+m_2)}{\ell c^2}
  \right.
  \nonumber\\
  &\left.\qquad\qquad
    +
    \frac{41}{10\pi}
    \frac{G\hbar}{\ell^2c^3}
  \right]
  \nonumber\\
  &
  +
  V_{\rm contact}(\ell)
  +
  \cdots .
  \label{eq:donoghue-effective-potential}
\end{align}
The first correction is classical and post-Newtonian. The second correction is
quantum and proportional to $\hbar$. Therefore, for a heavy source $M$ and a test
mass $m\ll M$,
\begin{equation}
  \Phi_{\rm eff}(\ell)
  =
  -\frac{GM}{\ell}
  \left[
    1
    +
    3\frac{GM}{\ell c^2}
    +
    \frac{41}{10\pi}
    \frac{G\hbar}{\ell^2c^3}
  \right]
  +
  \cdots .
\end{equation}
Equivalently, using the Planck length $\ell_{\rm P}$ the quantum correction is
\begin{equation}
  \Phi_{\rm q}(\ell)
  =
  -\frac{GM}{\ell}
  \left[
    \frac{41}{10\pi}
    \frac{\ell_{\rm P}^2}{\ell^2}
  \right].
\end{equation}
The theory would clearly diverge for $\ell \rightarrow 0$, but if we invoke the minimum length scale for gravitational masses $\ell_\text{min} = 2\sqrt[4]{3}\ell_\text{P}$ then the $h_{00}$ component becomes finite and we can thus regularize $\mathcal M_{\rm NR}$. Note that this minimum length scale regularization of the potential implies a maximum UV cut-off; i.e., a maximum momentum $p_\text{max} = \hbar \Lambda_\text{UV} = cM_\text{P} / 2\sqrt[4]{3}$ for gravitons because a minimum $\ell$ implies a maximum $Q_\text{max} = \Lambda_\text{UV}$ at Eq. (\ref{eq:one-loop-long-range-grav-amp}).
 \section{Conclusion}

We have constructed a covariant Hamiltonian formulation in which the role of time evolution
is played by the Lie derivative along an observer congruence. The preferred external time of
ordinary Hamiltonian mechanics is thereby replaced by the proper-time flow of a family of
physical observers, while the configuration variables are projected onto the local rest
spaces orthogonal to $n^\mu$. The result is a generally relativistic Hamiltonian formalism
that reduces to the classical one in the appropriate limit and, crucially, assigns a
well-defined evolution to the metric tensor itself.

The formalism also avoids most of the structural difficulties of the ADM construction
\cite{Arnowitt1962Dynamics}. There is no Hamiltonian or momentum constraint; the only constraints are
those demanding that the canonical pair be observer-spatial, $n^\mu\hat\gamma_{\mu\nu}
\ket{\Psi}=0$ and $n_\mu\hat\pi^{\mu\nu}\ket{\Psi}=0$. These are second class with an
invertible constraint matrix, so the Dirac bracket is trivial to construct and quantization
proceeds without secondary constraints.
Because the proper-time Schr\"odinger equation is
first order in $\tau$ and its generator is symmetric with respect to the invariant
configuration-space measure the theory carries the conserved positive inner product of
Eq. (\ref{eq:qg-inner-product}), so neither the frozen dynamics
nor the inner-product problem of the Wheeler--DeWitt kernel equation reappears
\cite{Kuchar2011Time,Isham1993Canonical}.
Gauge invariance under observer-spatial
diffeomorphisms is then carried by the metric density $\sqrt{\gamma}$ and the superspace density
$\sqrt{\mathscr{M}}$, which relate different coordinatizations both of spacetime and of the
reduced configuration space---the first on the lines of passing from $(ct,x)$ to $(ct,r)$,
the second on the lines of passing from $(a,p_a)$ to $(u,p_u)$ with $u=f(a)$.

Upon quantization, the ordering ambiguity of the kinetic term $\pi_{\mu\nu}\pi^{\mu\nu}
-\pi^2/2$ is resolved by Theorem~\ref{thm:ordering}.
This solves a classic and arguably the greatest difficulty in the
Wheeler--DeWitt approach \cite{DeWitt1967Quantum}.
Covariance under field redefinitions on configuration space in the sense
that quantizing $(a,p_a)$ and then passing to $(u,p_u)$ agrees with passing
first and quantizing after, combined with Hermiticity with respect to
the invariant measure (required by unitary proper-time evolution)
singles out the Laplace--Beltrami ordering up to scalar potentials of
the curvature type $\xi R$; the minimality condition $\hat T 1 = 0$ then
removes these as non-minimal $\hbar^2$-order interactions
\cite{DeWitt1957DynamicalTheory}. These hypotheses are available here
because the theory evolves unitarily in the proper time of the
congruence. In the ADM-based Wheeler--DeWitt quantization the analogous
operator can be written down but cannot be singled out in this way,
since no physical inner product is available before the problem of time
is resolved and the ordering is further entangled with the closure of
the constraint algebra \cite{Kuchar2011Time,Isham1993Canonical,TsamisWoodard1987Regulated}.
Indeed, Tsamis and Woodard showed that the ADM ordering problem cannot
even be stated before the constraints are regulated because verifying
Dirac consistency of the Hamiltonian and momentum constraints requires
multiplying coincident operator-valued distributions
\cite{TsamisWoodard1987Regulated}. But since the present theory possesses
neither constraint then that
necessity never arises and the theorem is stated and proved without any
ultraviolet regulator. The same theorem renders the symmetry-reduced
models we have derived unambiguous, because the reduced kinetic metrics
fix the reduced operators without further choices.

The applications of the theory yield the predictions one would hope for from a quantum
theory of gravity. In the cosmological sector the quantum-corrected Friedmann equations
acquire a repulsive $\propto I_q/a_{\rm eff}^6$ term controlled by the conserved quantum
invariant $I_q \geq \hbar^2/4$, so the effective volume of the universe is bounded below by
Eq. (\ref{eq:flrw-min-volume-bound}) and the Big Bang singularity is resolved for any value of the curvature
parameter $k$; the classical singularity is recovered only in the limit $\hbar\to 0$ where
the invariant is allowed to vanish. Note that if we followed the DFF trick we used with black holes we could improve the bound to $I_q \geq 9\hbar^2/4$ as well. In the Schwarzschild--Lemaître sector the boundary
dynamics reduce to a free particle on the half-line in $v\propto r^{3/2}$, and the analogous
invariant $I_r$ prevents the effective areal radius from reaching zero, keeping the effective Kretschmann scalar that a particles sees always finite. The DeWitt condition at $v=0$ places the boundary theory in the
representation $D^+(k=3/4)$ of the $\mathfrak{so}(2,1)$ conformal algebra of de Alfaro,
Fubini and Furlan (DFF), on whose discrete basis the invariant can be calculated to be
$I_r = \hbar^2(4n+3)^2/4$. Yet, all observables of the theory are continuous and the role of these moment invariants are just to give lower bounds to them.

The angle--area picture obtained by the canonical transformation to the horizon pair
$(\Theta, P_\Theta)$ of Carlip--Teitelboim type then organizes these results into black-hole
thermodynamics.
The entropy operator is conjugate to the boost angle, the temperature
acquires an effective-surface correction that keeps it below
$T_{\rm max} = 2E_P/15^{5/4}\pi k_B$ and sends it to zero on the terminal family, and the smallest horizon supported by the theory is
the terminal configuration with $r_{\rm min}=2\sqrt[4]{3}\ell_P$ and
$M_{\rm min}=\sqrt[4]{3} M_P$.
These however are kinematical statements about the terminal configuration. Whether the dynamical process of
evaporation actually terminates on the terminal family---leaving a Planck-scale remnant, with
the implications this would carry for the information problem---or whether the emission of the
last quanta disperses the configuration entirely, cannot be decided within the reduced
formalism and requires coupling to the radiation degrees of freedom. What the present
kinematics does establish is that, if a horizon exists, its size, mass and temperature are
bounded by the terminal configuration, and that these same bounds provide the microscopic continuous counting $\Omega \sim e^{2\pi\nu}$ that Carlip and Teitelboim identified as missing from the Euclidean treatment \cite{Carlip_1995}. If further studies confirm no exterior dynamics could remove the $2\sqrt[4]{3}\ell_\text{P}$ limit for the Schwarzschild radius it would automatically solve the information paradox.

Read in reverse, the terminal bound also states that no horizon forms below
$M_{\rm min}=\sqrt[4]{3}M_P$. Elementary particles therefore do not form black holes on their own and ordinary interactions have no channel by which to concentrate a super-Planckian energy in a sufficiently small region. One may picture this as a threshold on curvature excitations of the Schwarzschild minisuperspace, where a single electron cannot radiate arbitrarily small, smooth curvature wavepackets of this sector because the corresponding waves demand a minimum of energy that it does not possess.

Finally, this minimum radius supplies a physical ultraviolet cutoff for graviton dynamics. In the scattering sector the momentum transfer is cut off at
$Q_{\max} = \Lambda_{\rm UV} = (2\sqrt[4]{3}\ell_P)^{-1}$, i.e.\ $p_{\max} = \hbar\Lambda_{\rm UV} = cM_P/2\sqrt[4]{3}$, so the effective potential Eq. (\ref{eq:donoghue-effective-potential}) is regularized and the theory
acquires a UV-complete description of the graviton sector within its domain of validity.

\begin{appendices}
\section{Unique ordering}\label{app:ordering}
\begin{theorem}[Uniqueness of the minimal field-redefinition-covariant ordering]
\label{thm:ordering}

Let $\Phi^A(x)$ be a general tensor field on the observer-spatial
leaf $\Sigma_\lambda$, with $\lambda=c\tau$ and
$\dot{\Phi}^A \equiv \mathcal L_n\Phi^A$. The condensed index
$A$ contains tensor indices, projected internal labels, and
field-species labels.
Let the configuration space $\mathcal C$ be a
smooth manifold of finite dimension $N$ with coordinates $q^A$.
Sectors of this kind arise as the symmetry reductions of the main
text and as the single-point blocks of the ultralocal kinetic metric
of the full theory, through which Remark
\ref{rem:field-space-extension} carries the conclusion to cylinder
wavefunctionals on the full field space; the theorem itself is a
statement of finite-dimensional geometry, and no ultraviolet
regularization enters its hypotheses. We require every invertible
pointwise redefinition $\Phi\mapsto\chi(\Phi)$ to act as a
diffeomorphism of $\mathcal C$, which holds automatically in both
cases just named because such a redefinition acts separately on each
configuration variable; a sharp truncation in Fourier modes would in
general lack this property, since a nonlinear $\chi$ mixes retained
and discarded modes, and this is why the extension of Remark
\ref{rem:field-space-extension} proceeds through point blocks rather
than through modes.
Consider a
natural Lagrangian of the form
\begin{equation}
  L(q,\dot q)
  =
  \frac{1}{2}\mathscr G_{AB}(q)\dot q^A\dot q^B
  -
  U(q),
  \label{eq:natural-lagrangian}
\end{equation}
where $\mathscr G_{AB}$ is a non-degenerate, possibly indefinite, metric
on $\mathcal C$. The canonical momenta and Hamiltonian are
\begin{equation}
  \begin{aligned}
    p_A
    &=
    \frac{\partial L}{\partial \dot q^A}
    =
    \mathscr G_{AB}\dot q^B,
    \\
    H(q,p)
    &=
    \frac{1}{2}\mathscr G^{AB}(q)p_Ap_B
    +
    U(q).
  \end{aligned}
  \label{eq:natural-hamiltonian}
\end{equation}
Now let
\begin{equation}
  \sqrt{\mathscr M(q)}
  :=
  \sqrt{\left|\det \mathscr G_{AB}(q)\right|}
  \label{eq:field-space-measure-density}
\end{equation}
be the invariant density induced by the kinetic metric. Suppose that
the quantum kinetic operator is a second-order differential operator on
scalar wavefunctions $\Psi(q)$ with principal symbol $\mathscr G^{AB}$
\begin{equation}
  \hat T
  =
  -\frac{\hbar^2}{2}
  \left(
    \mathscr G^{AB}\partial_A\partial_B
    +
    B^A(q)\partial_A
    +
    C(q)
  \right).
  \label{eq:general-second-order-ordering}
\end{equation}
Assume the following four conditions.

First, $\hat T$ is covariant under arbitrary invertible field
redefinitions
\begin{equation}
  Q^I=Q^I(q).
  \label{eq:field-redefinition}
\end{equation}
That is, if $\widetilde\Psi(Q)=\Psi(q(Q))$, then quantizing after the
field redefinition gives the same operator as redefining after
quantization
\begin{equation}
  \left(\hat T_Q\widetilde\Psi\right)(Q)
  =
  \left(\hat T_q\Psi\right)(q(Q)).
  \label{eq:quantization-commutes-with-field-redefinition}
\end{equation}
Second, $\hat T$ is constructed only from the kinetic metric
$\mathscr G_{AB}$ and its derivatives, with no additional background
structure on $\mathcal C$. In particular, the inner product on
wavefunctions is the one defined by the invariant measure of the
kinetic metric,
\begin{equation}
  \langle\Psi,\Phi\rangle
  =
  \int_{\mathcal C} d^Nq\,\sqrt{\mathscr M}\,\Psi^*\Phi,
  \label{eq:induced-inner-product}
\end{equation}
since this is the only volume element that can be built without
introducing further structure. Third, $\hat T$ is symmetric with
respect to this inner product and commutes with the antiunitary
conjugation $\Psi\mapsto\Psi^*$, which implements time reversal in the
configuration representation. The first property is required for the
Hamiltonian to generate a unitary proper-time evolution, and the second
holds for any operator ordering of the classical kinetic term because
$\mathscr G^{AB}p_Ap_B$ is even in the momenta. Fourth, the quantization
is minimal, in the sense that a pure kinetic Hamiltonian contains no
additional scalar quantum potential
\begin{equation}
  \hat T\,1=0.
  \label{eq:minimality-condition}
\end{equation}
Then the unique admissible kinetic operator is the Laplace--Beltrami
operator on configuration space
\begin{equation}
  \hat T_{\rm min}
  =
  -\frac{\hbar^2}{2}\Delta_{\mathscr G}
  =
  -\frac{\hbar^2}{2}
  \frac{1}{\sqrt{\mathscr M}}
  \partial_A
  \left(
    \sqrt{\mathscr M}\,
    \mathscr G^{AB}\partial_B
  \right).
  \label{eq:minimal-laplace-beltrami-operator}
\end{equation}
\end{theorem}

\begin{proof}

Since $\mathcal C$ is an ordinary finite-dimensional manifold, every
measure, determinant, and integration by parts below is well defined,
and no regularization of any kind is invoked in the argument. The
functional expressions used in the main text stand for this
conclusion applied block by block to the ultralocal kinetic metric of
the full theory, in the sense made precise by Remark
\ref{rem:field-space-extension}.
Let $Q^I=Q^I(q)$ be an invertible field redefinition. The
canonical one-form must be invariant
\begin{equation}
  p_A\,dq^A
  =
  P_I\,dQ^I.
  \label{eq:canonical-one-form-invariance}
\end{equation}
Therefore the momenta transform as covectors on configuration space
\begin{equation}
  P_I
  =
  p_A\frac{\partial q^A}{\partial Q^I}.
  \label{eq:momentum-covector-law}
\end{equation}
On the other side, the kinetic metric transforms as
\begin{equation}
  \begin{aligned}
    \widetilde{\mathscr G}_{IJ}(Q)
    ={}&
    \frac{\partial q^A}{\partial Q^I}
    \frac{\partial q^B}{\partial Q^J}
    \mathscr G_{AB}(q),
    \\
    \widetilde{\mathscr G}^{IJ}(Q)
    ={}&
    \frac{\partial Q^I}{\partial q^A}
    \frac{\partial Q^J}{\partial q^B}
    \mathscr G^{AB}(q),
  \end{aligned}
  \label{eq:kinetic-metric-transformation}
\end{equation}
so it follows immediately that
\begin{equation}
  \frac{1}{2}\mathscr G^{AB}p_Ap_B
  =
  \frac{1}{2}\widetilde{\mathscr G}^{IJ}P_IP_J,
  \label{eq:classical-kinetic-invariance}
\end{equation}
and thus the classical kinetic Hamiltonian is invariant under field
redefinitions.

At the quantum level, the wavefunction is taken to be a scalar on
configuration space
\begin{equation}
  \widetilde\Psi(Q)
  =
  \Psi(q(Q)).
  \label{eq:wavefunction-is-scalar}
\end{equation}
The derivative transforms as
\begin{equation}
  \partial_A
  =
  \frac{\partial Q^I}{\partial q^A}\partial_I,
  \label{eq:first-derivative-chain-rule}
\end{equation}
and therefore
\begin{equation}
  \partial_A\partial_B
  =
  \frac{\partial Q^I}{\partial q^A}
  \frac{\partial Q^J}{\partial q^B}
  \partial_I\partial_J
  +
  \frac{\partial^2 Q^I}{\partial q^A\partial q^B}
  \partial_I.
  \label{eq:second-derivative-chain-rule}
\end{equation}
Hence
\begin{equation}
  \mathscr G^{AB}\partial_A\partial_B
  =
  \widetilde{\mathscr G}^{IJ}\partial_I\partial_J
  +
  \mathscr G^{AB}
  \frac{\partial^2 Q^I}{\partial q^A\partial q^B}
  \partial_I.
  \label{eq:bare-second-derivative-noncovariance}
\end{equation}
The second term in Eq. \eqref{eq:bare-second-derivative-noncovariance}
is a non-covariant term produced by a nonlinear field redefinition.
Therefore the bare operator $\mathscr G^{AB}\partial_A\partial_B$ is not
a scalar differential operator on configuration space.

To make it covariant let us consider the density induced by
$\mathscr G_{AB}$
\begin{equation}
  d\mu_{\mathscr G}
  =
  \sqrt{\mathscr M(q)}\,d^Nq.
  \label{eq:invariant-configuration-volume}
\end{equation}
Since
\begin{equation}
  \sqrt{\widetilde{\mathscr M}(Q)}\,d^NQ
  =
  \sqrt{\mathscr M(q)}\,d^Nq,
  \label{eq:measure-invariance}
\end{equation}
the differential operator
\begin{equation}
  \Delta_{\mathscr G}
  =
  \frac{1}{\sqrt{\mathscr M}}
  \partial_A
  \left(
    \sqrt{\mathscr M}\,
    \mathscr G^{AB}\partial_B
  \right)
  \label{eq:configuration-laplacian-definition}
\end{equation}
is a scalar operator
\begin{equation}
  \left(\Delta_{\widetilde{\mathscr G}}\widetilde\Psi\right)(Q)
  =
  \left(\Delta_{\mathscr G}\Psi\right)(q(Q)).
  \label{eq:laplace-beltrami-natural}
\end{equation}
Expanding Eq. \eqref{eq:configuration-laplacian-definition}, one obtains
\begin{equation}
  \Delta_{\mathscr G}
  =
  \mathscr G^{AB}\partial_A\partial_B
  +
  B^A_{\rm LB}\partial_A,
  \label{eq:laplace-beltrami-expanded}
\end{equation}
where
\begin{equation}
  B^A_{\rm LB}
  =
  \frac{1}{\sqrt{\mathscr M}}
  \partial_B
  \left(
    \sqrt{\mathscr M}\,
    \mathscr G^{BA}
  \right).
  \label{eq:lb-first-order-term}
\end{equation}
Thus the Laplace--Beltrami first-order term is sufficient for
field-redefinition covariance.

It remains to prove uniqueness. Suppose another operator $\hat T$
of the form \eqref{eq:general-second-order-ordering} satisfies the
same covariance condition and has the same principal symbol. Subtracting
the Laplace--Beltrami operator gives
\begin{equation}
  \hat T-\hat T_{\rm LB}
  =
  -\frac{\hbar^2}{2}
  \left(
    V^A\partial_A+C
  \right),
  \qquad
  V^A:=B^A-B^A_{\rm LB}.
  \label{eq:difference-from-lb}
\end{equation}
Since both $\hat T$ and $\hat T_{\rm LB}$ are scalar
differential operators, $V^A\partial_A$ must itself be a scalar
first-order differential operator. Therefore $V^A$ must be a vector
field on configuration space, and by the second condition it must be a
natural vector field constructed only from $\mathscr G_{AB}$.

Naturality alone, however, does not force $V^A$ to vanish. No nonzero
vector field can be built algebraically from the metric, but once
derivatives of $\mathscr G_{AB}$ are admitted the curvature supplies
natural candidates, the simplest being
\begin{equation}
  V^A
  =
  \xi\,\mathscr G^{AB}\partial_B R[\mathscr G],
  \label{eq:curvature-vector-candidate}
\end{equation}
with $R[\mathscr G]$ the scalar curvature of the kinetic metric and
$\xi$ a real constant. Candidates of this kind are the
configuration-space form of the curvature ambiguities long known in the
quantization of systems on curved spaces
\cite{DeWitt1957DynamicalTheory}. It is the Hermiticity condition that
removes them. For wavefunctions of compact support, integrating by
parts against the invariant measure of Eq.
\eqref{eq:induced-inner-product} gives
\begin{equation}
  \begin{aligned}
    &\left\langle
      \Psi,
      \left(V^A\partial_A+C\right)\Phi
    \right\rangle
    \\
    &\quad=
    \left\langle
      \left(
        -V^{A*}\partial_A
        -
        (\mathrm{div}_{\mathscr G}V)^*
        +
        C^*
      \right)\Psi,
      \Phi
    \right\rangle,
    \\
    &\mathrm{div}_{\mathscr G}V
    \equiv
    \frac{1}{\sqrt{\mathscr M}}
    \partial_A
    \left(\sqrt{\mathscr M}\,V^A\right).
  \end{aligned}
  \label{eq:first-order-adjoint}
\end{equation}
Symmetry of $\hat T$, and hence of the difference operator,
requires the first-order coefficients on the two sides to agree, so
$V^A=-V^{A*}$ and $V^A$ must be purely imaginary. This already
eliminates every real candidate such as Eq.
\eqref{eq:curvature-vector-candidate}. Invariance under the conjugation
$\Psi\mapsto\Psi^*$ requires $V^A$ to be real. The two conditions
together give $V^A=0$ on $\mathcal C$, and therefore
$B^A=B^A_{\rm LB}$, which proves uniqueness of the first-order part.
Comparing the zeroth-order coefficients in Eq.
\eqref{eq:first-order-adjoint} with $V^A=0$ then shows that $C$ must be
real.

However, the scalar term $C(q)$ is different. Field-redefinition
covariance allows it because a scalar transforms as
\begin{equation}
  \widetilde C(Q)=C(q(Q)).
  \label{eq:scalar-term-transformation}
\end{equation}
Hermiticity and conjugation invariance only require it to be real. In
particular, the curvature potential
\begin{equation}
  C(q)
  =
  \xi R[\mathscr G],
  \label{eq:curvature-scalar-candidate}
\end{equation}
satisfies the first three conditions for any real $\xi$, and adding it
to $\hat T_{\rm min}$ produces precisely the one-parameter family of
covariant Hermitian operators familiar from quantization on curved
configuration spaces
\cite{HawkingPage1986Ordering,DeWitt1957DynamicalTheory}. Thus
covariance, naturality, and Hermiticity fix the first-order part of the
operator but not an additional scalar quantum potential. It is the
minimality condition given by Eq. \eqref{eq:minimality-condition} that
removes this freedom.
Since
\begin{equation}
  \Delta_{\mathscr G}1=0,
  \label{eq:laplacian-annihilates-constants}
\end{equation}
we have
\begin{equation}
  \hat T1
  =
  -\frac{\hbar^2}{2}C(q).
  \label{eq:scalar-term-on-constant}
\end{equation}
Therefore $\hat T1=0$ implies
\begin{equation}
  C(q)=0.
  \label{eq:C-vanishes}
\end{equation}
Hence the unique minimal field-redefinition-covariant ordering of the kinetic Hamiltonian is
\begin{equation}
  \hat T_{\rm min}
  =
  -\frac{\hbar^2}{2}\Delta_{\mathscr G}.
  \label{eq:unique-minimal-lb-final}
\end{equation}
Therefore, for a general tensor field, first performing an invertible
field redefinition $\Phi\mapsto\chi(\Phi)$ and then quantizing gives
the same quantum theory as first quantizing and then rewriting the
result in the $\chi$ variables if and only if the kinetic term is
ordered by the Laplace--Beltrami operator associated with the kinetic
metric, up to non-minimal scalar quantum potentials. Imposing minimality
removes those potentials and gives the unique ordering of Eq.
\eqref{eq:minimal-laplace-beltrami-operator}.

\end{proof}

\begin{corollary}[Minimality removes the curvature-potential ambiguity]
\label{cor:xi-ambiguity}
Assume only the first three conditions of Theorem \ref{thm:ordering},
namely covariance under field redefinitions, naturality, and symmetry
together with conjugation invariance. Then the admissible kinetic
operators form the family
\begin{equation}
  \hat T_C
  =
  -\frac{\hbar^2}{2}
  \left(
    \Delta_{\mathscr G}
    +
    C[\mathscr G]
  \right),
  \label{eq:admissible-kinetic-family}
\end{equation}
where $C[\mathscr G]$ is an arbitrary real scalar built from the
kinetic metric and its derivatives, of which the curvature potential
$C=\xi R[\mathscr G]$ of Eq. \eqref{eq:curvature-scalar-candidate} is
the lowest-order representative, higher curvature invariants entering
with dimensionful couplings
\cite{DeWitt1957DynamicalTheory,HawkingPage1986Ordering}. Every member
of this family is separately covariant, so for each fixed $C$ first
quantizing and then performing an invertible field redefinition agrees
with redefining first and quantizing afterwards. Commutativity of
quantization with field redefinitions therefore holds for every value
of $\xi$ and cannot by itself select one. Each nonzero choice of $C$
also does more than reorder the classical kinetic term, since it adds
an interaction of order $\hbar^2$ that is absent from the classical
Hamiltonian and thereby defines \textit{a different quantum theory}
rather than an alternative ordering of the same one. The selection is
made by the minimality condition of Eq. \eqref{eq:minimality-condition}, which is
also what makes the assignment of a quantum kinetic operator to the
classical kinetic term single valued. Since $\Delta_{\mathscr G}1=0$,
\begin{equation}
  \hat T_C\,1
  =
  -\frac{\hbar^2}{2}\,C[\mathscr G],
  \label{eq:family-on-constants}
\end{equation}
so $\hat T_C\,1=0$ holds precisely when $C$ vanishes identically.
The condition already needed so that quantizing before or after a
relabeling of the configuration variables yields one and the same
quantum theory, with no residual family left over, therefore also sets
$\xi=0$, and it removes every higher curvature scalar in the same
stroke.
\end{corollary}

\begin{remark}[Indefinite kinetic metrics]
\label{rem:indefinite-kinetic-metric}
Nothing in the argument requires $\mathscr G_{AB}$ to be positive
definite. The invariant measure is built from
$\left|\det\mathscr G_{AB}\right|$, the inner product of Eq.
\eqref{eq:induced-inner-product} remains positive definite on
wavefunctions regardless of the signature of $\mathscr G_{AB}$, and the
integrations by parts leading to $V^A=0$ never invert the metric on a
null direction. The operator $\Delta_{\mathscr G}$ is then of
hyperbolic rather than elliptic type but the uniqueness statement is
unaffected. The theorem therefore applies verbatim to indefinite
kinetic metrics and in particular to the DeWitt-type supermetric of
the gravitational configuration space used in Section
\ref{sec:quantization-of-gravity}.
\end{remark}

\begin{remark}[Extension to the full field space]
\label{rem:field-space-extension}
The conclusion extends to the full field theory without the
introduction of any ultraviolet regulator, because the gravitational
kinetic metric of Section \ref{sec:quantization-of-gravity} is
ultralocal, coupling no two distinct points of $\Sigma_\lambda$. For
such a metric the invariant density and the Laplace--Beltrami
operator decompose as
\begin{equation}
  \sqrt{\mathscr M}
  =
  \prod_{x}\sqrt{\mathscr M_x\!\left(q(x)\right)},
  \qquad
  \Delta_{\mathscr G}
  =
  \sum_{x}\Delta_x,
  \label{eq:ultralocal-decomposition}
\end{equation}
where $\Delta_x$ is the Laplace--Beltrami operator of the
single-point block of $\mathscr G_{AB}$, the factors of the density
at all other points cancelling between the $1/\sqrt{\mathscr M}$ and
the $\sqrt{\mathscr M}$ of Eq.
\eqref{eq:configuration-laplacian-definition}. Consider now the
cylinder wavefunctionals, those depending on the field through its
values at finitely many points $x_1,\dots,x_k$. Every block
$\Delta_x$ with $x$ outside this set annihilates such a functional
so the action of $\Delta_{\mathscr G}$ on it reduces to the finite
sum $\Delta_{x_1}+\dots+\Delta_{x_k}$ of finite-dimensional
operators, each falling under the theorem applied to its own block.
Enlarging the set of points leaves this action unchanged and thus the
blocks form a family consistent under inclusion and define
$\Delta_{\mathscr G}$ on the linear span of all cylinder
wavefunctionals with no further choices. The ordering is therefore
fixed on the full field space by the finite-dimensional theorem
alone and at no stage does a divergent coefficient appear that a
regulator would be needed to control.
The inner product of Eq.
\eqref{eq:induced-inner-product} is evaluated in the same blockwise
manner; for two cylinder wavefunctionals depending on the same
points it factorizes into the finite integral over their common
block times a formal product of single-point volumes over the
remaining points, a factor that is the same for every state and
therefore cancels in each normalized matrix element. Fixing this
common factor amounts to choosing the overall normalization of the
functional measure, a choice logically independent of the ordering,
which is fixed block by block whatever normalization is adopted.
\end{remark}

\begin{remark}[Absence of the Tsamis--Woodard obstruction]
\label{rem:tsamis-woodard}
Tsamis and Woodard have argued that in canonical quantum gravity the
factor-ordering problem cannot even be posed before the theory is
regulated \cite{TsamisWoodard1987Regulated}. Since Theorem
\ref{thm:ordering} fixes an ordering with no regulator in sight, we
should explain why their argument, which is correct in its own
setting, does not apply here. In the ADM formulation the objects to
be ordered are the momentum and Hamiltonian constraints, of the
schematic form
\begin{equation}
  \begin{aligned}
    \mathcal H_i(x)
    &=
    -2\,g_{ij}(x)\,p^{jk}{}_{;k}(x),
    \\
    \mathcal H(x)
    &\sim
    \frac{1}{\sqrt{g}}\,
    \mathbb G_{ijkl}[g]\,p^{ij}p^{kl}
    \\
    &\quad+
    \sqrt{g}\,R[g],
  \end{aligned}
  \label{eq:tw-adm-constraints}
\end{equation}
with the precise coefficients of Eqs. (2.8) of Ref.
\cite{TsamisWoodard1987Regulated}, and physical states are defined as
their common kernel, $\hat{\mathcal H}_i\Psi=0$ and $\hat{\mathcal
H}\Psi=0$. For these conditions to admit nonzero solutions the
quantum constraint algebra must close in the Dirac sense,
\begin{equation}
  [\hat\phi_A,\hat\phi_B]
  =
  i\hbar\,\hat C_{AB}{}^{C}\,\hat\phi_C,
  \label{eq:tw-dirac-consistency}
\end{equation}
with the structure operators $\hat C_{AB}{}^{C}$ standing to the left
of the constraints, and the factor-ordering problem of that
formulation is precisely the search for an ordering of Eq.
\eqref{eq:tw-adm-constraints} realizing Eq.
\eqref{eq:tw-dirac-consistency}. Verifying the closure requires
commuting constraints that are nonlinear products of operator-valued
distributions at coincident points and reassembling the result into
constraints through distributional identities such as
$f(x)g(y)\,\delta^3(x-y)=f(x)g(x)\,\delta^3(x-y)$, identities that
hold for smooth test functions and are undefined when $f$ and $g$ are
themselves local operators. Tsamis and Woodard show that applying
such identities in different orders yields mutually inconsistent
results, that the ordering ambiguities of the constraints surface as
undetermined terms proportional to $\delta^3(0)$ and
$[\delta^3(0)]^2$, and that the ordering proposed by Komar dissolves
under this test; regularization must therefore precede ordering
wherever Dirac consistency is the criterion of success
\cite{TsamisWoodard1987Regulated,Komar1979Ordering,FriedmanJack1988Constraints}.
Every link of this chain is forged from the first-class constraints.
In the present theory there is no Hamiltonian constraint and no
momentum constraint, as Section \ref{sec:quantization-of-gravity}
establishes; the only constraints are the projection pair of Eq.
\eqref{eq:qgConstraints}, which are linear in the canonical
variables, carry the invertible bracket matrix of Eq.
\eqref{eq:second-class-matrix} built from the observer structure
alone, and are imposed strongly after passage to the Dirac bracket,
so that they restrict the argument of the wavefunctional instead of
acting as operator equations whose mutual consistency must be
verified. There is consequently no quantum algebra whose closure
could fail, no structure operators whose position matters, and no
product of coincident constraints to define, and the criterion that
makes regularization a logical prerequisite of ordering in the ADM
scheme never arises. The ordering question that does remain, the one
for the kinetic term, is settled by the theorem on each
finite-dimensional block, and Remark
\ref{rem:field-space-extension} supplies the only care that
coincident operator products still require by defining the
functional Laplacian blockwise on cylinder wavefunctionals.
\end{remark}

\noindent For example, let us work with a field $\phi(t)$ whose
Lagrangian is

\begin{equation}
  L_\phi = \frac{1}{2}e^{2\phi}\dot{\phi}^2 - \frac{1}{2}\Omega^2e^{2\phi},
\end{equation}

\noindent whose non-linearities can be removed by the field-redefinition
$\chi \equiv e^{\phi}$, which gives

\begin{equation}
  L_\chi = \frac{1}{2}\dot{\chi}^2 - \frac{1}{2}\Omega^2\chi^2.
\end{equation}

\noindent Their Hamiltonians are

\begin{equation}
  H_\phi = \frac{1}{2}e^{-2\phi}p^2_\phi + \frac{1}{2}\Omega^2e^{2\phi},
  \qquad
  H_\chi = \frac{1}{2}p^2_\chi + \frac{1}{2}\Omega^2\chi^2.
\end{equation}

\noindent Classically, these two are the same Hamiltonians and thus
describe the same dynamics. However, if we try to quantize both we may
find an ordering ambiguity for $H_\phi$, yet no ambiguity for $H_\chi$.
More notably, we would have

\begin{equation}
  \hat{H}_\chi = \frac{-\hbar^2}{2}\partial^2_\chi
  + \frac{1}{2}\Omega^2\hat{\chi}^2,
\end{equation}

\noindent for $\hat{H}_\chi$ but

\begin{equation}
  \hat{H}_\phi = \frac{-\hbar^2}{2}e^{-2\hat{\phi}}
  \left(
    \partial^2_\phi + s\partial_\phi + b
  \right)
  + \frac{1}{2}\Omega^2e^{2\hat{\phi}},
\end{equation}

\noindent and transforming from $\hat{\phi}$ to $\hat{\chi}$ gives

\begin{equation}
  \hat{H}_{\phi\mapsto\chi} = \frac{-\hbar^2}{2}
  \left[
    \partial^2_\chi + \frac{1 + s}{\hat{\chi}}\partial_\chi + \frac{b}{\hat{\chi}^2}
  \right]
  + \frac{1}{2}\Omega^2\hat{\chi}^2,
\end{equation}

\noindent but $\hat{H}_\chi$ and $\hat{H}_{\phi\mapsto\chi}$ should
coincide, which only occurs when $s = -1$ and $b=0$; i.e., by demanding
covariance and the minimality condition $\hat{T}1 = 0$ in
Theorem \ref{thm:ordering}. If we do not impose it the underlying
physics shall change depending on whether we work with $\phi$ or a
redefinition $\chi$, which is difficult to justify physically.
The parameter $b$ survives in one dimension of the scalar freedom
$C(q)$ of Theorem \ref{thm:ordering}. With a single configuration
variable the scalar curvature of the kinetic metric vanishes
identically, so the curvature potential of Eq.
\eqref{eq:curvature-scalar-candidate} cannot be seen in reduced
models and the entire content of minimality is the vanishing of $b$.
On the full gravitational configuration space of Section
\ref{sec:quantization-of-gravity} the curvature is nonzero and
Corollary \ref{cor:xi-ambiguity} shows that the same condition
removes the $\xi R$ freedom there.
 \end{appendices}

\bibliography{references}

\end{document}